\documentclass[journal]{IEEEtran}

\pdfoutput=1

\usepackage{amsmath,amssymb,amsfonts}
\usepackage{amsthm}
\usepackage{physics}
\usepackage{braket}
\usepackage{bbm}
\usepackage{mathtools}
\usepackage{microtype}
\usepackage[acronym]{glossaries}

\usepackage{graphicx}
\usepackage[caption=false,font=footnotesize]{subfig}

\usepackage{multirow}
\usepackage{booktabs}
\usepackage{array}

\usepackage{xcolor}
\usepackage[english]{babel}
\usepackage[hypertexnames=false,colorlinks=true,linkcolor=blue,citecolor=blue,urlcolor=blue]{hyperref}

\newtheorem{theorem}{Theorem}

\newtheorem{proposition}[theorem]{Proposition}

\newacronym{sbs}{SBS}{Stimulated Brillouin Scattering}
\newacronym{qkd}{QKD}{Quantum Key Distribution}

\allowdisplaybreaks

\begin{document}

\title{Quantum Information-Theoretic Analysis of Optical Fiber Sensing via Stimulated Brillouin Scattering
\thanks{The research is part of the Munich Quantum Valley, which is supported by the Bavarian state government with funds from the Hightech Agenda Bayern Plus. This work was financed by the DFG via grant NO 1129/2-1 and by the Federal Ministry of Education and Research of Germany in the Q-STARS project, grant number 16KIS2604, as well as via grants 16KISQ093, 16KISQ039 and 16KISQ077. The generous support of the state of Bavaria via the 6GQT project is greatly appreciated.}}

\author{
Kiran Adhikari, Janis N\"otzel (Member, IEEE)\\
\textit{Emmy-Noether Group Theoretical Quantum Systems Design},\\
\textit{Technical University of Munich, Munich, Germany}}

\maketitle

\begin{abstract}
Stimulated Brillouin scattering (SBS) is a powerful mechanism for distributed sensing in optical fibers, as the Brillouin response is highly sensitive to local material parameters. Assuming the input probe is a quantum state, we develop a quantum-information-theoretic framework for SBS-based fiber sensing and intrusion detection. Starting from a linearized opto-acoustic Hamiltonian, we derive an effective input-output relation for a short fiber segment. We then model each fiber segment as a single-mode bosonic Gaussian channel and the entire fiber as a cascade of quantum channels. We then compute the quantum Fisher information for general SBS-dependent parameters and identify optimal measurement settings to achieve the quantum limit, comparing them with receiver-constrained strategies such as heterodyne detection and photon counting. In addition, we propose an adaptive approach that achieves first-order asymptotic optimality and provides a novel second-order correction to address preliminary miscalibration. As an application, we apply the framework to the fiber-tapping detection problem. 
\end{abstract}


\section{Introduction}

Optical fiber networks are a critical component of modern communication infrastructure and are often assumed to provide strong physical-layer security as the optical field is confined within a dielectric waveguide. In practice, however, the practical attacks, such as evanescent coupling or fiber bending, allow 
 an adversary to extract optical power with minimal disturbance, making reliable detection of eavesdropping a fundamental challenge \cite{cca879e457e44e34b704c801c6aaa6e8, Shaneman2004OpticalNS}. While prior work has studied eavesdropper detection and joint communication–sensing strategies in optical and quantum channels \cite{eavesdropperDetection, munarvallespirJCAS, boche2016, boche2017}, a quantitative understanding of the fundamental limits of detection in realistic fiber systems remains incomplete yet desirable when planning for the adoption of modern quantum technology into future networks \cite{integratingQuantumSimulation}. A central problem is therefore to determine the limits imposed by quantum mechanics on the detection of local perturbations resulting from fiber tapping attacks.

In optical fibers, Stimulated Brillouin scattering (SBS) provides a natural physical mechanism to model such perturbations \cite{Wolff:21, Lu2019DistributedOF}. SBS is a nonlinear optoacoustic interaction between a pump field, a counter-propagating Stokes field, and an acoustic phonon. Although SBS is often undesirable in optical communications because it depletes the pump and generates backward scattering, it is highly useful for sensing, as the Brillouin resonance depends on local material properties such as strain and temperature. This enables spatially resolved sensing along the fiber and has been extensively exploited in techniques such as Brillouin optical correlation domain analysis (BOCDA) \cite{e83-c_3_405, Lu2019DistributedOF, senaOTDR}.

\cite{popp2023eavesdropperlocalizationquantumclassical} demonstrated experimentally that \gls{sbs} techniques can be used to localize the eavesdropping in a quantum channel. This can be highly attractive in the \gls{qkd} setting, as \gls{qkd} already provides information on whether an eavesdropper is present in the channel \cite{Gisin_2002}. \cite{Li_2022} demonstrated that quantum techniques can indeed enhance \gls{sbs} estimation. Our work can be thought of as developing a quantum information-theoretic framework \cite{adhikari2026fundamentallimitseavesdropperdetection,noetzel2026quantumlimitedblindsourceseparation} for \gls{sbs}-based intrusion detection in optical fibers. From the theoretical perspective, this is quite attractive as SBS can be modeled using quantum mechanics, in which the Stokes field and the acoustic phonon become entangled as it propagates \cite{zhu2024optoacousticentanglementcontinuousbrillouinactive, rachbauerHowFindOptimal2024,zuoQuantumEmpoweredFiberSensing2025,tanEnhancedSensitivityPolarization2023,tianRobustPhotonEntanglement2013}. The entire optical fiber can be modeled as a cascade bosonic Gaussian channel, and the input Stokes field can be an arbitrary quantum input. This formalism allows comparison between practical receivers, such as heterodyne detection or photon counting, and the best measurement allowed by quantum mechanics. 

Our main contributions are as follows:
\begin{itemize}
    \item In section \ref{sec:quantum_mechannical}, we derive a quantum Hamiltonian model for the SBS interaction. Dividing the entire fiber, we derive for each segment $i$ an input-output relation of the Stokes field $a_s$. We show that, in the appropriate limit, this reproduces the classical SBS power-gain equation.  
    \item The input-output relation connects the microscopic SBS model to experimentally accessible fiber parameters. In section \ref{sec:fiber_cascade}, we show that this allows us to model each fiber segment as a single-mode Gaussian quantum channel on the Stokes field. The entire fiber is then a cascade of such Gaussian channels. 
    \item Assuming the coherent-state input, the output state of such a channel is a displaced thermal state. In section \ref{sec:qfi_sbs_parameter_sensing}, we compute the quantum Fisher information (QFI) of the displaced thermal output state for a general parameter $\theta$. We also derive the corresponding symmetric logarithmic derivative and identify the optimal measurement structure.
    \item In section \ref{subsubsec:derivatives_sbs_output_state},  we apply the QFI framework to Brillouin-frequency and temperature sensing. For Brillouin-frequency estimation, the largest sensitivity occurs near the steepest part of the Brillouin line, where the response is most sensitive to changes in $\Omega_B$. For temperature sensing, we consider a model in which $\Omega_B$, $\Gamma_B$, the gain coefficient, and the thermal phonon occupation can depend on the temperature.
    \item In section \ref{sec:finite_sample_adaptive}, we show how the optimal measurement of the probes can be implemented via an adaptive method. We also show that the adaptive protocol achieves first-order asymptotic optimality and provide a second-order correction to account for preliminary miscalibration. 
    \item In section \ref{sec:application1}, we present an application in the fiber tapping problem. We formulate the intrusion-detection problem as a local parameter estimation and quantify it using QFI and receiver-constrained Fisher information.
\end{itemize}

\section{Notation}
To shorten equations, we define for every $x\in[0,1]$ the quantity $x':=1-x$. The underlying Hilbert space in our analysis is the Fock space. Notation is borrowed from \cite{bookserfini} and \cite{HolevoBook}, so that $\hat a^\dagger$ and $\hat a$ are creation- and annihilation operators, and $S_N(\alpha):=D(\alpha)S_N(0)D(\alpha)^\dagger$ denotes the displaced thermal state with mean thermal photon number $N\geq0$ and displacement $\alpha\in\mathbb C$.

\section{Quantum Mechanical Modeling of SBS}
\label{sec:quantum_mechannical}

Stimulated Brillouin scattering (SBS)  \cite{Wolff:21, Lu2019DistributedOF,10.1007/3-540-46629-0_9}, as depicted in Figure \ref{fig:brillouin}, is a third-order nonlinear optoacoustic interaction between a strong optical pump, a counter-propagating red-shifted Stokes field, and an acoustic phonon mode of the medium. SBS is a good probe of distributed sensing, as the Brillouin resonance depends on local material properties such as strain and temperature. Strain $p$ and temperature $T$ modify the effective refractive index $n_{\mathrm{eff}}(p,T)$ and acoustic velocity $v_{\mathrm{ac}}(p,T)$ which then affects Brillouin frequency $\Omega_B$ as (Eq. 1 of \cite{thevenazBrillouinDistributedTimedomain2010})
\begin{equation}
    \frac{\Omega_B}{2\pi}=\frac{2}{\lambda}\,n_{\mathrm{eff}}(p,T)\,v_{\mathrm{ac}}(p,T).
\end{equation}

Physically, the pump drives a periodic density modulation via electrostriction, producing a traveling acoustic wave that acts as a moving refractive-index grating, scattering pump photons into Stokes photons. Energy conservation gives (Eq. 1 of \cite{Wolff:21}) :
\begin{equation}
    \omega_{\mathrm{s}}=\omega_{\mathrm{p}}-\Omega
\end{equation}
where $\Omega$ is the pump-Stokes frequency difference. The SBS interaction is resonantly enhanced when $\Omega = \Omega_B$, in which case $\omega_s = \omega_{\mathrm{p}}-\Omega_B$.

\begin{figure}[!t]
    \centering
    \includegraphics[width=0.8\linewidth]{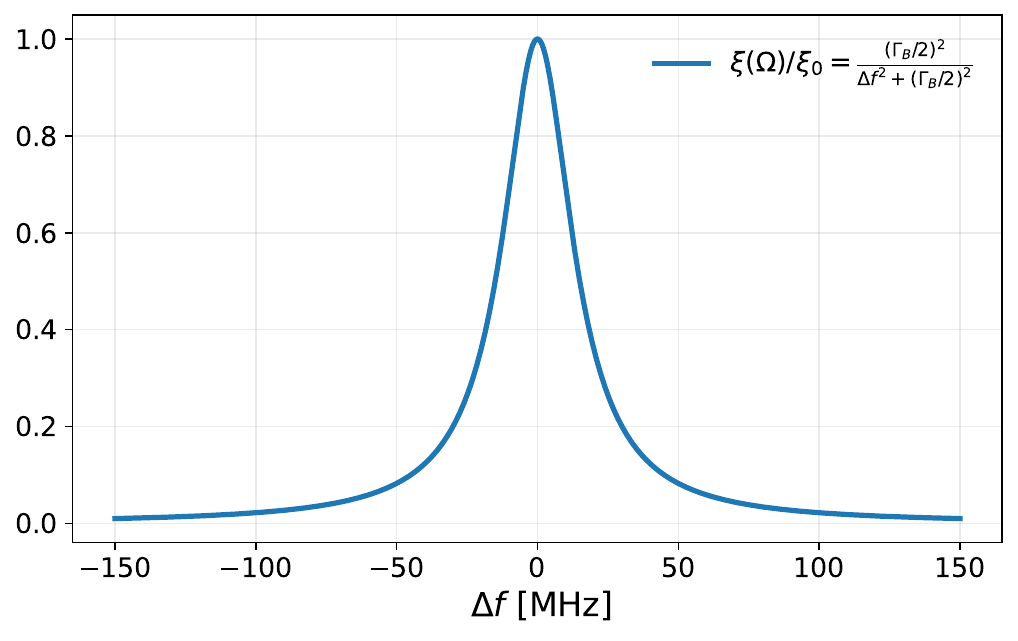}
    \caption{Normalized Brillouin gain spectrum for a standard silica single-mode fiber near $1550\,\mathrm{nm}$ with parameters summarized in Table~\ref{tab:physical_sbs_parameters}. 
At zero detuning, the gain is maximum. It then decreases rapidly with a linewidth set by $\Gamma_B$. }
    \label{fig:Brillouin_gain_spectrum}
\end{figure}

Because the acoustic mode is damped, SBS exhibits a narrow Lorentzian gain 
 spectrum,
\begin{equation}
\xi(\Omega)=\xi_0\,\frac{(\Gamma_B/2)^2}{(\Omega_B-\Omega)^2+(\Gamma_B/2)^2},   
\end{equation}
where \(\Omega_B\) is the resonance frequency, \(\Gamma_B\) is the Brillouin linewidth, and \(\xi_0 = \xi(\Omega_B)\) is the peak Brillouin gain coefficient expressed in terms of materials properties as (Eq. 2 of \cite{polym3020886})
\begin{equation}
\label{eq:xi_0}
    \xi_0=\frac{\gamma_e^2\omega^2}{n\,v\,c^3\,\rho_0 A_{\rm eff} \Gamma_B}
=\frac{4\pi^2 n^7 p_{12}^2}{c\,\lambda_0^2\,v\,\rho_0 A_{\rm eff} \Gamma_B},
\end{equation}
where \(\gamma_e\) is the electrostrictive constant, \(\omega\) the optical frequency, \(n\) the refractive index of the bulk material, \(v\) the acoustic bulk velocity, \(\rho_0\) the density of the bulk material, $A_{\rm eff}$ is the effective area and \(p_{12}\) the Pockels coefficient. Accordingly, the Stokes power grows as  \cite{10.1007/3-540-46629-0_9} 
\begin{equation}
    \frac{dP_s}{dz}=\xi(\Omega)P_pP_s.
\end{equation}
In Figure~\ref{fig:Brillouin_gain_spectrum}, we have plotted the Lorentzian SBS gain response for a standard silica
single-mode fiber near $1550\,\mathrm{nm}$ with the physical parameters summarized
in Table~\ref{tab:physical_sbs_parameters}.

\begin{table*}[!t]
\centering
\footnotesize
\begin{tabular}{llll}
\hline
Quantity & Symbol & Value used & Notes \\
\hline
Optical wavelength & $\lambda$ & $1550\,\mathrm{nm}$ & Telecom C-band wavelength \\
Fiber attenuation & $\alpha_{\mathrm{dB}}$ & $0.2\,\mathrm{dB/km}$ & Standard SMF-scale loss \\
Brillouin frequency shift & $\Omega_B/2\pi$ & $10.8\,\mathrm{GHz}$ & Silica fiber near $1550\,\mathrm{nm}$ \\
Brillouin linewidth & $\Gamma_B/2\pi$ & $30\,\mathrm{MHz}$ & Room-temperature SBS linewidth \\
Brillouin gain coefficient & $g_{B,0}$ & $2.5\times10^{-11}\,\mathrm{m/W}$ & Peak material/fiber gain coefficient \\
Effective mode area & $A_{\mathrm{eff}}$ & $80\,\mu\mathrm{m}^2$ & Standard single-mode fiber scale \\
Pump power & $P_p$ & $10$--$100\,\mathrm{mW}$ & Classical pump power \\
Temperature & $T$ & $300\,\mathrm{K}$ & Room temperature \\
\hline
\end{tabular}
\caption{Physical parameters for a standard silica single-mode fiber near $1550\,\mathrm{nm}$ \cite{Wolff:21, poppBrillouin}.}
\label{tab:physical_sbs_parameters}
\end{table*}

\subsection{Quantum Hamiltonian}
We model SBS as an interaction of three bosonic fields acting on the 
bosonic Fock space 
$\mathcal{F}_{p}\otimes\mathcal{F}_{s}\otimes\mathcal{F}_{b}$:
an strong optical pump field $ a_p(z,t)$, a counter-propagating Stokes (or probe)
optical field $ a_s(z,t)$, and an acoustic phonon field $\hat b(z,t)$ \cite{zhu2024optoacousticentanglementcontinuousbrillouinactive, rakich2016quantumtheorycontinuumoptomechanics, sipeHamiltonianTreatmentStimulated2016}. Then, the Hamiltonian is:

\begin{align}
\label{eq:H_1}
H_{\mathrm{int}}
&=
\int dz\,\hbar g_0(z)
\Bigl(
 a_p(z,t)\,a_s^\dagger(z,t)\,b^\dagger(z,t)
\nonumber\\
&\hspace{2.7cm}
+ a_p^\dagger(z,t)\,a_s(z,t)\,b(z,t)
\Bigr).
\end{align}
where $g_0(z)$ is an effective position-dependent opto-acoustic coupling coefficient. 
It is related to the Brillouin gain coefficient $\xi_0$ (Eq. \ref{eq:xi_0})  after the appropriate classical-limit matching carried out in section \ref{subsec:classical_equivalence}. In Hamiltonian \ref{eq:H_1}, the first term, $a_p\,a_s^\dagger b^\dagger,$ indicates the annihilation of a pump photon together with the creation of a Stokes photon and a phonon: 
\begin{equation}
\text{pump photon} \;\rightarrow\; \text{Stokes photon} + \text{phonon}.
\end{equation}
The second term is its Hermitian conjugate and describes the reverse conversion process. Together, this ensures that the Hamiltonian is hermitian and that the evolution is unitary. 

We assume that the optical pump is so strong that it is approximated by a classical coherent field $a_p(z,t) \approx \alpha_p(z,t)e^{-\mathrm{i}\omega_p t}$ with \(|\alpha_p| \gg 1\). The envelope $\alpha_p(z,t) $is allowed to depend on position and time because the classical pump may be attenuated during propagation and may have a finite temporal pulse shape. In this sense, the pump behaves as a classical drive for interaction between the Stokes and the phonon field.

Under this classical pump approximation \cite{sipeHamiltonianTreatmentStimulated2016}, we define $G(z,t) = g_0(z)\alpha_p(z,t)$ as the pump-enhanced coupling. Now, the Hamiltonian has a simple form of a two-mode squeezing Hamiltonian \cite{sipeHamiltonianTreatmentStimulated2016}:
\begin{align}
H_{\mathrm{lin}}
&=
\int dz\,\hbar
\Bigl(
G(z,t)\,a_s^\dagger(z,t)b^\dagger(z,t)
\nonumber\\
&\hspace{2.5cm}
+G^*(z,t)\,a_s(z,t)b(z,t)
\Bigr).
\end{align}
This means that the Stokes photon gets entangled with the phonon as it propagates along the fiber.

 \begin{figure}[!t]
    \centering
    \includegraphics[width=0.9\linewidth]{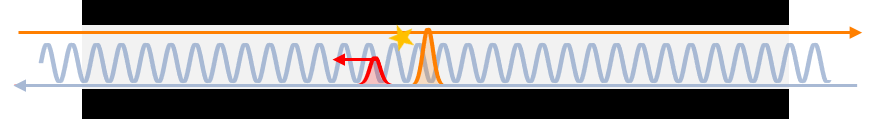}
    \caption{Light blue: CW light, depicting the probe, operating at frequency $\omega_s$. Orange: short, bright pulse, depicting the pump, operating at frequency $\omega_p$. The backscattered stokes photons are at the frequency $\omega_s$ of the probe. }
    \label{fig:brillouin}
\end{figure}

\subsection{Small-segment input--output map}
\label{sec:derivation_of_input_output}
In this section, we will consider how the Stokes photon changes as it passes through a small segment of the fiber. For this, we consider a fiber of length $L$, divided into $N$ segments of equal length $\delta z = L/N$. We can match the length of each segment, $\delta z$, to the spatial resolution of the fiber sensing. Then it would correspond to the pulse of duration. For example, nanosecond pulses would correspond to meter or sub-meter-scale resolution segments.

Let $z_i$ denote the center of the $i$th segment. As the Stokes field counterpropagates relative to the pump, we introduce a coordinate direction $s =-z$ along the Stokes field. Over $\delta z$,  we treat the following parameters as constant:
\begin{equation}
G(z)\approx G_i,\qquad
\Omega_B(z)\approx \Omega_{B,i},\qquad
\alpha_s(z)\approx \alpha_{s,i}. \nonumber
\end{equation}
The linearized \gls{sbs} interaction on the segment is
\begin{equation}
H_{\mathrm{lin},i}
=
\hbar\left(
G_i e^{\mathrm{i}\Delta k_i z} a_s^\dagger b^\dagger
+
G_i^* e^{-\mathrm{i}\Delta k_i z} a_s b
\right). \nonumber
\end{equation}
For a sufficiently short segment, the phase mismatch satisfies
$\Delta k_i\,\delta z \ll 1$, so the phase can be absorbed into the coupling
$\widetilde G_i = G_i e^{\mathrm{i}\Delta k_i z_i}$. The Hamiltonian then becomes
\begin{equation}
\label{eq:linearH2}
H_{\mathrm{lin},i}
=
\hbar\left(
\widetilde G_i a_s^\dagger b^\dagger
+
\widetilde G_i^* a_s b
\right).
\end{equation}
From the Hamiltonian \eqref{eq:linearH2}, we obtain the explicit input--output
relation of the Stokes photon mode in Proposition~\ref{prop:input-ouput}. The ingredients entering Proposition~ \ref{prop:input-ouput} are standard in quantum SBS theory \cite{zhu2024optoacousticentanglementcontinuousbrillouinactive, rakich2016quantumtheorycontinuumoptomechanics,sipeHamiltonianTreatmentStimulated2016}. The new step is the organization into an explicit input-output map of the Stokes mode which allows us to model each fiber segment as a bosonic Gaussian channel. This forms the starting point for the Quantum Fisher Information (QFI) and intrusion-detection analysis developed in the rest of the work.

\begin{proposition}
\label{prop:input-ouput}
The output Stokes operator satisfies the input--output relation
\begin{equation}
\label{eq:input_output}
a_{s,\mathrm{out}}(\Omega)
=
\mu_i(\Omega)\, a_{s,\mathrm{in}}(\Omega)
+
\sqrt{\eta'_i}\, v_i(\Omega)
+
\nu_i(\Omega)\, b_{i,\mathrm{in}}^\dagger(\Omega),
\end{equation}
where $\eta'_i=1-\eta_i$,
\begin{equation}
\begin{aligned}
\mu_i(\Omega)&=\sqrt{\eta_i}+\kappa_i(\Omega),\\
\eta_i&=e^{-\alpha_{s,i}\delta z},\\
\kappa_i(\Omega)&=|\widetilde G_i|^2\chi_i(\Omega)\,\delta z,\\
\nu_i(\Omega)&=-\mathrm{i}\widetilde G_i
\sqrt{\Gamma\delta z}\,\chi_i(\Omega).
\end{aligned}
\end{equation}
The local Brillouin susceptibility is
\begin{equation}
\chi_i(\Omega)
=
\frac{1}{\Gamma/2+\mathrm{i}(\Omega_{B,i}-\Omega)}.
\label{eq:chi_def}
\end{equation}
Here, $\kappa_i(\Omega)$ is the frequency-dependent SBS-induced gain/loss
coefficient, $\nu_i(\Omega)$ describes the added phonon noise coupled into the
Stokes mode, $\chi_i(\Omega)$ is the local Brillouin susceptibility,
$\sqrt{\eta_i}$ describes pure attenuation over a length $\delta z$, $v_i(\Omega)$
is the canonical vacuum mode associated with optical loss, and
$b_{i,\mathrm{in}}(\Omega)$ is the normalized segment phonon bath mode.
\end{proposition}

Note that the input--output relation Eq.~\eqref{eq:input_output} does not mean
that three independent input Hilbert spaces are mapped to a single output Hilbert
space. Rather, it is the Heisenberg-picture evolution of the Stokes mode under a
joint unitary interaction with environmental degrees of freedom. Such expressions
are standard in quantum optics. For example, a lossy channel can be modeled as a
beam splitter with transmissivity $\eta\in[0,1]$, which mixes the signal mode
$a_{\mathrm{in}}$ with an environmental vacuum mode $v$ according to $a_{\mathrm{out}}
=
\sqrt{\eta}\,a_{\mathrm{in}}
+
\sqrt{1-\eta}\,v.$ As a consistency check, if the SBS interaction is absent, $\widetilde G_i=0$,
then $\kappa_i(\Omega)=0$, $\nu_i(\Omega)=0$, and
$\mu_i(\Omega)=\sqrt{\eta_i}$. Hence Eq.~\eqref{eq:input_output} reduces to
\begin{equation}
a_{s,\mathrm{out}}(\Omega)
=
\sqrt{\eta_i}\,a_{s,\mathrm{in}}(\Omega)
+
\sqrt{1-\eta_i}\,v_i(\Omega).
\end{equation}
 The detailed analysis of the beam splitter input-output map is in Appendix \ref{sec:beam-splitter}. This can serve as a warm-up example for more complicated SBS analysis.

\begin{proof}[Proof of Proposition \ref{prop:input-ouput}]

 We assume the bosonic commutation relations \([a_s,a_s^\dagger]=1\), \([b,b^\dagger]=1\), and \([a_s,b]=[a_s,b^\dagger]=0\). Using the Heisenberg equation of motion \(\dot O = \frac{i}{\hbar}[H,O]\), we obtain:
\begin{align}
\dot a_s
&=\frac{i}{\hbar}[H_{\mathrm{lin},i},a_s] \nonumber\\
&=\frac{i}{\hbar}
\left[
\hbar\left(\widetilde G_i a_s^\dagger b^\dagger
+\widetilde G_i^* a_s b\right),
 a_s
\right] \nonumber\\
&= -\mathrm{i}\widetilde G_i b^\dagger .
\end{align}
Next, for \(b^\dagger\), we obtain
\begin{equation}
\dot b^\dagger
=
\frac{i}{\hbar}[H_{\mathrm{lin},i},b^\dagger]
=
i \widetilde G_i [a_s^\dagger b^\dagger,b^\dagger]
+
i \widetilde G_i^* [a_s b,b^\dagger]
=  i \widetilde G_i^* a_s .
\end{equation}
Similarly, \(\dot b = -i \widetilde G_i a_s^\dagger\). 

These equations describe the coherent part. We should also add damping and Langevin noise, which gives the evolution equation for slowly varying Stokes envelope as (Eq. 72 of \cite{sipeHamiltonianTreatmentStimulated2016}):
\begin{equation}
\partial_s a_s(s,t)
=
-\frac{\alpha_{s,i}}{2}a_s(s,t)
-i\widetilde G_i\,b^\dagger(s,t)
+\sqrt{\alpha_{s,i}}\,v_i(s,t),
\label{eq:as_s_equation_time}
\end{equation}
where \(\alpha_{s,i}\) is the optical power attenuation coefficient in the segment and \(v_i(s,t)\) is the corresponding optical vacuum Langevin field.  The phonon mode is modeled as a position-dependent local damped oscillator (Eq. 72 of \cite{sipeHamiltonianTreatmentStimulated2016}):
\begin{equation}
\partial_t b(s,t)
=
-\qty(\frac{\Gamma}{2}+i\Omega_{B,i})b(s,t)
-i\widetilde G_i\,a_s^\dagger(s,t)
+\sqrt{\Gamma}\,b_{\mathrm{in}}(s,t),
\label{eq:phonon_langevin}
\end{equation}
where \(\Gamma\) is the phonon damping rate, \(\Omega_{B,i}\) is the local Brillouin resonance frequency, and \(b_{\mathrm{in}}(s,t)\) describes the thermal Langevin noise of the phonon bath. Its nonvanishing correlations are
\begin{align}
\langle b_{\mathrm{in}}^\dagger(s,t)b_{\mathrm{in}}(s',t')\rangle
&=
n_{\mathrm{th},i}\,\delta(s-s')\delta(t-t'),\nonumber \\
\langle b_{\mathrm{in}}(s,t)b_{\mathrm{in}}^\dagger(s',t')\rangle
&=
\qty(n_{\mathrm{th},i}+1)\,\delta(s-s')\delta(t-t'), \nonumber
\end{align}
where
\begin{equation}
    n_{\mathrm{th},i}
=
\frac{1}{\exp\qty(\hbar\Omega_{B,i}/k_B T_m)-1}
\end{equation}
denotes the thermal occupation of the acoustic field at temperature \(T_m\).

Using the Fourier convention \(o(s,t)=\int\frac{d\Omega}{2\pi}\,e^{-i\Omega t}o(s,\Omega)\), Eq.~\eqref{eq:phonon_langevin} becomes
\begin{equation}
\begin{aligned}
-\mathrm{i}\Omega\,b(s,\Omega)
&=
-\left(\frac{\Gamma}{2}+\mathrm{i}\Omega_{B,i}\right)b(s,\Omega)
\\
&\quad
-\mathrm{i}\widetilde G_i\,a_s^\dagger(s,\Omega)
+\sqrt{\Gamma}\,b_{\mathrm{in}}(s,\Omega).
\end{aligned}
\end{equation}
Rearranging the terms gives,
\begin{equation}
b(s,\Omega)
=
\chi_i(\Omega)
\qty[-i\widetilde G_i\,a_s^\dagger(s,\Omega)
+
\sqrt{\Gamma}\,b_{\mathrm{in}}(s,\Omega)],
\label{eq:b_solution}
\end{equation}
where we identify $\chi_i(\Omega)$ as the local Brillouin susceptibility written explicitly as:
\begin{equation}
    \chi_i(\Omega)=\frac{1}{\Gamma/2+i(\Omega_{B,i}-\Omega)}
\end{equation}
 Similarly, the adjoint of $b(s,\Omega)$ is:
\begin{equation}
b^\dagger(s,\Omega)
=
\chi_i^*(\Omega)
\qty[
i\widetilde G_i^*\,a_s(s,\Omega)
+
\sqrt{\Gamma}\,b_{\mathrm{in}}^\dagger(s,\Omega)
].
\label{eq:bdag_solution}
\end{equation}
In frequency space, Eq.~\eqref{eq:as_s_equation_time} becomes
\begin{align}
&\partial_s a_s(s,\Omega)\nonumber\\
&= -\frac{\alpha_{s,i}}{2}a_s(s,\Omega)
-\mathrm{i}\widetilde G_i\,b^\dagger(s,\Omega)\nonumber\\
&\quad+\sqrt{\alpha_{s,i}}\,v_i(s,\Omega)\nonumber\\
&= -\frac{\alpha_{s,i}}{2}a_s(s,\Omega)\nonumber\\
&\quad-\mathrm{i}\widetilde G_i\chi_i(\Omega)
\left[\mathrm{i}\widetilde G_i^*\,a_s(s,\Omega)
+\sqrt{\Gamma}\,b_{\mathrm{in}}^\dagger(s,\Omega)\right]\nonumber\\
&\quad+\sqrt{\alpha_{s,i}}\,v_i(s,\Omega)\nonumber\\
&=\left[-\frac{\alpha_{s,i}}{2}
+|\widetilde G_i|^2\chi_i(\Omega)\right]a_s(s,\Omega)\nonumber\\
&\quad-\mathrm{i}\widetilde G_i\sqrt{\Gamma}\,\chi_i(\Omega)
b_{\mathrm{in}}^\dagger(s,\Omega)\nonumber\\
&\quad+\sqrt{\alpha_{s,i}}\,v_i(s,\Omega).
\label{eq:as_s_freq}
\end{align}

Integrating Eq.~\eqref{eq:as_s_freq} from \(s_i\) to \(s_i+\delta z\) gives

\begin{align}
a_s&(s_i+\delta z,\Omega)
=
a_s(s_i,\Omega)
+
\delta z\,K_i(\Omega)\,
a_s(s_i,\Omega) + 
\nonumber\\
&\quad
L_i(\Omega) +\sqrt{\alpha_{s,i}}
\int_{s_i}^{s_i+\delta z}
ds\,
v_i(s,\Omega)
\label{eq:short_segment_integrated}
\end{align}
where, $K_i(\Omega)
=
-\frac{\alpha_{s,i}}{2}
+
|\widetilde G_i|^2\chi_i(\Omega)$ and $L_i(\Omega) = -\mathrm{i}\widetilde G_i\sqrt{\Gamma}\,
\chi_i(\Omega)
\int_{s_i}^{s_i+\delta z}
ds\,
b_{\mathrm{in}}^\dagger(s,\Omega)$. We then define normalized bath modes:
\begin{equation}
\begin{aligned}
b_{i,\mathrm{in}}(\Omega)
&=
\frac{1}{\sqrt{\delta z}}
\int_{s_i}^{s_i+\delta z} ds\,
 b_{\mathrm{in}}(s,\Omega),\\
v_i(\Omega)
&=
\frac{1}{\sqrt{\delta z}}
\int_{s_i}^{s_i+\delta z} ds\,v_i(s,\Omega).
\end{aligned}
\label{eq:normalized_segment_modes}
\end{equation}
This normalization is essential as a spatially white noise field integrated over a segment should scale as \(\sqrt{\delta z}\).  Eq.~\eqref{eq:normalized_segment_modes} becomes
\begin{align}
a_s(s_i+ &\delta z,\Omega)
=
a_s(s_i,\Omega)
+
\delta z
K_i(\Omega)a_s(s_i,\Omega)
\nonumber\\
&\quad
-\mathrm{i}\widetilde G_i\sqrt{\Gamma\delta z}\,
\chi_i(\Omega)
 b_{i,\mathrm{in}}^\dagger(\Omega) +
\sqrt{\alpha_{s,i}\delta z}\,v_i(\Omega).
\label{eq:short_segment_integrated_segment_modes}
\end{align}
We now define the input and output operators of the Stokes field as
\begin{align}
    a_{s,\mathrm{in}}(\Omega)&=a_s(s_i,\Omega) \\
    a_{s,\mathrm{out}}(\Omega)&=a_s(s_i+\delta z,\Omega)
\end{align}
Then, Eq. \ref{eq:short_segment_integrated_segment_modes} becomes:
\begin{align}
a_{s,\mathrm{out}}(\Omega)
&=
\left[
1-
\frac{\alpha_{s,i}\delta z}{2}
+
|\widetilde G_i|^2\chi_i(\Omega)\delta z
\right]a_{s,\mathrm{in}}(\Omega)
\nonumber\\
&\quad
-\mathrm{i}\widetilde G_i\sqrt{\Gamma\delta z}\,
\chi_i(\Omega)b_{i,\mathrm{in}}^\dagger(\Omega)
+
\sqrt{\alpha_{s,i}\delta z}\,v_i(\Omega).
\label{eq:segment_before_loss_rewrite}
\end{align}
We can further simplify using the definition of amplitude transmissivity $\sqrt{\eta_i}=e^{-\alpha_{s,i}\delta z/2}$, which describes pure attenuation over a length $\delta z$. To first order in $\delta z$, $\sqrt{\eta_i}\approx 1-\frac{\alpha_{s,i}\delta z}{2}$ and $\sqrt{1-\eta_i}\approx \sqrt{\alpha_{s,i}\delta z}$.  Eq.~\eqref{eq:segment_before_loss_rewrite} becomes
\begin{equation}
\begin{aligned}
a_{s,\mathrm{out}}(\Omega)
&=
\left[
\sqrt{\eta_i}
+
|\widetilde G_i|^2\chi_i(\Omega)\delta z
\right]a_{s,\mathrm{in}}(\Omega)
\\
&\quad
+
\sqrt{1-\eta_i}\,v_i(\Omega)
-
\mathrm{i}\widetilde G_i\sqrt{\Gamma\delta z}\,
\chi_i(\Omega)b_{i,\mathrm{in}}^\dagger(\Omega).
\end{aligned}
\label{eq:segment_after_loss_rewrite}
\end{equation}
Defining
\begin{equation}
  \kappa_i(\Omega)=|\widetilde G_i|^2\chi_i(\Omega)\delta z, \quad  \nu_i(\Omega)=-i\widetilde G_i\sqrt{\Gamma\delta z} \,\chi_i(\Omega)
\end{equation}
where, $\mu_i(\Omega)=\sqrt{\eta_i}+\kappa_i(\Omega)$, we finally obtain
\begin{equation}
a_{s,\mathrm{out}}(\Omega)
=
\mu_i(\Omega)\,a_{s,\mathrm{in}}(\Omega)
+
\sqrt{1-\eta_i}\,v_i(\Omega)
+
\nu_i(\Omega)\,b_{i,\mathrm{in}}^\dagger(\Omega),
\end{equation}
which is the desired small-segment input--output relation.

\end{proof}

To the first order in $\delta z$, we have 
\begin{equation}
[a_{s,\mathrm{out}},a^\dagger_{s,\mathrm{out}}]
=
|\mu_i|^2+\eta_i'-|\nu_i|^2  \simeq
1+2\mathrm{Re}\,\kappa_i-|\nu_i|^2 .
\end{equation}
The susceptibility satisfies $2\mathrm{Re}\,\chi_i(\Omega)
=
\Gamma|\chi_i(\Omega)|^2 ,$
and therefore $|\nu_i|^2
=
2\mathrm{Re}\,\kappa_i .$
Hence,  $[a_{s,\mathrm{out}},a^\dagger_{s,\mathrm{out}}]
\simeq
1$, and the Eq.~\eqref{eq:input_output} satisfies the bosonic commutation relation. Physically, this means the phonon-noise term compensates the SBS gain contribution and ensures that the output Stokes operator remains a valid bosonic mode.

Therefore, each segment implements a bosonic Gaussian channel, consisting of a
deterministic linear transformation, vacuum noise associated with optical
attenuation, and thermal noise from the phonon bath. 

\begin{table*}[!t]
\centering
\footnotesize
\begin{tabular}{llll}
\hline
Quantity & Symbol & Expression & Value used \\
\hline
Linear attenuation coefficient
& $\alpha_s$
& $\dfrac{\ln 10}{10}\dfrac{\alpha_{\mathrm{dB}}}{1000}$
& $4.61\times10^{-5}\,\mathrm{m}^{-1}$ \\[1.0ex]

Angular Brillouin frequency
& $\Omega_B$
& $2\pi\times10.8\,\mathrm{GHz}$
& $6.79\times10^{10}\,\mathrm{rad/s}$ \\[1.0ex]

Angular Brillouin linewidth
& $\Gamma_B$
& $2\pi\times30\,\mathrm{MHz}$
& $1.88\times10^{8}\,\mathrm{rad/s}$ \\[1.0ex]

Effective one-dimensional SBS gain
& $\xi_0$
& $\dfrac{g_{B,0}}{A_{\mathrm{eff}}}$
& $0.31\,\mathrm{W}^{-1}\mathrm{m}^{-1}$ \\[1.0ex]

Thermal phonon occupation
& $n_{\mathrm{th}}$
& $\dfrac{1}{\exp(\hbar\Omega_B/k_BT)-1}$
& $\simeq 5.8\times10^2$ \\[1.0ex]

Peak small-signal gain rate
& $\xi_0P_p$
& $\xi_0P_p$
& $3.1\times10^{-3}$--$3.1\times10^{-2}\,\mathrm{m}^{-1}$ \\
\hline
\end{tabular}
\caption{Derived parameters from the parameters in Table \ref{tab:physical_sbs_parameters}. These derived parameters are used in the numerical plots.}
\label{tab:derived_sbs_parameters}
\end{table*}

\begin{figure}[!t] \centering \subfloat[Normalized susceptibility $(\Gamma_B/2)\chi_i(\Omega)$, where $\chi_i(\Omega)=\bigl(\Gamma_B/2+\mathrm{i}(\Omega_{B,i}-\Omega)\bigr)^{-1}$. \label{fig:brillouin_susceptibility_real_imag_panel}]{%
\includegraphics[width=0.48\linewidth]{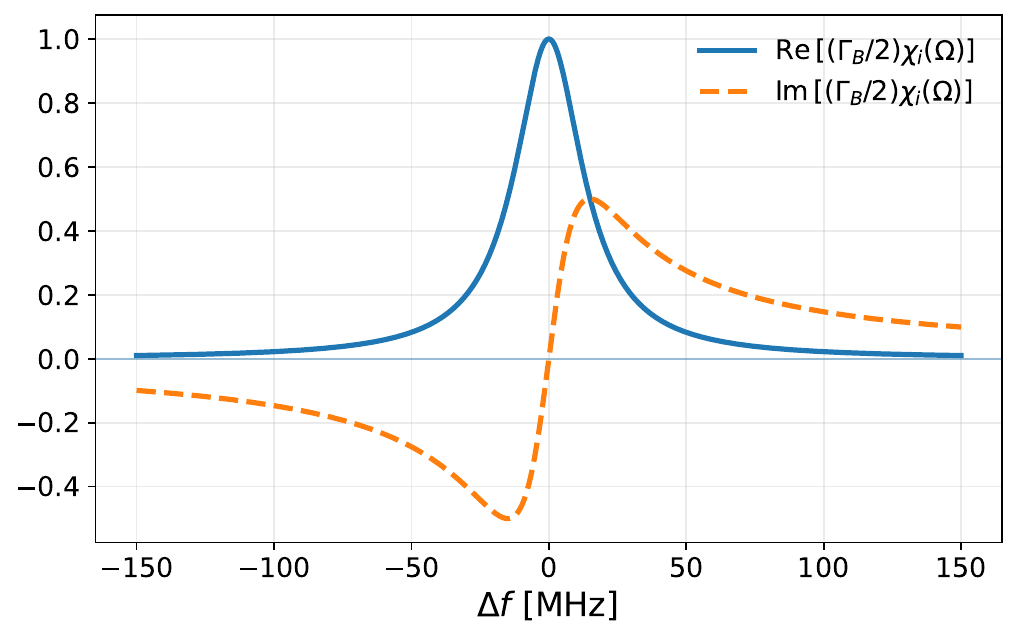}} \hfill \subfloat[Corresponding SBS-induced channel coefficient $\kappa_i(\Omega)=|\tilde{G}_i|^2\chi_i(\Omega)\delta z$ for $P_p=100\,\mathrm{mW}$ and $\delta z=0.5\,\mathrm{m}$. \label{fig:kappa_real_imag_panel}]
{ \includegraphics[width=0.48\linewidth]{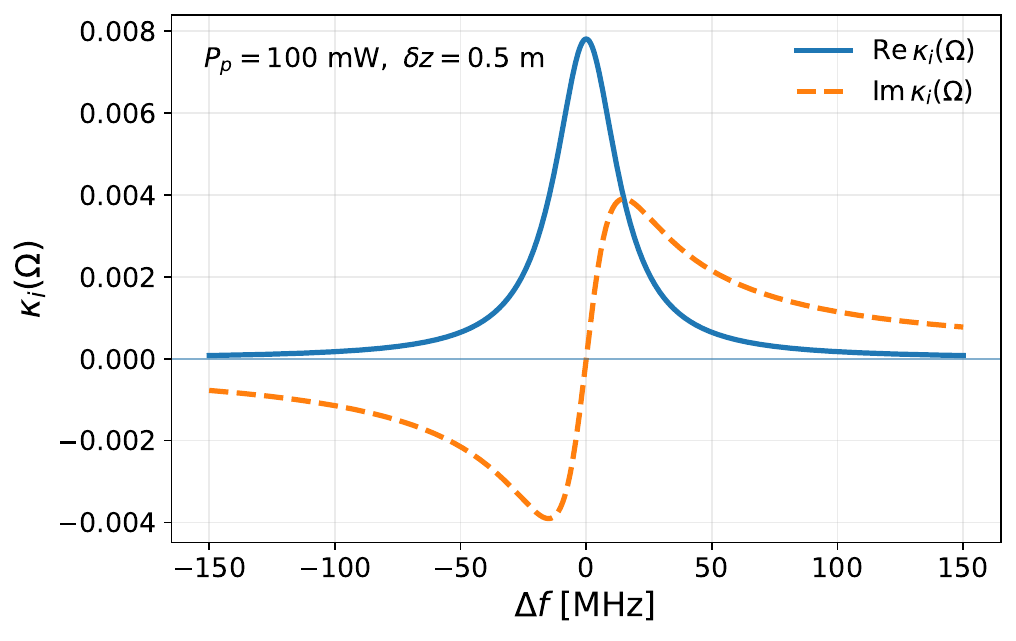}} \caption{ Real and imaginary parts of the Brillouin response using representative parameters for a standard silica single-mode fiber near $1550\,\mathrm{nm}$, summarized in Tables~\ref{tab:physical_sbs_parameters} and~\ref{tab:derived_sbs_parameters}. In both panels, the real part gives the resonant SBS gain response, while the imaginary part gives the dispersive phase response and changes sign across the Brillouin resonance. } \label{fig:brillouin_response_real_imag} \end{figure}

We will now numerically evaluate the SBS channel coefficients. For this, we have considered a standard silica single-mode fiber near $1550\,\mathrm{nm}$. The parameters are summarized in Tables~\ref{tab:physical_sbs_parameters} and~\ref{tab:derived_sbs_parameters}. Figure~\ref{fig:brillouin_response_real_imag} shows the frequency response of the local SBS channel. The real part gives the Lorentzian gain response while the imaginary part gives the dispersive phase response. Figure~\ref{fig:short_segment_power_factor} shows the short-segment power factor
$|\mu_i(\Omega)|^2 \simeq 1-\alpha_s\delta z+\xi(\Omega)P_p\delta z$,
which quantifies the deterministic attenuation and SBS-induced amplification
experienced by the Stokes mode over a fiber segment of length $\delta z$. We have considered two cases: different segment length $\delta z$, and the different pump power $P_p$. In both cases, the response is peaked at $\Omega=\Omega_B$, where the Brillouin gain spectrum is maximal. Figure~\ref{fig:sbs_phonon_noise} shows the corresponding SBS-induced
phonon-noise coefficient $|\nu_i(\Omega)|^2$. This response is also peaked at $\Omega=\Omega_B$.

\begin{figure}[!t]
    \centering

    \subfloat[Dependence on the segment length $\delta z$ for fixed pump power.
    \label{fig:mu_delta_z}]{%
        \includegraphics[width=0.48\linewidth]
        {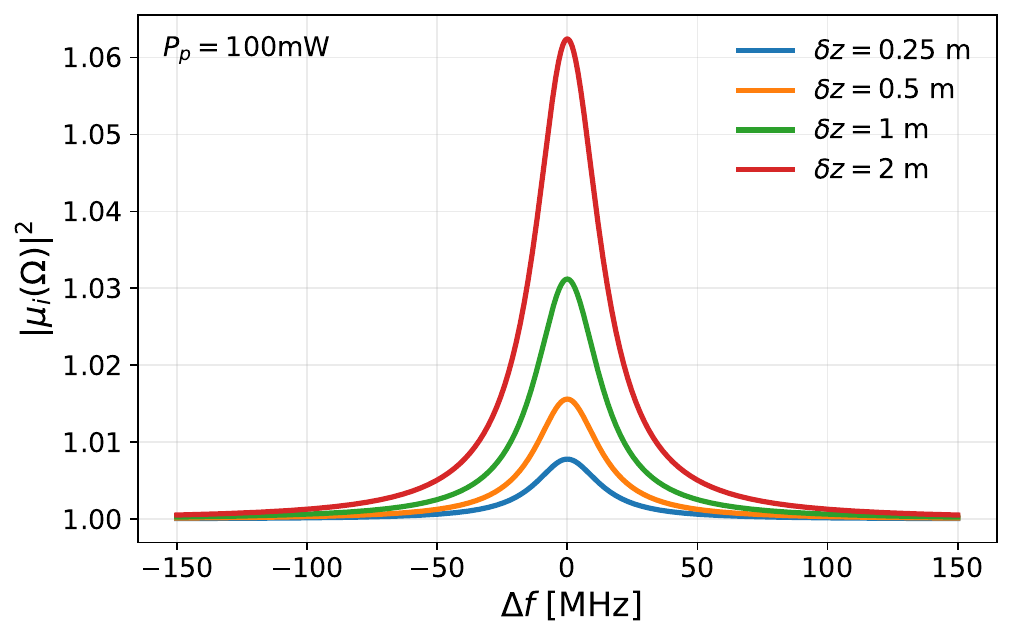}}
    \hfill
    \subfloat[Dependence on pump power $P_p$ for fixed segment length.
    \label{fig:mu_power}]{%
        \includegraphics[width=0.48\linewidth]
        {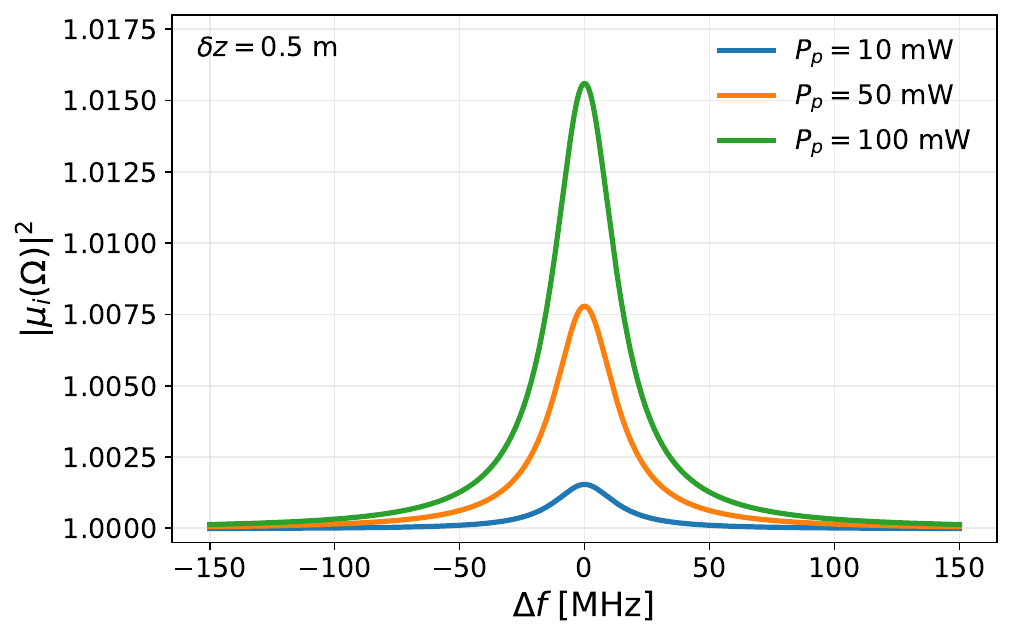}}

    \caption{
    Short-segment SBS power factor
    $|\mu_i(\Omega)|^2$
    as a function of detuning $\Delta f=(\Omega-\Omega_B)/2\pi$ using representative parameters for a standard silica single-mode fiber near
$1550\,\mathrm{nm}$, summarized in
Tables~\ref{tab:physical_sbs_parameters} and~\ref{tab:derived_sbs_parameters}.
    }
    \label{fig:short_segment_power_factor}
\end{figure}

\begin{figure}[!t]
    \centering

    \subfloat[Dependence on the segment length $\delta z$ for fixed pump power.
    \label{fig:nu_delta_z}]{%
        \includegraphics[width=0.48\linewidth]{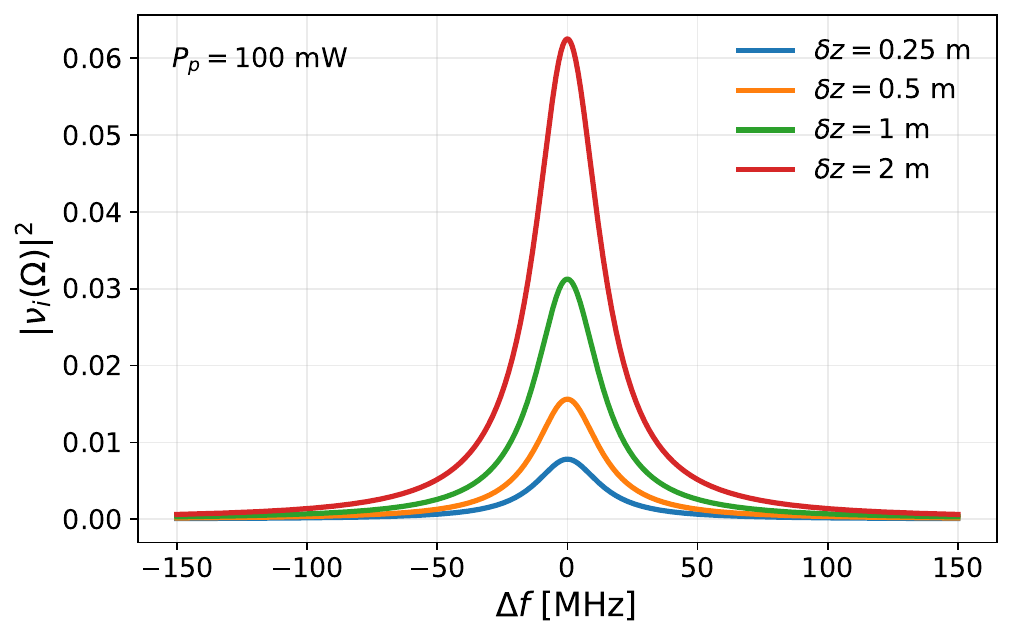}}
    \hfill
    \subfloat[Dependence on pump power $P_p$ for fixed segment length.
    \label{fig:nu_power}]{%
        \includegraphics[width=0.48\linewidth]{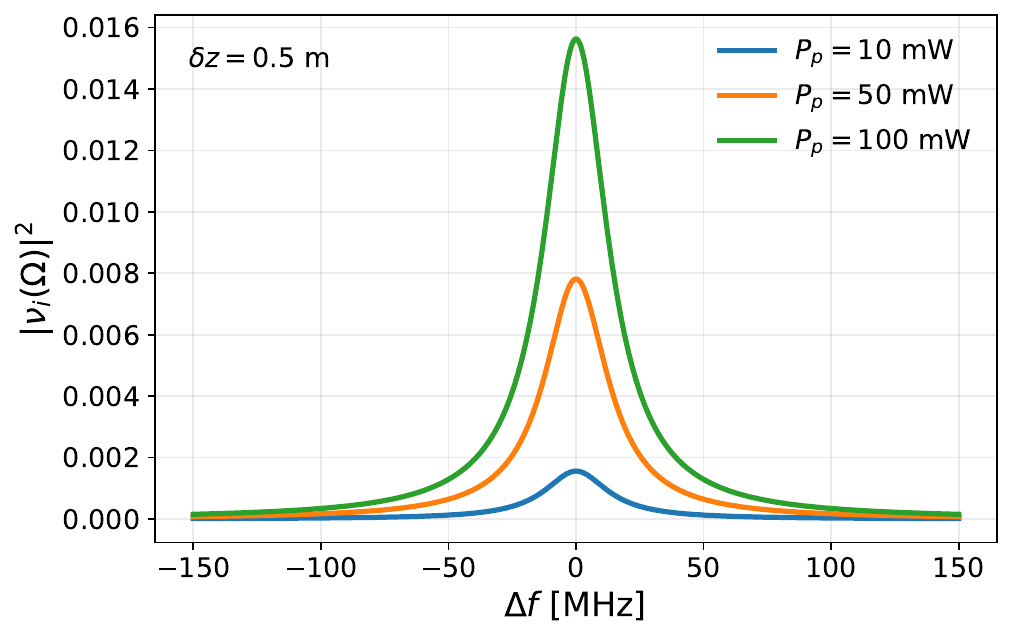}}

   \caption{
SBS-induced phonon-noise spectrum 
$|\nu_i(\Omega)|^2 $
as a function of detuning $\Delta f=(\Omega-\Omega_B)/2\pi$ using representative parameters for a standard silica single-mode fiber near
$1550\,\mathrm{nm}$, summarized in
Tables~\ref{tab:physical_sbs_parameters} and~\ref{tab:derived_sbs_parameters}.
}
    \label{fig:sbs_phonon_noise}
\end{figure}

\subsection{Classical limit}
\label{subsec:classical_equivalence}

We now show that the quantum input--output model reproduces the classical SBS
power-gain equation in the appropriate limit. Starting from the local propagation
equation for the Stokes mode,
\begin{equation}
\begin{aligned}
\partial_s a_s(s,\Omega)
&=
\left(
-\frac{\alpha_{s,i}}{2}
+|\widetilde G_i|^2\chi_i(\Omega)
\right)a_s(s,\Omega)
\\
&\quad
-\mathrm{i}\widetilde G_i\sqrt{\Gamma}\chi_i(\Omega)\,
b_{\mathrm{in}}^\dagger(s,\Omega)
+
\sqrt{\alpha_{s,i}}\,v_i(s,\Omega).
\end{aligned}
\label{eq:classical_start}
\end{equation}
we consider the deterministic mean-field limit. Since the Langevin noise operators
have zero mean, $\langle v_i(s,\Omega)\rangle = 0,
\langle b_{\mathrm{in}}(s,\Omega)\rangle = 0,$
they do not contribute to the evolution of the coherent field amplitude
$\langle a_s\rangle$. Taking the expectation value of
Eq.~\eqref{eq:classical_start} therefore gives
\begin{equation}
\partial_s \langle a_s(s,\Omega)\rangle
=
\left(
-\frac{\alpha_{s,i}}{2}
+
|\widetilde G_i|^2\chi_i(\Omega)
\right)
\langle a_s(s,\Omega)\rangle .
\label{eq:classical_amp}
\end{equation}

The optical Stokes power is proportional to the photon number, or equivalently to
the squared field amplitude in the coherent-state limit $P_s(s,\Omega)
\propto
|\langle a_s(s,\Omega)\rangle|^2 .$
Using Eq.~\eqref{eq:classical_amp}, the corresponding power evolution becomes
\begin{equation}
\frac{dP_s}{ds}
=
\left[
-\alpha_{s,i}
+
2|\widetilde G_i|^2\operatorname{Re}\chi_i(\Omega)
\right]
P_s .
\label{eq:power_quantum}
\end{equation}

In the absence of ordinary fiber attenuation, the classical SBS equation is $\frac{dP_s}{ds}
=
\xi(\Omega)P_pP_s ,$
where $P_p$ is the pump power and $\xi(\Omega)$ is the Brillouin gain spectrum.
Including ordinary propagation loss gives the more general form $\frac{dP_s}{ds}
=
\left[
-\alpha_{s,i}
+
\xi(\Omega)P_p
\right]
P_s$. We thus identify the effective quantum coupling
through $2|\widetilde G_i|^2\operatorname{Re}\chi_i(\Omega)
=
\xi(\Omega)P_p .$
Equivalently, using
$\kappa_i(\Omega)=|\widetilde G_i|^2\chi_i(\Omega)\delta z$, we obtain $2\operatorname{Re}\kappa_i(\Omega)
=
\xi(\Omega)P_p\,\delta z .$ Thus the real part of $\kappa_i(\Omega)$ reproduces the classical SBS
power gain over the segment, while the imaginary part gives the associated
dispersive phase shift.

$2|\widetilde G_i|^2\operatorname{Re}\chi_i(\Omega)$ and the classical Brillouin gain spectrum $\xi(\Omega)$ have the same Lorentzian form:

\begin{align}
2|\widetilde G_i|^2\operatorname{Re}\chi_i(\Omega)
&=
|\widetilde G_i|^2
\frac{\Gamma}
{(\Gamma/2)^2+(\Omega_{B,i}-\Omega)^2},
\nonumber\\
\xi(\Omega)
&=
\xi_0
\frac{(\Gamma_B/2)^2}
{(\Omega_B-\Omega)^2+(\Gamma_B/2)^2}.
\label{eq:quantum_lorentzian}
\end{align}
Assuming $\Gamma=\Gamma_B$, the effective quantum coupling may be calibrated by matching  $2|\widetilde G_i|^2\operatorname{Re}\chi_i(\Omega)
=
\xi(\Omega)P_p .$ Equivalently, up to normalization conventions  $|\widetilde G_i|^2
=
\frac{\xi_0P_p\Gamma_B}{4}$.

The classical limit can also be obtained from the input--output map,
\begin{equation}
\langle a_{s,\mathrm{out}}^\dagger a_{s,\mathrm{out}}\rangle
=
|\mu_i(\Omega)|^2
\langle a_{s,\mathrm{in}}^\dagger a_{s,\mathrm{in}}\rangle
+
|\nu_i(\Omega)|^2\left(n_{\mathrm{th},i}+1\right),
\end{equation}
where the cross terms vanish because the input, optical vacuum, and phonon bath
modes are independent, and the optical vacuum contributes no normally ordered
photons. The second term describes spontaneous and thermal Brillouin scattering.
For the deterministic stimulated contribution, we neglect this additive noise
term and obtain $P_{s,i+1}(\Omega)
\approx
|\mu_i(\Omega)|^2P_{s,i}(\Omega).$
To first order in $\delta z$, $P_{s,i+1}(\Omega)
\simeq
\left[
1+
\left(
-\alpha_{s,i}
+
2|\widetilde G_i|^2\operatorname{Re}\chi_i(\Omega)
\right)\delta z
\right]
P_{s,i}(\Omega).$
Taking the continuum limit gives
\begin{equation}
\frac{dP_s}{ds}
=
\left[
-\alpha_s
+
2|\widetilde G|^2\operatorname{Re}\chi(\Omega)
\right]P_s.
\end{equation}
Comparing with the classical SBS equation gives  $2\operatorname{Re}\kappa_i(\Omega)
=
\xi(\Omega)P_p\,\delta z.$

\section{Continuous variable formalism}
\label{sec:fiber_cascade}
Continuous-variable (CV) quantum systems are bosonic modes described by canonical quadrature operators rather than finite-dimensional qubits \cite{bookserfini, HolevoBook}. For an $n$-mode system, we define the quadrature vector $\hat{\mathbf R}
=
(\hat x_1,\hat p_1,\ldots,\hat x_n,\hat p_n)^T,$ with canonical commutation relations $[\hat R_j,\hat R_k]= \mathrm{i}\,\tilde \Omega_{jk},$
where $\tilde \Omega$ is the symplectic form. For a single mode with annihilation operator $\hat a$, the quadratures are
\begin{equation}
q=\frac{ a+ a^\dagger}{\sqrt 2},
\qquad
 p=\frac{ a-  a^\dagger}{\mathrm{i}\sqrt 2}.
\end{equation}
A Gaussian state is any state whose Wigner function is Gaussian in phase space. Such a state is completely characterized by its first and second moments. The mean vector is $ d=\langle \hat{R}\rangle$
and the covariance matrix is $V_{jk}
=
\frac12\left\langle
\{\hat R_j-d_j,\hat R_k-d_k\}
\right\rangle.$

We will now show that each fiber segment $i$ under the linearized stimulated Brillouin scattering interaction defines a phase-insensitive single-mode Gaussian channel acting on the Stokes mode. For this, we consider an optical probe mode entering segment \(i\) as modeled by a local Gaussian channel \(\mathcal E_i(z_c)\), whose parameters depend on the correlation-peak position \(z_c\). As derived in Section \ref{sec:derivation_of_input_output}, the input Stokes mode in the segment $i$ evolves as 
\begin{equation}
\begin{aligned}
\label{eq:aout_ain}
a_{s,\mathrm{out}}(\Omega)
&=
\mu_i(\Omega)\,a_{s,\mathrm{in}}(\Omega)
+
\sqrt{\eta'_i}\,v_i(\Omega)
+
\nu_i(\Omega)\,b_{\mathrm{in}}^\dagger(\Omega), \\
a_{s,\mathrm{out}}^\dagger(\Omega)
&=
\mu_i^*(\Omega)\, a_{s,\mathrm{in}}^\dagger(\Omega)
+
\sqrt{\eta'_i}\, v_i^\dagger(\Omega)
+
\nu_i^*(\Omega)\, b_{\mathrm{in}}(\Omega).
\end{aligned}
\end{equation}
where $\mu_i(\Omega)=\sqrt{\eta_i}+\kappa_i(\Omega)$. We will now show that each fiber segment can be modeled as a single-mode Gaussian channel.

\begin{proposition}[Gaussian channel representation of a fiber segment]
 In quadrature form, the channel of the segment $i$ is described by
\begin{align}
d_{i+1} &= X_i d_i, \\
V_{i+1} &= X_i V_i X_i^T + Y_i,
\end{align}
where the deterministic matrix is
\begin{align}
X_i(\Omega) 
&=
\begin{pmatrix}
\Re \mu_i(\Omega) & -\Im \mu_i(\Omega) \\
\Im \mu_i(\Omega) & \Re \mu_i(\Omega)
\end{pmatrix} \\
&=
\begin{pmatrix}
\sqrt{\eta_i} + \Re \kappa_i(\Omega) & -\Im \kappa_i(\Omega) \\
\Im \kappa_i(\Omega) & \sqrt{\eta_i} + \Re \kappa_i(\Omega)
\end{pmatrix},
\end{align}
with $\mu_i(\Omega) = \sqrt{\eta_i} + \kappa_i(\Omega)$. The added-noise matrix is
\begin{equation}
Y_i =
\left[
\left(n_{\mathrm{th},i}+\frac{1}{2}\right)|\nu_i(\Omega)|^2
+
\frac{1-\eta_i}{2}
\right]\mathbb{I}_2.
\end{equation}
\end{proposition}

\begin{proof}
    We write $\mu_i$ and $\nu_i$ in terms of their real and imaginary parts as $\mu_i=\mu_{R,i}+ \mathrm{i} \mu_{I,i}$, and $\nu_i=\nu_{R,i}+ \mathrm{i}\nu_{I,i}$ respectively. Then Eq.~\eqref{eq:aout_ain} becomes
\begin{align}
q_{\mathrm{out}}
&=
\mu_{R,i}\, q - \mu_{I,i}\, p
+ \sqrt{1-\eta_i}\, q_v
+ \nu_{R,i}\, q_b + \nu_{I,i}\, p_b,
\nonumber\\
p_{\mathrm{out}}
&=
\mu_{I,i}\, q + \mu_{R,i}\, p
+ \sqrt{1-\eta_i}\, p_v
+ \nu_{I,i}\, q_b - \nu_{R,i}\, p_b.
\end{align}
where the quadrature vectors are 
\begin{equation}
\begin{gathered}
R_{\mathrm{in}}=
\begin{pmatrix}q\\ p\end{pmatrix},\qquad
R_{\mathrm{out}}=
\begin{pmatrix}q_{\mathrm{out}}\\ p_{\mathrm{out}}\end{pmatrix},\\
R_b=
\begin{pmatrix}q_b\\ p_b\end{pmatrix},\qquad
R_v=
\begin{pmatrix}q_v\\ p_v\end{pmatrix}.
\end{gathered}
\end{equation}
the segment transformation can be written as $R_{\mathrm{out}} = X_i R_{\mathrm{in}} + Z_i R_b + \sqrt{1-\eta_i}\, R_v$ with
\begin{equation}
X_i=
\begin{pmatrix}
\Re \mu_i & -\Im \mu_i \\
\Im \mu_i & \Re \mu_i
\end{pmatrix},
\quad
Z_i=
\begin{pmatrix}
\Re \nu_i & \Im \nu_i \\
\Im \nu_i & -\Re \nu_i
\end{pmatrix}.
\label{eq:XiZi}
\end{equation}
We define the displacement vectors as $d_{\mathrm{in}}=\langle R_{\mathrm{in}} \rangle$, $d_{\mathrm{out}}=\langle R_{\mathrm{out}} \rangle$, $d_b=\langle R_b \rangle$, and $d_v=\langle R_v \rangle$, and the covariance matrices as $(V_{\mathrm{in}})_{mn}
=
\frac{1}{2}
\left\langle
\Delta R_{\mathrm{in},m}\Delta R_{\mathrm{in},n}
+
\Delta R_{\mathrm{in},n}\Delta R_{\mathrm{in},m}
\right\rangle$ with analogous definitions for $V_{\mathrm{out}}, V_b,$ and $V_v$, where $\Delta R = R - \langle R \rangle$.

Taking expectation values gives $d_{\mathrm{out}} = X_i d_{\mathrm{in}} + Z_i d_b + \sqrt{1-\eta_i}\, d_v.$ For vacuum and thermal baths, $d_b = 0$ and $d_v = 0$, hence
\begin{equation}
d_{\mathrm{out}} = X_i d_{\mathrm{in}}.
\label{eq:dout}
\end{equation}

For the covariance matrix, we have $\Delta R_{\mathrm{out}} = X_i \Delta R_{\mathrm{in}} + Z_i \Delta R_b + \sqrt{1-\eta_i}\,\Delta R_v.$ Assuming all inputs are uncorrelated, this yields
\begin{equation}
V_{\mathrm{out}} = X_i V_{\mathrm{in}} X_i^T + Z_i V_b Z_i^T + (1-\eta_i)V_v.
\end{equation}
Defining $Y_i = Z_i V_b Z_i^T + (1-\eta_i)V_v,$ we obtain the Gaussian-channel form
$V_{\mathrm{out}} = X_i V_{\mathrm{in}} X_i^T + Y_i.$

For a thermal phonon bath with mean occupation $n_{\mathrm{th},i}$,
\begin{equation}
V_b =
\left(
n_{\mathrm{th},i}+\frac{1}{2}
\right) I_2,
\quad
V_v = \frac{1}{2} I_2.
\label{eq:VbVv}
\end{equation}

Thus,
\begin{align}
Y_i
&= Z_i V_b Z_i^T + (1-\eta_i)V_v
\nonumber\\
&=
\left(n_{\mathrm{th},i}+\frac{1}{2}\right)|\nu_i|^2 I_2
+
\frac{1-\eta_i}{2} I_2.
\label{eq:Yi_final}
\end{align}

Therefore,
\begin{align}
d_{\mathrm{out}} &= X_i d_{\mathrm{in}}, \nonumber\\
V_{\mathrm{out}} &= X_i V_{\mathrm{in}} X_i^T
+
\left[
\left(
n_{\mathrm{th},i}+\frac{1}{2}
\right)|\nu_i|^2
+
\frac{1-\eta_i}{2}
\right] I_2.
\label{eq:gaussian_channel_final}
\end{align}
where
\begin{equation}
X_i(\Omega)=
\begin{pmatrix}
\sqrt{\eta_i}+\Re \kappa_i(\Omega) & -\Im \kappa_i(\Omega) \\
\Im \kappa_i(\Omega) & \sqrt{\eta_i}+\Re \kappa_i(\Omega)
\end{pmatrix}.
\label{eq:Xi_final}
\end{equation}

This shows that $\sqrt{\eta_i}$ represents attenuation, $\Re \kappa_i(\Omega)$ gives SBS gain/loss, $\Im \kappa_i(\Omega)$ produces phase rotation, and $Y_i$ contains both vacuum noise and thermal Brillouin noise. Thus, each short fiber segment is modeled as a one-mode Gaussian channel.
\end{proof}

The complete-positivity condition imposes a fundamental constraint on the relationship between gain and added noise:
\begin{equation}
Y_i+\frac{\mathrm{i}}{2} \tilde \Omega-X_i\frac{\mathrm{i}}{2} \tilde \Omega X_i^T\ge 0,
\quad
\tilde \Omega=
\begin{pmatrix}
0&1\\
-1&0
\end{pmatrix}.
\end{equation}
This reduces to the determinant being non-negative as:
\begin{equation}
\frac{1-\eta_i}{2}
+
\left(n_{\mathrm{th},i}+\frac12\right)|\nu_i|^2
\ge
\frac12\left|1-\left|\sqrt{\eta_i}+\kappa_i\right|^2\right|.
\end{equation}
The inequality represents a quantum noise bound: the total added noise must exceed a minimum value determined by the channel's net gain or loss. Physically, this is necessary to satisfy the bosonic commutation relationships at all points.

\subsection{Fiber as a cascade of quantum channels}

The entire fiber is therefore described by a cascade of quantum channels,
\begin{equation}
\mathcal E_{\mathrm{fiber}}
=
\mathcal E_L
\circ
\mathcal E_{L-1}
\circ
\cdots
\circ
\mathcal E_1.
\label{eq:fiber_cascade}
\end{equation} 
The cascade of two Gaussian channels with parameters \((X_1,Y_1)\) and \((X_2,Y_2)\) is again Gaussian, with total deterministic matrix $X_{21}=X_2X_1$, and the noise matrix $Y_{21}=X_2Y_1X_2^T+Y_2.$ 

Hence, the entire fiber, modeled as a cascade of \(L\) local Gaussian channels, has total parameters
\begin{equation}
\begin{aligned}
\label{eqn:fiber-channel}
X_{\mathrm{tot}} &=X_LX_{L-1}\cdots X_1, \\
Y_{\mathrm{tot}}
&=
\sum_{k=1}^L
\left(
X_LX_{L-1}\cdots X_{k+1}
\right)
Y_k
\left(
X_LX_{L-1}\cdots X_{k+1}
\right)^T.
\end{aligned}
\end{equation}

\begin{proposition}
    The explicit form of the parameters $X_{\mathrm{tot}}$ and $Y_{\mathrm{tot}}$ are as follows

\begin{equation}
  X_{\mathrm{tot}}=
\begin{pmatrix}
\Re \mu_{\mathrm{tot}} & -\Im \mu_{\mathrm{tot}}\\
\Im \mu_{\mathrm{tot}} & \Re \mu_{\mathrm{tot}}
\end{pmatrix},
\qquad
Y_{\mathrm{tot}}=y_{\mathrm{tot}}\mathbb I_2
\end{equation}
where parameters $\mu_{\mathrm{tot}}$ and $ y_{\mathrm{tot}}$ given by
\begin{equation}
    \begin{aligned}
        \mu_{\mathrm{tot}}&=\prod_{k=1}^L \mu_k \\
        y_{\mathrm{tot}} &=\sum_{k=1}^L
\left(
\prod_{j=k+1}^N |\mu_j|^2
\right)
\left[
\tfrac{\eta'_k}{2}
+
\left(n_{\mathrm{th},k}+\tfrac12\right)|\nu_k|^2
\right].
    \end{aligned}
\end{equation}
\end{proposition}
\begin{proof}

We now show that the total deterministic matrix \(X_{\mathrm{tot}}\) has
the same rotation-scaling structure as each local segment matrix \(X_i\). For each fiber segment \(i\), the deterministic part of the Gaussian
channel is
\begin{equation}
X_i =
\begin{pmatrix}
\Re \mu_i & -\Im \mu_i \\
\Im \mu_i & \Re \mu_i
\end{pmatrix},
\qquad
\mu_i = \sqrt{\eta_i} + \kappa_i .
\end{equation}
This matrix is the real \(2\times 2\) representation of multiplication by
the complex number \(\mu_i\). To see this, let $\mu_i = a_i + \mathrm{i} b_i$ with \(a_i = \Re \mu_i\) and \(b_i = \Im \mu_i\). Then
\begin{equation}
X_i =
\begin{pmatrix}
a_i & -b_i \\
b_i & a_i
\end{pmatrix}.
\end{equation}

Consider now two consecutive segments with coefficients
\(\mu_1 = a_1 + i b_1\) and \(\mu_2 = a_2 + i b_2\). Their deterministic
matrices satisfy
\begin{equation}
\begin{aligned}
X_2 X_1
&=
\begin{pmatrix}
a_2 & -b_2 \\
b_2 & a_2
\end{pmatrix}
\begin{pmatrix}
a_1 & -b_1 \\
b_1 & a_1
\end{pmatrix}
\\
&=
\begin{pmatrix}
a_2 a_1-b_2 b_1 & -(a_2 b_1+b_2 a_1) \\
a_2 b_1+b_2 a_1 & a_2 a_1-b_2 b_1
\end{pmatrix}.
\end{aligned}
\end{equation}
On the other hand, the product of the corresponding complex coefficients is $\mu_2 \mu_1
=
(a_2 + i b_2)(a_1 + i b_1)
=
(a_2 a_1 - b_2 b_1) + i(a_2 b_1 + b_2 a_1).$ Therefore,
\begin{equation}
X_2 X_1 =
\begin{pmatrix}
\Re(\mu_2 \mu_1) & -\Im(\mu_2 \mu_1) \\
\Im(\mu_2 \mu_1) & \Re(\mu_2 \mu_1)
\end{pmatrix}.
\end{equation}
Hence, the product of two matrices of this form is again a matrix of the same form, with a coefficient given by the product of the associated complex
numbers. Applying this recursively to the cascade of \(L\) fiber segments, $X_{\mathrm{tot}} = X_L X_{L-1} \cdots X_1,$  we conclude that
\begin{equation}
X_{\mathrm{tot}} =
\begin{pmatrix}
\Re \mu_{\mathrm{tot}} & -\Im \mu_{\mathrm{tot}} \\
\Im \mu_{\mathrm{tot}} & \Re \mu_{\mathrm{tot}}
\end{pmatrix},
\qquad
\mu_{\mathrm{tot}} = \prod_{k=1}^{L} \mu_k .
\end{equation}

This shows that the full fiber again acts as a phase-insensitive
rotation-scaling transformation on the quadratures. The magnitude
\(|\mu_{\mathrm{tot}}|\) gives the net attenuation or amplification, while
the argument \(\arg(\mu_{\mathrm{tot}})\) gives the total phase rotation
accumulated along the fiber.

We now derive the isotropic form of the total added-noise matrix
\(Y_{\mathrm{tot}}\). From the cascade rule for Gaussian channels,
\begin{equation}
Y_{\mathrm{tot}}
=
\sum_{k=1}^{L}
T_k Y_k T_k^T,
\qquad
T_k := X_L X_{L-1} \cdots X_{k+1},
\end{equation}
where \(T_k\) describes the deterministic propagation through all segments
following segment \(k\). For each segment \(k\), the local added-noise matrix is proportional to
the identity, $Y_k = y_k I_2,$ with $y_k :=
\frac{1-\eta_k}{2}
+
\left(n_{\mathrm{th},k}+\frac12\right)|\nu_k|^2.$ Hence, $Y_{\mathrm{tot}}
=
\sum_{k=1}^{L} y_k\, T_k T_k^T.$ We further denote $\mu^{(k)} := \prod_{j=k+1}^{L} \mu_j $, and obtain $T_k T_k^T =|\mu^{(k)}|^2 I_2. $  Using $|\mu^{(k)}|^2
=
\left|\prod_{j=k+1}^{L}\mu_j\right|^2
=
\prod_{j=k+1}^{L} |\mu_j|^2,$ we find
\begin{equation}
\begin{aligned}
T_k Y_k T_k^T
&=
\left(
\prod_{j=k+1}^{L} |\mu_j|^2
\right)
\\
&\quad\times
\left[
\frac{1-\eta_k}{2}
+
\left(n_{\mathrm{th},k}+\frac12\right)|\nu_k|^2
\right] I_2.
\end{aligned}
\end{equation}
Summing over all segments yields
\begin{equation}
\begin{aligned}
Y_{\mathrm{tot}} &= y_{\mathrm{tot}} I_2,\\
y_{\mathrm{tot}}
&=
\sum_{k=1}^{L}
\left(
\prod_{j=k+1}^{L} |\mu_j|^2
\right)
\left[
\frac{1-\eta_k}{2}
+
\left(n_{\mathrm{th},k}+\frac12\right)|\nu_k|^2
\right].
\end{aligned}
\end{equation}

\end{proof}

\section{Quantum Fisher Information for SBS Parameter Sensing}
\label{sec:qfi_sbs_parameter_sensing}

In this section, we will formulate stimulated Brillouin scattering (SBS) sensing as a local quantum-parameter-estimation problem of an unknown single scalar parameter $\theta$. We leave the multi-parameter case for future work.  As derived in the section \ref{sec:derivation_of_input_output}, the local input--output relation in the active segment is
\begin{equation}
\label{eq:sbs_local_channel_theta}
\hat a_{\mathrm{out}}(\Omega;\theta)
=
\mu_i(\Omega;\theta)\hat a_{\mathrm{in}}(\Omega)
+
\sqrt{1-\eta}\,\hat v_i(\Omega)
+
\nu_i(\Omega;\theta)\hat b_i^\dagger(\Omega).
\end{equation}
 The parameter $\theta$ may be the Brillouin frequency $f_B=\Omega_B/(2\pi)$, the temperature $T$, the strain $\epsilon$, the Brillouin linewidth $\Gamma_B$, or the effective gain coefficient $\xi_0$.

 We assume uniform optical loss \(\eta_k = \eta\) for all segments and a coherent input probe state with amplitude \(\sqrt{E}\), where $E$ is the mean input photon number. We assume that the pump and probe overlap at a segment $i$, so that the SBS interaction is active only in that segment. This assumption is also behind the BOCDA system. For every segment $j\neq i$, we thus set $\kappa_j=0$ and $\nu_j=0$. In the active segment $i$, the SBS coefficients depend on the parameter $\theta$. 
 
 \subsection{Input as a coherent state, output a displaced thermal state}
 Throughout this work, we assume that the input state is a coherent state. For a coherent input $\ket{\sqrt E}$, the output is a displaced thermal state $\rho_\theta
=
S_{N_i(\Omega;\theta)}
\left(
\beta_i(\Omega;\theta)
\right),$
with (derived at Appendix \ref{app:displaced-thermal-output})
\begin{align}
  \label{eq:clean_parameters_theta}
\beta_i(\Omega;\theta)
&=
\eta^{(L-1)/2}
\left[
\sqrt{\eta}+\kappa_i(\Omega;\theta)
\right]\sqrt E, 
\end{align}
\begin{align}
 \label{eq:clean_parameters_N}
N_i(\Omega;\theta)
&=
\eta^{L-i}
\Bigg[
\left(n_{\mathrm{th}}(\theta)+\frac{1}{2}\right)
|\nu_i(\Omega;\theta)|^2
\nonumber\\
&\hspace{2.0cm}
+
\frac{|\sqrt{\eta}+\kappa_i(\Omega;\theta)|^2-\eta}{2}
\Bigg].
\end{align}

\paragraph{Coherent part}
Eq. \ref{eq:clean_parameters_theta} represents the coherent part where the factor $\eta^{(L-1)/2} =\eta^{(L-i)/2} \eta^{(i-1)/2} $ accounts for attenuation in the
passive segments before and after the interaction, while
$\sqrt{\eta}+\kappa_i$ describes the active segment's response.
Since $\kappa_i$ is complex, this response can modify both
the amplitude and phase of the coherent field. 

\paragraph{Thermal part}
The thermal component, Eq. \ref{eq:clean_parameters_N},  can be understood through an ideal,
lossless two-mode-squeezing interaction between the Stokes mode
$\hat a_s$ and the phonon mode $\hat b$. Choosing a convenient
coupling phase, the evolution is
\begin{equation}
    \hat U(r)
    = \exp\!\left[
        r(\hat a_s^\dagger\hat b^\dagger-\hat a_s\hat b)
      \right],
\end{equation}
where $r$ is the squeezing strength. The Stokes operator
therefore transforms as
\begin{equation}
    \hat a_{s,\mathrm{out}}
    = \hat U^\dagger(r)\hat a_s\hat U(r)
    = \hat a_s\cosh r+\hat b^\dagger\sinh r.
\end{equation}
For an initially uncorrelated coherent Stokes state
$|\sqrt{E}\rangle$ and a thermal phonon state with occupation
$n_{\mathrm{th}}$, tracing out the phonons yields a displaced
thermal optical state with
\begin{equation}
\begin{aligned}
    \beta
    &= \langle\hat a_{s,\mathrm{out}}\rangle
     = \sqrt{E}\cosh r,\\
    N
    &= \langle\hat a_{s,\mathrm{out}}^\dagger
              \hat a_{s,\mathrm{out}}\rangle-|\beta|^2\\
    &= (n_{\mathrm{th}}+1)\sinh^2r.
\end{aligned}
\end{equation}
The factor $n_{\mathrm{th}}+1$ follows from
$\langle\hat b\hat b^\dagger\rangle=n_{\mathrm{th}}+1$.
Thus, the added thermal occupation includes both a
thermal-bath contribution and spontaneous emission, which
remains nonzero for vacuum phonons when $r>0$. The coefficient $|\nu_i(\Omega;\theta)|^2$ can be interpreted
as $\sinh^2 r_i(\Omega;\theta)$, where
$r_i=\operatorname{arsinh}|\nu_i|$ is an effective local
two-mode-squeezing parameter.

The factor $\eta^{L-i}$
accounts for attenuation after the active segment $i$.
The second term in the bracket accounts for the SBS-induced change in the
optical input's vacuum fluctuations. Both terms arise from the
same interaction and, to first order in the segment length,
combine to give
\begin{equation}
  N_i \simeq \eta^{L-i}
[n_{\mathrm{th}}(\theta)+1]|\nu_i(\Omega;\theta)|^2.  
\end{equation}

Within this model, the coherent part $|\beta_i|^2\propto E$, whereas thermal part $N_i$ is independent of $E$. The explanation goes as follows. For the ideal
two-mode amplifier,
\begin{equation}
    N = (n_{\mathrm{th}}+1)\sinh^2 r.
\end{equation}
At fixed interaction conditions, the undepleted pump determines $r$,
while the bath temperature and acoustic frequency determine
$n_{\mathrm{th}}$. Neither depends on the coherent probe photon
number $E$ in this model. Even at $n_{\mathrm{th}}=0$, spontaneous
pair creation gives $N=\sinh^2 r>0$ for nonzero coupling. Finally, in the pure-loss limit $\kappa_i(\Omega;\theta)=0$ and $\nu_i(\Omega;\theta)=0$, one obtains $N_i(\Omega;\theta)=0$, so a coherent input remains coherent under vacuum loss.

As we would like to track the change in fiber parameters, we can do it by tracking the change in the local Brillouin frequency $f_B$ which depends on the local fiber parameters. Operationally, when the pump and probe frequencies are held fixed, change in $f_B$ means change in mean thermal photon number and coherent amplitude. This can be made precise via Quantum Fisher Information explained in next paragraph.

\paragraph{Quantum Fisher Information (QFI) of displaced thermal state}
Our next task is to compute the Quantum Fisher Information for $\rho_\theta$ and find the optimal measurement setting. We have provided a review of quantum Fisher information in Appendix \ref{sec:QFI}. For a Gaussian state with covariance matrix $\Sigma(\theta)$, purity $P(\theta)$, and displacement vector
$\Delta X(\theta)$, the QFI is given in  Eq.~(10) of \cite{Pinel_2013}

\begin{equation}
\label{eq:general_form}
F_Q(\theta)
=
\frac12
\frac{\mathrm{tr}\!\left[(\Sigma^{-1}\dot{\Sigma})^2\right]}{1+P^2}
+
2\frac{\dot P^2}{1-P^4}
+
\Delta \dot X^T \Sigma^{-1} \Delta \dot X
\end{equation}
where $\dot{(\cdot)} \equiv \partial_\theta (\cdot)$. For the displaced thermal state
$\rho_{\rm out}(\theta)=S_{N(\theta)}(\beta(\theta))$, the covariance matrix is
isotropic, $\Sigma(\theta)
=
\left(N(\theta)+\frac{1}{2}\right) I_2,$ and the displacement vector is $\Delta X(\theta)
=
\sqrt{2}
\begin{pmatrix}
\mathrm{Re}\,\beta(\theta)\\
\mathrm{Im}\,\beta(\theta)
\end{pmatrix}.$ The purity is $P(\theta)
=
\frac{1}{2N(\theta)+1}.$

Now, let us compute the first two terms in Eq. \ref{eq:general_form}. Let $s(\theta)=N(\theta)+\frac{1}{2}.$
Then, $\Sigma(\theta)=s(\theta)I_2$, $\Sigma^{-1}(\theta)=\frac{1}{s(\theta)}I_2$,   $\dot{\Sigma}=(\partial_\theta N)I_2$, and $\Sigma^{-1}\dot{\Sigma}=
\frac{\partial_\theta N}{s}I_2.$ Squaring and taking the trace gives
\begin{align}
\operatorname{tr}\!\left[
(\Sigma^{-1}\dot{\Sigma})^2
\right]
&=
\operatorname{tr}\!\left[
\frac{(\partial_\theta N)^2}{s^2}I_2
\right] =
2\frac{(\partial_\theta N)^2}{s^2}.
\end{align}
Since $P=\frac{1}{2N+1}=\frac{1}{2s}$, the first term in Eq. \ref{eq:general_form} becomes
\begin{align}
\frac{1}{2}
\frac{
\operatorname{tr}\!\left[
(\Sigma^{-1}\dot{\Sigma})^2
\right]
}{
1+P^2
}
&=
\frac{
(\partial_\theta N)^2/s^2
}{
1+P^2
} = \frac{4(\partial_\theta N)^2}{4s^2+1}.
\end{align}
Now, let us do the calculation for the second term in Eq. \ref{eq:general_form}. The purity derivative is $\dot P
=
\partial_\theta\left(\frac{1}{2s}\right)
=
-\frac{\partial_\theta N}{2s^2}$ and the purity term becomes
\begin{align}
2\frac{\dot P^2}{1-P^4}=
\frac{8(\partial_\theta N)^2}{16s^4-1}.
\end{align}
Adding the covariance and purity contributions gives
\begin{align}
\label{eq:first_two_term}
F_{\rm cov}(\theta)
&=
\frac{4(\partial_\theta N)^2}{4s^2+1}
+
\frac{8(\partial_\theta N)^2}{16s^4-1}
\nonumber \\
&=
(\partial_\theta N)^2
\left[
\frac{4}{4s^2+1}
+
\frac{8}{(4s^2-1)(4s^2+1)}
\right]
\nonumber \\
&=
(\partial_\theta N)^2
\frac{
4(4s^2-1)+8
}{
(4s^2-1)(4s^2+1)
}
\nonumber \\
&=
(\partial_\theta N)^2
\frac{
4(4s^2+1)
}{
(4s^2-1)(4s^2+1)
}
\nonumber \\
&=
\frac{
4(\partial_\theta N)^2
}{
4s^2-1
} \nonumber \\
&=   \frac{
(\partial_\theta N)^2
}{
N(N+1)
}.
\end{align}
where, we set \(s=N+\frac{1}{2}\) in the last step.

Now, let us compute the third term of Eq. \ref{eq:general_form}. Taking the derivative of the displacement vector gives
\begin{equation}
\Delta \dot X(\theta)
=
\sqrt{2}
\begin{pmatrix}
\mathrm{Re}\,\partial_\theta\beta(\theta)\\
\mathrm{Im}\,\partial_\theta\beta(\theta)
\end{pmatrix}.
\end{equation}
Thus, $\Delta \dot X^T \Delta \dot X
=
2\left[
(\mathrm{Re}\,\partial_\theta\beta)^2
+
(\mathrm{Im}\,\partial_\theta\beta)^2
\right]
=
2|\partial_\theta\beta(\theta)|^2.$ Since $\Sigma^{-1} = [N(\theta)+\tfrac12]^{-1} I_2$, we obtain
\begin{align}
\label{eq:third_term}
\Delta \dot X^T \Sigma^{-1} \Delta \dot X
&=
\frac{1}{N(\theta)+\tfrac12}
\Delta \dot X^T \Delta \dot X
\nonumber\\
&=
\frac{2|\partial_\theta\beta(\theta)|^2}{N(\theta)+\tfrac12}
=
\frac{4|\partial_\theta \beta(\theta)|^2}{2N(\theta)+1}.
\end{align}
Summing both contributions Eq. \ref{eq:third_term} and \ref{eq:first_two_term}, we obtain
\begin{equation}
\label{eq:F_coh}
F_Q(\theta)
=
\frac{4|\partial_\theta \beta(\theta)|^2}{2N(\theta)+1}
+
\frac{[\partial_\theta N(\theta)]^2}{N(\theta)\,[N(\theta)+1]}.
\end{equation}
The first term quantifies sensitivity to displacement (signal attenuation), while the second term captures sensitivity to changes in thermal noise.

As an example, consider a toy state,  $\rho(\theta) = S_{a+b\theta}(c\sqrt{\theta})$ such that $N(\theta) = a  + b\theta$ and $\beta(\theta) = c \sqrt{\theta}$ for which QFI becomes:
\begin{equation}
\begin{aligned}
F_Q(\theta)
&=
\frac{c^2}{\left(2a+2b\theta+1\right)\theta}
+
\frac{b^2}{(a+b\theta)(a+b\theta+1)}.
\end{aligned}
\label{eq:QFI_Toy}
\end{equation}

The quantum Cramér--Rao bound gives the bound on variance as:
\begin{equation}
\mathrm{Var}_\theta(\hat{\theta})
\ge
\frac{1}{M F_Q(\theta)}
.
\label{eq:CRB_final}
\end{equation}
where $M$ is the number of probes. Quantum Cramér--Rao bound is saturated by measurements in the eigenbasis of the symmetric logarithmic derivative (SLD). We find that for the displaced thermal states, the optimal measurement is obtained by a photon
counting after an appropriate displacement of the field:
\begin{equation}
\begin{aligned}
\Pi_n
&=
D\!\left(\beta(\theta)-\frac{A}{\theta}\right)
|n\rangle\langle n|
D^\dagger\!\left(\beta(\theta)-\frac{A}{\theta}\right),\\
A
&=
\frac{2\partial_\theta\beta(\theta)}{2N(\theta)+1}, \quad
\theta
=
\frac{\partial_\theta N(\theta)}{N(\theta)[N(\theta)+1]}.
\end{aligned}
\label{eq:optimal_POVMs}
\end{equation}
The details of the derivation are in Appendix \ref{sec:SLD_optimal}. If \(\partial_\theta N=0\), the SLD reduces to a linear quadrature operator,  and the optimal measurement is homodyne detection along the direction
\(\Delta \dot X\). Conversely, if \(\partial_\theta\beta=0\), the SLD is
purely number-like, and the optimal measurement is photon counting in the displaced frame.

It is crucial to note that the optimal POVM Eq.~\ref{eq:optimal_POVMs} depends on the parameter $\theta$, which we are trying to estimate. To tackle this issue, one can use the Quantum Local Asymptotic Normality (QLAN) approach \cite{Yamagata_2013}, which implements an adaptive, near-optimal measurement that asymptotically saturates the QFI. We have shifted the adaptive discussions to section \ref{sec:finite_sample_adaptive}.

\subsection{QFI of SBS}
\label{subsubsec:derivatives_sbs_output_state}

The quantum Fisher information of SBS is determined by the derivatives of $\beta_i(\Omega;\theta)$ and $N_i(\Omega;\theta)$ with respect to $\theta$. Differentiating  gives
\begin{align}
\label{eq:dbeta_dtheta_general}
\partial_\theta \beta_i(\Omega;\theta)
&=
\eta^{(L-1)/2}\sqrt{E}\,
\partial_\theta \kappa_i(\Omega;\theta)
\nonumber\\
&=
\eta^{(L-1)/2}\sqrt{E}\,
\partial_\theta
\left[
|\widetilde G_i(\theta)|^2
\chi_i(\Omega;\theta)\delta z
\right]. 
\end{align}
Similarly, 
\begin{align}
\label{eq:dN_dtheta_general}
\partial_\theta N_i(\Omega;\theta)
&=
\eta^{L-i}
\partial_\theta
\left[
\left(n_{\mathrm{th}}+\frac{1}{2}\right)|\nu_i|^2
+
\frac{|\mu_i|^2-\eta}{2}
\right].
\end{align}

\begin{figure}[!t]
    \centering
    \includegraphics[width=0.8\linewidth]{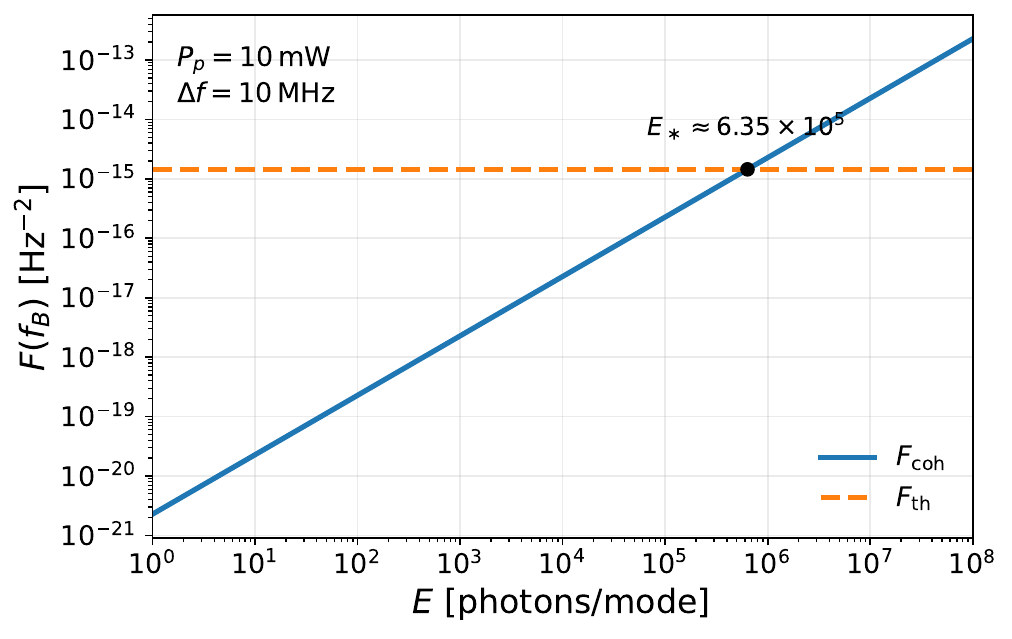}
\caption{Coherent (solid) and thermal (dashed) contributions to the
quantum Fisher information per mode for estimating the Brillouin
frequency $f_B$, as a function of the input probe photon number $E$.
For $P_p=10\,\mathrm{mW}$ and $\Delta f=10\,\mathrm{MHz}$,
$F_{\mathrm{coh}}$ increases linearly with $E$, whereas
$F_{\mathrm{th}}$ remains constant. The contributions become equal
at $E_\ast \approx 6.35\times10^5$ photons per mode, marking the
transition from thermal to coherent information dominance.}
    \label{fig:coherent_thermal}
\end{figure}

The figure \ref{fig:coherent_thermal}  illustrates two ways in which the Brillouin frequency
is encoded in the output state: through changes in the coherent
amplitude $\beta_i$ and in the thermal occupation $N_i$.
Within the fixed-pump model, increasing the probe photon number
enhances the displacement sensitivity because
$\beta_i\propto\sqrt{E}$, giving $F_{\mathrm{coh}}\propto E$,
while $F_{\mathrm{th}}$ remains independent of $E$.
For the chosen operating point, thermal information dominates
below $E_\ast\approx6.35\times10^5$ photons per mode, so ideal
photon-number-resolving detection after local nulling can capture
most of the available frequency information when $E\ll E_\ast$.
Above this crossover, changes in the coherent field become the
dominant information source, and accessing its amplitude and
phase response becomes increasingly important.

\subsubsection{Brillouin-frequency sensing}
\label{subsubsec:brillouin_frequency_sensing}

For Brillouin-frequency sensing, the unknown parameter is $\theta=f_B$, with $\Omega_B=2\pi f_B$, and QFI $F_Q^{(i)}(f_B)$ is: 
\begin{equation}
\label{eq:qfi_brillouin_final}
F_Q^{(i)}(f_B)
=
\frac{
4|\partial_{f_B} \beta_i(\Omega;f_B)|^2
}
{
2N_i(\Omega;f_B)+1
}
+
\frac{
[\partial_{f_B} N_i(\Omega;f_B)]^2
}
{
N_i(\Omega;{f_B} )[N_i(\Omega;f_B)+1]
}.
\end{equation}
The corresponding frequency sensitivity is
\begin{equation}
\label{eq:qcrb_fB}
\Delta f_B
\geq
\frac{1}
{\sqrt{kF_Q^{(i)}(f_B)}}.
\end{equation}
In Figure~\ref{fig:qfi_spectrum_Eprobe} we have plotted $F_Q^{(i)}(f_B)$ as a function of the probe detuning $\Delta f=(\Omega-\Omega_B)/2\pi$ for different coherent probe energies. We find that the largest QFI occurs not at $\Delta = 0$. This is because, while the Lorentzian gain is maximal at exact resonance $\Delta=0$, the slope vanishes. Instead, the largest QFIs are on the sides, where a small shift of $f_B$ produces the largest change in the output state.

\begin{figure}[!t]
    \centering
    \includegraphics[width=0.8\linewidth]{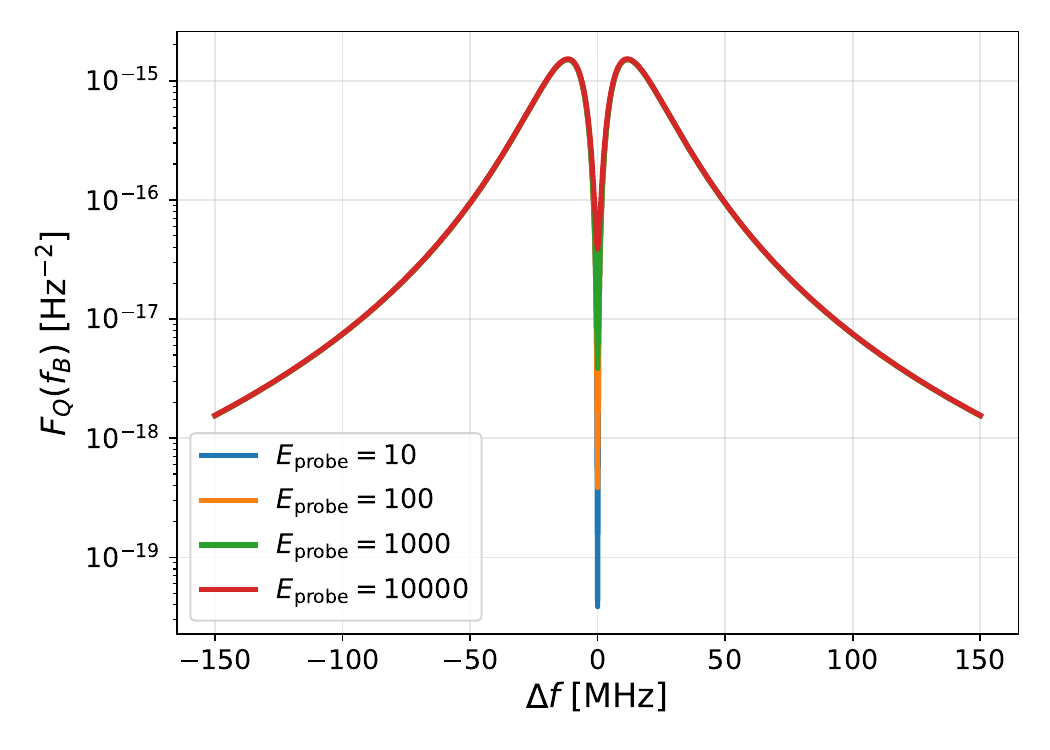}
    \caption{
Quantum Fisher information $F_Q(f_B)$ for estimating the Brillouin frequency 
$f_B=\Omega_B/2\pi$. The parameters are 
$f_B=10.8~{\rm GHz}$, $\Gamma_B/2\pi=30~{\rm MHz}$, 
$P_p=10~{\rm mW}$, $\delta z=0.5~{\rm m}$, $L=100$, $i=50$, 
fiber attenuation $0.2~{\rm dB/km}$, 
$\xi_0=0.31~{\rm W}^{-1}{\rm m}^{-1}$, and $T=300~{\rm K}$ with 
$n_{\rm th}\simeq 5.8\times 10^2$.
}
    \label{fig:qfi_spectrum_Eprobe}
\end{figure}

\begin{figure}[!t]
    \centering
    \includegraphics[width=0.8\linewidth]{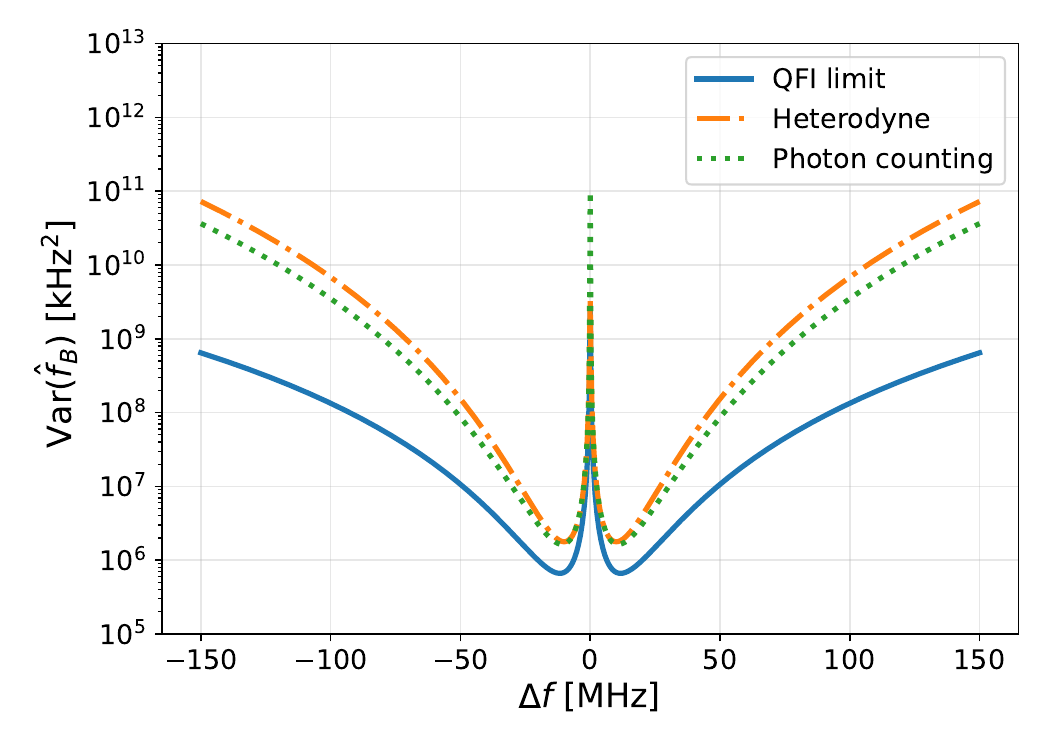}
  \caption{
Receiver-constrained estimator variance  comparing the ultimate quantum limit $F=F_Q$, heterodyne detection $F=F_{\rm het}$, and photon counting $F=F_{\rm pc}$. The parameters are $k=1000$, $E_{\rm probe}=100$, $P_p=10~{\rm mW}$, $f_B=10.8~{\rm GHz}$, $\Gamma_B/2\pi=30~{\rm MHz}$, $\delta z=0.5~{\rm m}$, $L=100$, $i=50$, fiber attenuation $0.2~{\rm dB/km}$, $\xi_0=0.31~{\rm W}^{-1}{\rm m}^{-1}$, and $T=300~{\rm K}$.  
}
    \label{fig:receiver_fB_variance}
\end{figure}

\begin{figure}[!t]
    \centering
    \includegraphics[width=0.8\linewidth]{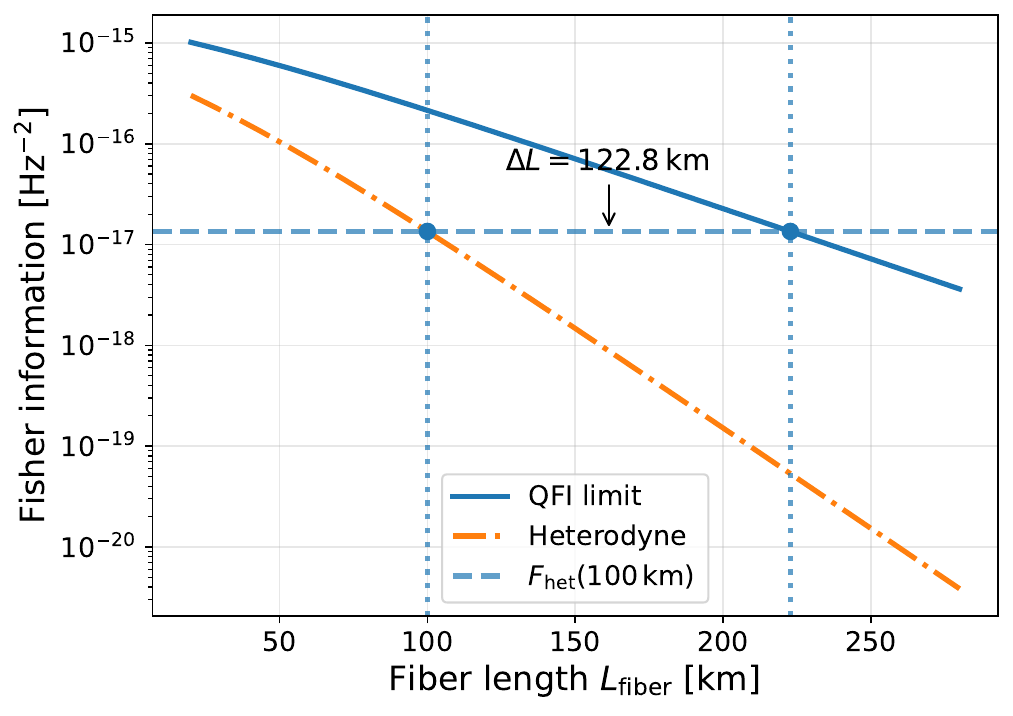}
 \caption{
Quantum sensing reach advantage as a function of the physical fiber length $L_{\rm fiber} = L \delta z$. The quantum Fisher information $F_Q$ and the heterodyne Fisher information $F_{\rm het}$ are evaluated at a fixed operating detuning chosen to maximize $F_{\rm het}$ at the reference length $L_{\rm ref}=100~{\rm km}$. The parameters are $E_{\rm probe}=100$, $P_p=10~{\rm mW}$, $f_B=10.8~{\rm GHz}$, $\Gamma_B/2\pi=30~{\rm MHz}$, $\delta z=0.5~{\rm m}$, fiber attenuation $0.2~{\rm dB/km}$, $\xi_0=0.31~{\rm W}^{-1}{\rm m}^{-1}$, and $T=300~{\rm K}$, with the active SBS segment located at the midpoint of the fiber.
}
    \label{fig:quantum_reach_advantage}
\end{figure}

For the displaced thermal state, the heterodyne outcomes are distributed according to the Husimi \(Q\) function,
\begin{equation}
    q(z|\theta)
    =
    \frac{1}{\pi [N(\theta)+1]}
    \exp\!\left[
        -\frac{|z-\beta(\theta)|^2}{N(\theta)+1}
    \right].
    \label{eq:heterodyne_distribution}
\end{equation}
The corresponding classical Fisher information is
\begin{equation}
    F_{\rm het}(\theta)
    =
    \frac{2|\partial_\theta \beta(\theta)|^2}{N(\theta)+1}
    +
    \frac{[\partial_\theta N(\theta)]^2}{[N(\theta)+1]^2}.
    \label{eq:heterodyne_fisher}
\end{equation}
 Since heterodyne detection adds one unit of vacuum noise, the displacement contribution is smaller than the ultimate quantum Fisher information. We also consider photon counting as a receiver-constrained measurement.  For the state \(S_{N(\theta)}(\beta(\theta))\), this gives the photon-counting probabilities
\begin{align}
    p_n(\theta)
    =
    \frac{N(\theta)^n}{[N(\theta)+1]^{n+1}}
   & \exp\!\left[
        -\frac{|\beta(\theta)|^2}{N(\theta)+1}
    \right]\times \nonumber \\
    &L_n\!\left(
        -\frac{|\beta(\theta)|^2}
        {N(\theta)[N(\theta)+1]}
    \right),
    \label{eq:photon_counting_distribution}
\end{align}
where \(L_n(\cdot)\) is the Laguerre polynomial. The associated classical Fisher information is
\begin{equation}
    F_{\rm pc}(\theta;\alpha_0)
    =
    \sum_{n=0}^{\infty}
    \frac{\left[\partial_\theta p_n(\theta|\alpha_0)\right]^2}
    {p_n(\theta|\alpha_0)} .
    \label{eq:photon_counting_fisher}
\end{equation}

Figure~\ref{fig:receiver_fB_variance} compares the ultimate quantum limit with two receiver-constrained strategies, photon counting and heterodyne. The lower curves belong to better sensitivity. Photon counting has substantially lower variance than heterodyne detection, but it does not saturate the quantum Cramér-Rao bound. These gaps quantify the quantum advantage with an SLD-optimal receiver. Fig.~\ref{fig:quantum_reach_advantage} shows that the quantum optimal receiver achieves the same estimation precision over a substantially longer fiber than heterodyne. Figure~\ref{fig:pump_power_sensitivity} shows the  Brillouin-frequency sensitivity as a function of the SBS pump power. The sensitivity improves rapidly at low pump powers. At larger pump powers, however, the curve flattens, and the sensitivity improvement is small.

\begin{figure}[!t]
    \centering
    \includegraphics[width=0.8\linewidth]{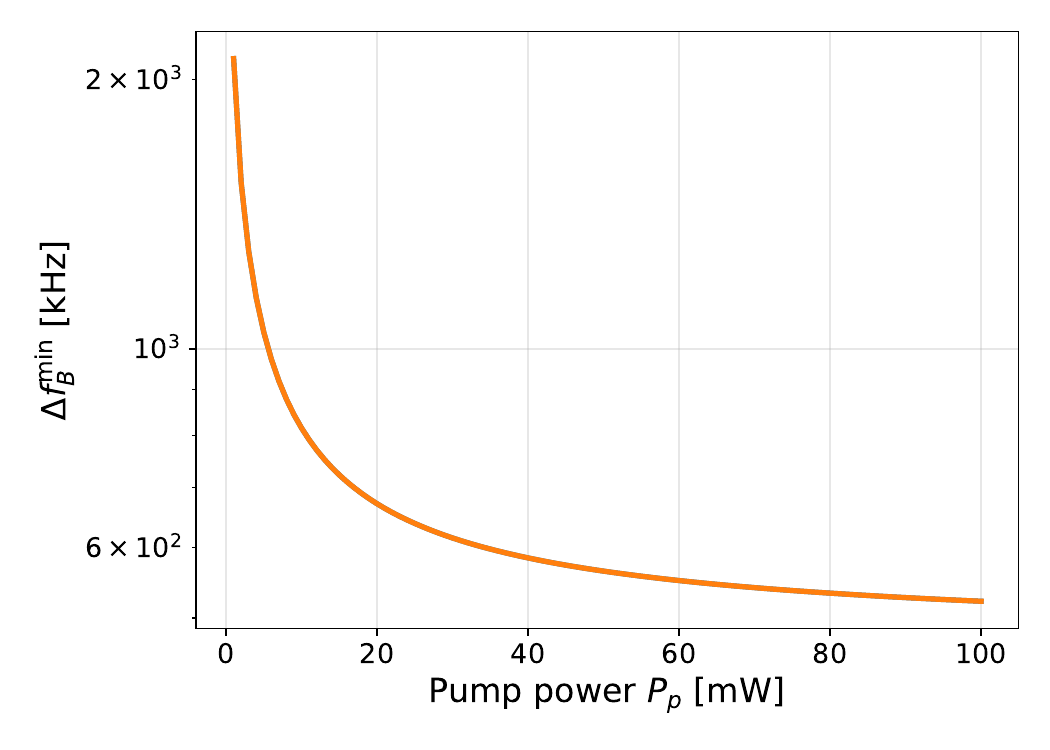}
 \caption{
Minimum achievable Brillouin-frequency uncertainty as a function of SBS pump power. The parameters are $k = 1000$,  $E_{\rm probe}=100$, 
$f_B=10.8~{\rm GHz}$, $\Gamma_B/2\pi=30~{\rm MHz}$,
$\delta z=0.5~{\rm m}$, $L=100$, $i=50$, fiber attenuation
$0.2~{\rm dB/km}$, $\xi_0=0.31~{\rm W}^{-1}{\rm m}^{-1}$, and
$T=300~{\rm K}$.
}
    \label{fig:pump_power_sensitivity}
\end{figure}

\subsubsection{Temperature sensing}
\label{subsubsec:genuine_temperature_sensing}

A simple temperature sensing can be via the following chain rule:
\begin{equation}
\label{eq:simple_temperature_chain_rule}
F_Q(T)
=
F_Q(f_B)
\left(
\frac{\partial f_B}{\partial T}
\right)^2.
\end{equation}
While this conversion is useful, it does not change the shape of the QFI spectrum. To obtain a genuine temperature-dependent SBS model, we allow temperature to simultaneously modify $f_B(T),\Gamma_B(T),\xi_0(T), n_{\mathrm{th}}(T).$ The  temperature QFI is
\begin{equation}
\label{eq:qfi_temperature_final}
F_Q^{(i)}(T)
=
\frac{
4|\partial_T\beta_i(\Omega;T)|^2
}
{
2N_i(\Omega;T)+1
}
+
\frac{
[\partial_T N_i(\Omega;T)]^2
}
{
N_i(\Omega;T)[N_i(\Omega;T)+1]
}.
\end{equation}
For $k$ independent probes, the temperature variance satisfies
\begin{equation}
\label{eq:qcrb_temperature}
\mathrm{Var}(\hat T)
\geq
\frac{1}{kF_Q^{(i)}(T)}.
\end{equation}

\begin{figure}[!t]
    \centering
    \includegraphics[width=0.8\linewidth]{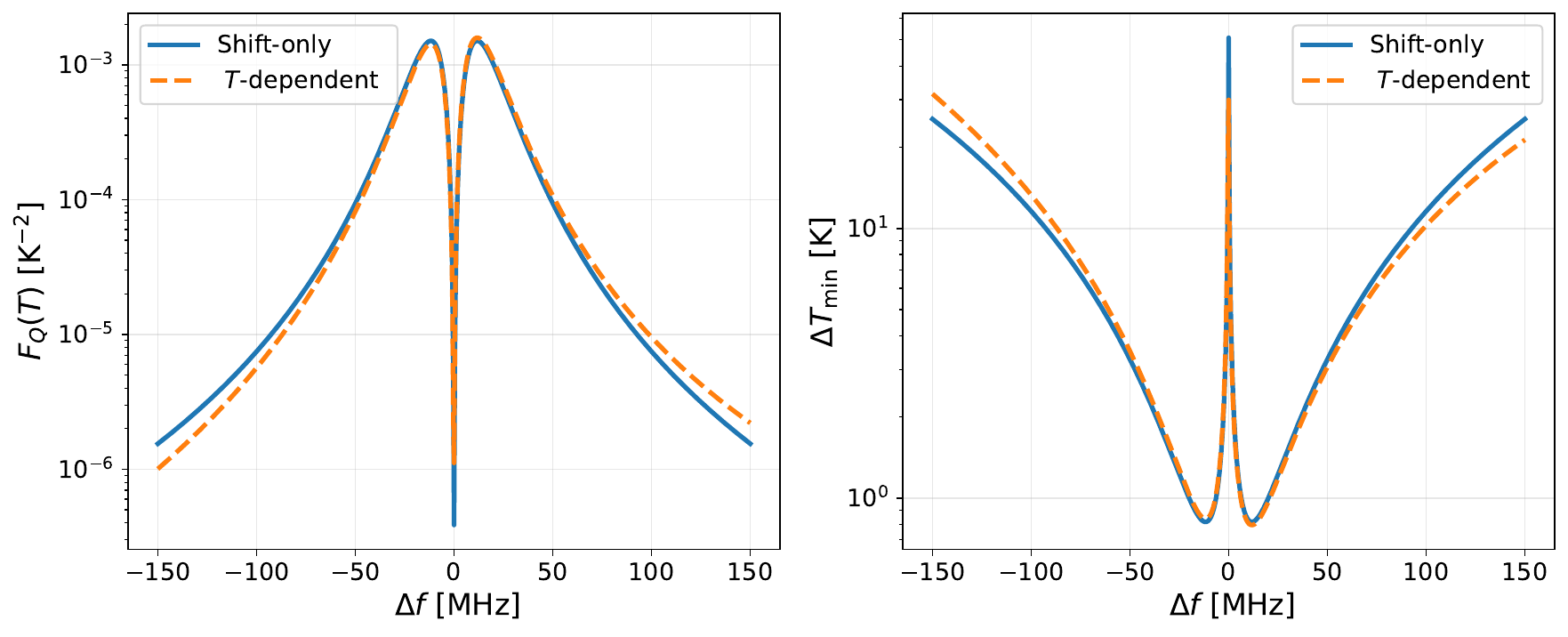}
\caption{
Temperature-sensing performance for stimulated Brillouin scattering. 
(a) Temperature quantum Fisher information $F_Q(T)$ as a function of detuning for a simple frequency-shift model and a genuine 
temperature-dependent SBS model. 
(b) Corresponding Cram\'er--Rao temperature sensitivity for $k=1000$ independent probe shots. 
}
    \label{fig:temperature_sensing}
\end{figure}

In figure~\ref{fig:temperature_sensing}, we have compared the temperature QFI obtained from a simple frequency-shift chain-rule with that of a genuine temperature-dependent SBS channel. These two QFI spectra have identical shapes; there are small quantitative corrections. This shows that the dominant contribution to temperature sensitivity comes from the Brillouin-frequency shift, and that additional parameters such as linewidth, gain, and thermal poise provide only small corrections. 
\section{Adaptive measurement of probes}
\label{sec:finite_sample_adaptive}

As discussed in section \ref{sec:qfi_sbs_parameter_sensing}, the optimal POVMs Eq.~\ref{eq:optimal_POVMs} depend on the parameter $\theta$, which we are trying to estimate. 
Inspired by the QLAN/adaptive-estimation framework \cite{Yamagata_2013} , we use a two-stage local measurement protocol to address this issue. The first stage is the preliminary stage, where we obtain a rough estimate $\hat{\theta}_p$. Conditioned on $\hat{\theta}_p=\theta_0$, in the refinement stage, we plug this value into the optimal POVM. 

Let $n$ be the total number of copies of the probes $\rho_{\theta_t}$, where $\theta_t$ is the true value of the parameter. We use $f(n)$ copies in the preliminary stage, such that $f(n)\in \omega(1)\cap o(n)$. This means that $f(n)\to\infty$, so that the preliminary estimate becomes consistent, but $f(n)/n\to0$, so that only a vanishing fraction of the samples is used for preliminary localization. The number of refinement-stage copies is then $m_n=n-f(n)$, and $\frac{m_n}{n}\to1.$ A convenient balanced choice is $f(n)=\sqrt n$.

As an example, suppose the output displaced thermal state is $\rho(\theta) = S_{a+b\theta}(c\sqrt{\theta})$ such that $N(\theta) = a  + b\theta$ and $\beta(\theta) = c \sqrt{\theta}$. 
A simple way to obtain an initial estimate is via heterodyne detection on a small number of probe copies. The heterodyne POVM is $ E_z
    =
    \frac{1}{\pi}|z\rangle\langle z|,
    z\in\mathbb C,$
where $|z\rangle$ is a coherent state. For the example state, the sample mean of the heterodyne outcomes gives a rough estimate of the displacement as:
\begin{equation}
    \theta_0
    =
    \frac{|\overline z|^2}{|c|^2}.
    \label{eq:initial_lambda_heterodyne}
\end{equation}
Given  $\theta_0$, the necessary displacement amplitude is
\begin{equation}
    \gamma_{\theta_0}
    =
    \beta(\theta_0)-\frac{A(\theta_0)}{B(\theta_0)}  =  c\sqrt{\theta_0}
    -
    \frac{
        c(a+b\theta_0)(a+b\theta_0+1)
    }{
        b\sqrt{\theta_0}(2a+2b\theta_0+1)
    } .
    \label{eq:adaptive_gamma_lambda0}
\end{equation}
In lab, one has to apply the inverse displacement, $\alpha_{\mathrm{disp}}^{(0)}
    =
    -\gamma_{\theta_0}$, which is achieved using a highly transmissive beam splitter with reflectivity $R$ and a coherent local oscillator of amplitude $\beta_{\mathrm{LO}}^{(0)}$. The calibration condition is $\beta_{\mathrm{LO}}^{(0)}
    =
    -\frac{\gamma_{\theta_0}}{\sqrt{R}} .$
After applying $D(-\gamma_{\theta_0})$, the state entering the photon-number-resolving detector is $\rho_{\mathrm{out}}(\theta;\theta_0)
    =
    D(-\gamma_{\theta_0})
    \rho(\theta)
    D^\dagger(-\gamma_{\theta_0})$.  Each experimental run produces one photon-number outcome $  n_j\in\{0,1,2,\ldots\}.$
After $M_0$ repetitions at the setting $\theta_0$, the data record is $  \mathcal D_0
    =
    \{n^{(0)}_1,n^{(0)}_2,\ldots,n^{(0)}_{M_0}\}.$ With this, one obtains the estimate $\theta_1$.

\subsection{First-order asymptotic optimality}
We will now show that the adaptive approach achieves first-order asymptotic optimality. First, let us show that the heterodyne preliminary estimator is strongly consistent. From Eq.~\eqref{eq:heterodyne_distribution}, the heterodyne outcomes have mean $\mathbb{E}_{\theta_t}[Z_j]
=
c\sqrt{\theta_t}.$
Since $f(n)\to\infty$, the strong law of large numbers gives
\begin{equation}
\bar Z_{f(n)}
=
\frac{1}{f(n)}
\sum_{j=1}^{f(n)} Z_j
\xrightarrow{\mathrm{a.s.}}
c\sqrt{\theta_t}.
\label{eq:proof_sample_mean_convergence}
\end{equation}
The map $z\mapsto |z|^2/c^2$ is continuous for $c\neq0$. Therefore, $\hat\theta_p
=
\frac{|\bar Z_{f(n)}|^2}{c^2}
\xrightarrow{\mathrm{a.s.}}
\frac{|c\sqrt{\theta_t}|^2}{c^2}
=
\theta_t.$ Hence $\hat\theta_p$ is a strongly consistent preliminary estimator. 

Let $F_{\rm cl}(\theta;\theta_0)$ be the classical Fisher information of the photon-counting distribution when the true parameter is $\theta$, but the measurement is tuned at parameter $\theta_0$. Since the POVM $\left\{
\Pi_q^{(\theta_t)}
=
D(\gamma_{\theta_t}) |q\rangle\langle q|D^\dagger(\gamma_{\theta_t})
\right\}_{q=0}^{\infty}$
is the spectral measure of the SLD operator, the standard attainability theorem implies $F_{\rm cl}(\theta_t;\theta_t)
=
F_Q(\theta_t).$

In the actual protocol, since we do not know $\theta_t$ a priori, we calibrate the receiver using $\hat\theta_p$. Since $\hat\theta_p\xrightarrow{\mathrm{a.s.}}\theta_t$ and $\gamma_\theta$ is continuous in a neighbourhood of $\theta_t$, the continuous mapping theorem gives $\gamma_{\hat\theta_p}
\xrightarrow{\mathrm{a.s.}}
\gamma_{\theta_t}.$
Consequently,
\begin{equation}
\Pi_q^{(\hat\theta_p)}
\longrightarrow
\Pi_q^{(\theta_t)},
\quad
F_{\rm cl}(\theta_t;\hat\theta_p)
\longrightarrow
F_{\rm cl}(\theta_t;\theta_t)
=
F_Q(\theta_t).
\end{equation}

Let $\hat\theta_r(\theta_0)$ denote the refinement-stage maximum-likelihood estimator when the receiver is calibrated at $\theta_0$. For every fixed $\theta_0$ sufficiently close to $\theta_t$, the regularity assumptions imply the standard maximum-likelihood asymptotic normality relation $\sqrt{m_n}
\left(
\hat\theta_r(\theta_0)-\theta_t
\right)
\xrightarrow{d}
\mathcal{N}
\left(
0,
F_{\rm cl}(\theta_t;\theta_0)^{-1}
\right).$
Taking $\theta_0=\hat\theta_p$, Slutsky's theorem gives
\begin{equation}
\sqrt{m_n}
\left(
\hat\theta_r-\theta_t
\right)
\xrightarrow{d}
\mathcal{N}
\left(
0,
F_Q(\theta_t)^{-1}
\right).
\label{eq:m_n_asymptotic_normality}
\end{equation}

Since $m_n/n\to1$, Eq.~\eqref{eq:m_n_asymptotic_normality} is equivalently $\sqrt{n}
\left(
\hat\theta_r-\theta_t
\right)
\xrightarrow{d}
\mathcal{N}
\left(
0,
F_Q(\theta_t)^{-1}
\right).$
Thus the refinement estimator attains the QFI at first order,
\begin{equation}
\operatorname{Var}(\hat\theta_r)
=
\frac{1}{nF_Q(\theta_t)}
+
o\!\left(\frac{1}{n}\right).
\label{eq:variance_final_proof}
\end{equation}

\subsection{Second-order correction from preliminary miscalibration}

Now, we will consider the second-order correction due to preliminary miscalibration. For $\theta_0$ close to $\theta_t$, using Taylor series, we expand
\begin{equation}
F_{\rm cl}(\theta_t;\theta_0)
=
F_Q(\theta_t)
-
C(\theta_t)(\theta_0-\theta_t)^2
+
o\!\left((\theta_0-\theta_t)^2\right),
\label{eq:Fcl_calibration_expansion}
\end{equation}
where $C(\theta_t)
=
-\frac{1}{2}
\left.
\frac{\partial^2}{\partial \theta_0^2}
F_{\rm cl}(\theta_t;\theta_0)
\right|_{\theta_0=\theta_t}$ is a measure of calibration error. A small value means that the receiver is robust to miscalibration, while a large value means that a small calibration error can cause a significant loss of Fisher information. For fixed $\theta_0$, the refinement-stage variance satisfies
\begin{align}
&\operatorname{Var}(\hat{\theta}_r|\theta_0)
=
\frac{1}
{m_n F_{\rm cl}(\theta_t;\theta_0)}
+
o\!\left(\frac{1}{m_n}\right)
\nonumber\\
&=
\frac{1}
{m_n\left[
F_Q(\theta_t)
-
C(\theta_t)(\theta_0-\theta_t)^2
+
o((\theta_0-\theta_t)^2)
\right]} \nonumber \\
& \qquad+
o\!\left(\frac{1}{m_n}\right)
\nonumber\\
&=
\frac{1}
{m_nF_Q(\theta_t)}
+
\frac{
C(\theta_t)(\theta_0-\theta_t)^2
}{
m_nF_Q(\theta_t)^2
}
+
o\!\left(
\frac{(\theta_0-\theta_t)^2}{m_n}
\right).
\label{eq:conditional_variance_fixed_lambda0}
\end{align}
Thus, the second-order correction due to imperfect preliminary calibration is controlled by the mean-square error of the preliminary estimator. Equation~\eqref{eq:conditional_variance_fixed_lambda0} gives the conditional variance of the refinement-stage estimator when the calibration value is fixed at $\theta_0=\hat\theta_p$. The unconditional variance is obtained from the law of total variance,
\begin{equation}
\operatorname{Var}(\hat\theta_r)
=
\mathbb{E}_{\hat\theta_p}\!\left[
\operatorname{Var}\!\left(\hat\theta_r\,\big|\,\hat\theta_p\right)
\right]
+
\operatorname{Var}_{\hat\theta_p}\!\left(
\mathbb{E}\!\left[
\hat\theta_r\,\big|\,\hat\theta_p
\right]
\right).
\end{equation}
Under the standard regularity assumptions for the refinement-stage maximum-likelihood estimator, the conditional bias contributes only at higher order. Therefore, to the order considered here, the second term may be neglected. Averaging Eq.~\eqref{eq:conditional_variance_fixed_lambda0} over the preliminary estimator then yields
\begin{align}
\operatorname{Var}(\hat\theta_r)
=
\frac{1}{m_nF_Q(\theta_t)}
&+
\frac{C(\theta_t)}
{m_nF_Q(\theta_t)^2}
\mathbb{E}\!\left[
(\hat\theta_p-\theta_t)^2
\right] \nonumber \\ 
&+ 
o\!\left(
\frac{
\mathbb{E}[(\hat\theta_p-\theta_t)^2]
}{m_n}
\right).
\label{eq:Var_lamda_r}
\end{align}
The quantity $\mathbb{E}[(\hat{\theta}_p-\theta_t)^2]$ depends on the preliminary measurement, while the first term is independent of it.

For the displaced thermal state $S_{N_t}(\beta_t)$, the heterodyne outcome is
distributed as $Z\sim \mathcal{CN}(\beta_t,N_t+1)$.   With $f(n)$ independent heterodyne samples, the sample mean satisfies $\bar Z
=
\frac{1}{f(n)}
\sum_{j=1}^{f(n)}Z_j
\sim
\mathcal{CN}\!\left(
\beta_t,
\frac{N_t+1}{f(n)}
\right).$
The estimate of $\theta_t$ depends on the specific encoding map $\beta(\theta)$. For the toy model, $\beta(\theta)=c\sqrt{\theta}$ and $N(\theta)=a+b\theta,$ we define the preliminary estimator $\hat\theta_p
=
\frac{|\bar Z|^2}{c^2}.$ Let $\sigma_n^2=\frac{N_t+1}{f(n)}$. For $X\sim\mathcal{CN}(\mu,\sigma^2)$, $\mathbb{E}[|X|^2]=|\mu|^2+\sigma^2,
\mathbb{E}[|X|^4]
=
|\mu|^4+4|\mu|^2\sigma^2+2\sigma^4 .$
Therefore,
\begin{align}
\mathbb{E}\!\left[(\hat\theta_p-\theta_t)^2\right]
&=
\frac{1}{c^4}
\mathbb{E}\!\left[
\left(|\bar Z|^2-c^2\theta_t\right)^2
\right] \nonumber\\
&=
\frac{1}{c^4}
\left[
\mathbb{E}[|\bar Z|^4]
-
2c^2\theta_t\,\mathbb{E}[|\bar Z|^2]
+
c^4\theta_t^2
\right] \nonumber\\
&=
\frac{2\theta_t\sigma_n^2}{c^2}
+
\frac{2\sigma_n^4}{c^4} \nonumber\\
&=
\frac{2\theta_t(N_t+1)}
{c^2 f(n)}
+
\frac{2(N_t+1)^2}
{c^4 f(n)^2}.
\end{align}
For $f(n)= \sqrt{n}$, we have the expansion to the leading order of the form:
\begin{equation}
 \mathbb{E}\!\left[(\hat\theta_p-\theta_t)^2\right]
    =
    \frac{2\theta_t(N_t+1)}{c^2\sqrt n}
    +
    O(n^{-1}).
\end{equation}
Furthermore, $\frac{1}{m_n}
    =
    \frac{1}{n-\sqrt n}
    =
    \frac{1}{n}
    +
    \frac{1}{n^{3/2}}
    +
    O(n^{-2})$. Substituting these into Eq.~\eqref{eq:Var_lamda_r}, we obtain
\begin{align}
    \operatorname{Var}(\hat\theta_r)
    &=
    \frac{1}{nF_Q(\theta_t)}\nonumber\\
    &\quad+ \left(  \frac{1}{F_Q(\theta_t)} 
    +
    \frac{
    2\theta_t(N_t+1)C(\theta_t)
    }{
    c^2F_Q(\theta_t)^2
    }\right)
    n^{-3/2}\nonumber \\
    &+
    o(n^{-3/2}),
    \label{eq:MSE_one_shot_final}
\end{align}

So far, we have focused on a two-stage approach, where the samples are divided into a single calibration step followed by a refinement stage. We now consider a feedback-adaptive protocol in which the refinement data itself are divided into $R$ batches, $m_1+m_2+\cdots+m_R=m_n=n-f(n)$. After each batch, the displacement is recalibrated using the current estimate. Let, $e_{r-1}=\hat{\theta}_{r-1}-\theta_t$ be the calibration error at batch $r$, and the corresponding classical Fisher Information and $F_{\rm cl}(\theta_t;\hat{\theta}_{r-1})=
F_Q(\theta_t)-C(\theta_t)e_{r-1}^{2}+O(e_{r-1}^{3}).$ The total accumulated information is
\begin{equation}
I_{\rm batch}
=
m_nF_Q(\theta_t)
-
C(\theta_t)\sum_{r=1}^{R}m_r e_{r-1}^2
+
o\!\left(\sum_{r=1}^{R}m_r e_{r-1}^2\right).
\end{equation}
For equal batches and $f(n)=\sqrt n$, the leading correction becomes
\begin{align}
\operatorname{Var}&(\hat\theta_r)_{\rm batch}
=
\frac{1}{nF_Q(\theta_t)}
+ \nonumber \\
&\left[
\frac{1}{F_Q(\theta_t)}
+
\frac{1}{R}
\frac{2\theta_t(N_t+1)C(\theta_t)}
{c^2F_Q(\theta_t)^2}
\right]n^{-3/2}
+
o(n^{-3/2}),
\end{align}
For $R=1$, this reduces to the two-stage approach. One could consider a fully online limit, where the displacement is updated after every $j$-th photon-counting measurement outcome $n_j$ by maximizing the accumulated likelihood $\ell_j(\theta)=\ell_{j-1}(\theta)+\log p_{n_j}(\theta;\hat\theta_{j-1})$. If the online estimator is asymptotically efficient, then $\mathbb{E}[e_j^2]\simeq [jF_Q(\theta_t)]^{-1}$, and via Harmonic series sum expansion, we obtain $\sum_{j=1}^{n}\mathbb{E}[e_j^2]\simeq \log n/F_Q(\theta_t)$. Consequently, to the leading order,
\begin{equation}
\operatorname{Var}(\hat\theta_r)_{\rm online}
=
\frac{1}{nF_Q(\theta_t)}
+
\frac{C(\theta_t)\log n}
{n^2F_Q(\theta_t)^3}
+
o\!\left(\frac{\log n}{n^2}\right).
\end{equation}

Thus, feedback does not change the leading QCRB term, but it can reduce the finite-sample calibration penalty. Batch feedback suppresses the calibration part of the $n^{-3/2}$ correction by a factor $1/R$, while online feedback is expected to reduce the calibration penalty to order $(\log n)/n^2$.

\begin{figure}[h]
    \centering
    \includegraphics[width=\linewidth]{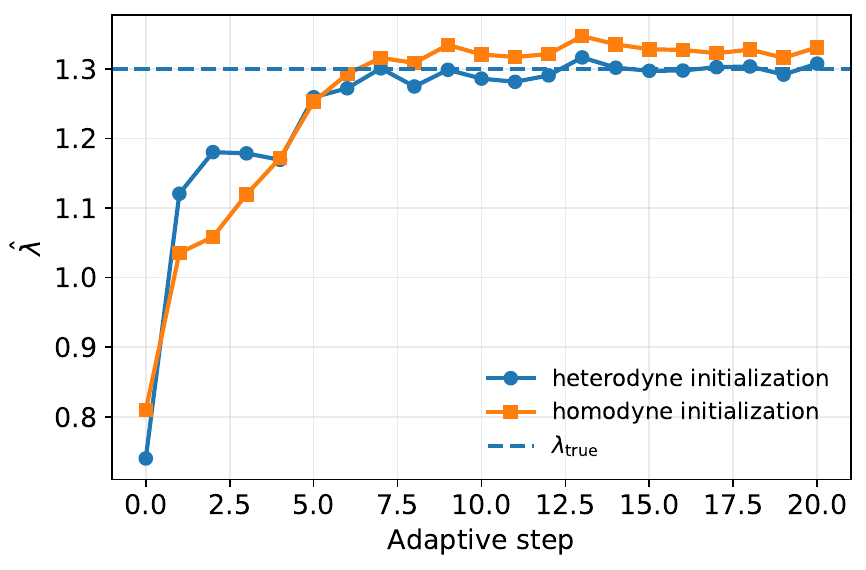}
   \caption{
Adaptive estimation of $\theta$ for the displaced thermal-state model with $N(\theta)=a+b\theta$ and $\beta(\theta)=c\sqrt{\theta}$, using parameters $\theta_{\mathrm{true}}=1.3$, $a=0.2$, $b=0.5$, and $c=1.0$. The protocol uses $M_{\mathrm{init}}=50$ heterodyne or homodyne samples for initialization, followed by displaced photon counting with $M_{\mathrm{per\ round}}=100$ samples per adaptive round for $K=20$ rounds, with maximum-likelihood optimization over $\theta\in[10^{-4},5.0]$. The dashed line denotes the true value $\theta_{\mathrm{true}}$, and both initialization methods converge toward it under adaptive updates.
}
    \label{fig:adaptive step}
\end{figure}

The figure \ref{fig:adaptive step} shows the adaptive estimation protocol via either heterodyne or homodyne initial rough estimation. As per the QLAN philosophy, the asymptotic limit is roughly independent of the initial estimation method. For this simple toy model, both estimators converge quickly to the true parameter.

\section{Application in Eavesdropper detection}
\label{sec:application1}

Given the channel \eqref{eqn:fiber-channel}, we formulate the task of detecting whether, at a given segment $i$, the channel parameters deviate from their nominal values due to an eavesdropping attack. The input states are coherent states, which are transformed by the Gaussian channel $(X,Y)$ into displaced thermal states.

An attack could directly influence the SBS parameters, but since
such changes were studied in the previous section, we consider
an evanescent-coupling attack with the SBS coefficients held fixed.
We model the attack as a beam-splitter interaction immediately
after the SBS interaction in segment $i$, with transmissivity
$\tau_E$. The auxiliary optical input is a thermal state with
mean photon number $n_E$ and zero displacement. The output state
is $S_N(\beta)$, with
\begin{align}
N
&=
\eta^{L-i}
\Bigg[
\tau_E
\left(n_{\mathrm{th}}+\tfrac{1}{2}\right)|\nu_i|^2
\nonumber\\
&\qquad
+\frac{\tau_E}{2}
\left(
\left|\sqrt{\eta}+\kappa_i\right|^2-\eta
\right)
+(1-\tau_E)n_E
\Bigg],
\label{eq:disturbed_parameters}
\\
\beta
&=
\eta^{(L-1)/2}
\sqrt{\tau_E}
\left(\sqrt{\eta}+\kappa_i\right)\sqrt{E}.
\end{align}
Here, $1-\tau_E$ is the fraction of incident optical power
coupled out at the tap. For a passive attack ($n_E=0$), the
auxiliary input is vacuum: the tap attenuates the coherent
field and the existing SBS thermal occupation without adding
excess thermal noise. For $n_E>0$, the injected thermal noise
contributes $\eta^{L-i}(1-\tau_E)n_E$ photons per output mode.
Setting $\tau_E=1$ recovers the unattacked channel, while
$\kappa_i=\nu_i=n_E=0$ gives $N=0$, as required for a coherent
input transmitted through passive vacuum-loss channels.

We can model the intrusion-detection problem from a local parameter-estimation perspective. We treat the attack strength
\(\rho = 1 - \tau_E\) as an unknown scalar parameter and quantify its estimability via the quantum Fisher information (QFI). In Figure~\ref{fig:rho_variance}, we compare the precision for different receiver constraints on weak-probe limit, $E_{\text{probe}}= 5$. The ultimate bound is from the quantum bound  $\mathrm{Var}(\hat\rho)$, while the photon counting and heterodyne are other experimentally constrained measurements. For this weak-probe limit, Photon counting is close to the QFI limit, whereas heterodyne detection is less efficient.

\begin{figure}[!t]
    \centering
    \includegraphics[width=\linewidth]{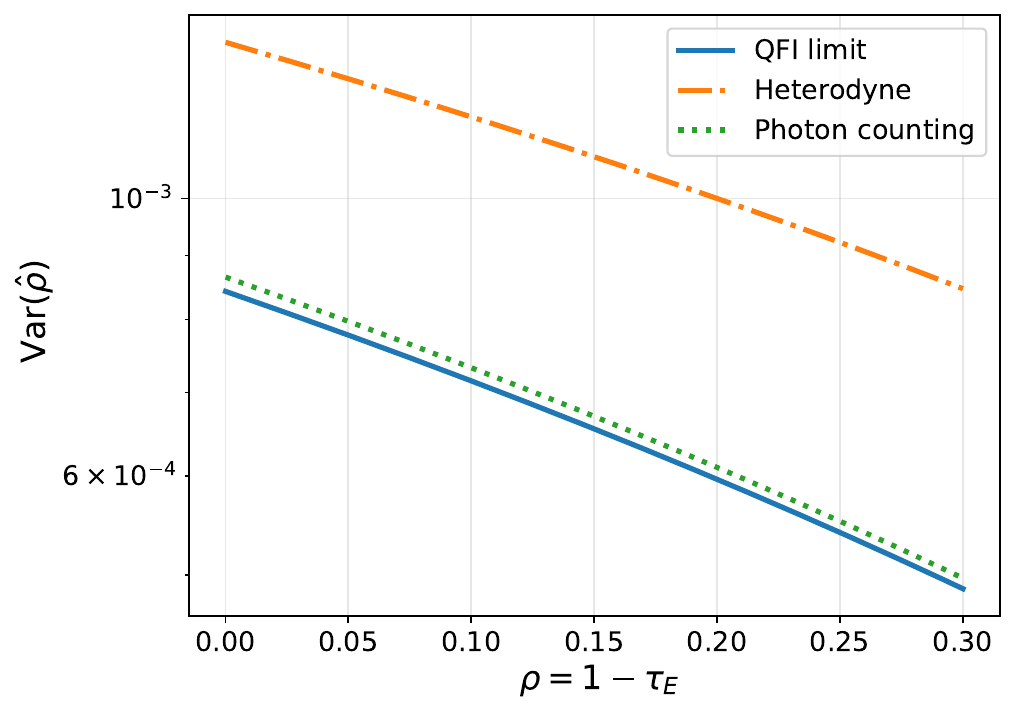}
  \caption{
Receiver-constrained estimation of the tapping strength comparing the quantum Fisher information limit, heterodyne detection, and photon counting for the passive attack, $n_E = 0$, $k = 1000$. The parameters are a $1550~{\rm nm}$ optical probe, fiber attenuation $0.2~{\rm dB/km}$, $\delta z=0.5~{\rm m}$, $L=100$ , $i=50$,  $f_B=10.8~{\rm GHz}$, $\Gamma_B/2\pi=30~{\rm MHz}$,  $\xi_0=0.31~{\rm W}^{-1}{\rm m}^{-1}$, $T=300~{\rm K}$, $P_p=10~{\rm mW}$, $E_{\rm probe}=5$.
}
    \label{fig:rho_variance}
\end{figure}

\section{Conclusion}

In this work, we developed a quantum-information-theoretic framework for stimulated Brillouin scattering-based sensing in optical fibers. We formulated each fiber segment as a single-mode bosonic Gaussian channel acting on the Stokes field. The Gaussian channel thus includes the local parameters such as ordinary fiber attenuation, SBS-induced gain, dispersive phase response, and phonon-induced thermal noise. Applying this framework to Brillouin-frequency and temperature sensing showed that the largest information is not obtained at the maximum of the Brillouin gain profile, but rather near the steepest parts of the resonance, where small parameter changes produce the largest variation in the output quantum state. We also compared the ultimate quantum limit with receiver-constrained strategies such as heterodyne detection and photon counting. In addition, we proposed an adaptive approach that achieves first-order asymptotic optimality and provides a second-order correction to address preliminary miscalibration. Finally, we applied this to the fiber-tapping problem and provided a quantum-information-theoretic benchmark for SBS-based intrusion detection in optical fibers.

Several directions remain open. First, the present analysis can be extended to a multiparameter estimation problem in which parameters such as temperature, strain, and Brillouin frequency can vary simultaneously. Techniques from modern quantum computing could be exported for this \cite{adhikari2026certifiedfidelitysusceptibilityclassical,adhikari2025quantumrampsecretsharing}. Second, the present work focused on coherent-state probes. It would be interesting to see if nonclassical Gaussian probes, such as squeezed or entangled states \cite{Adhikari_2024}, can improve SBS sensing. Third, further work is needed to sharpen the information-theoretic definition of the intrusion detection task and to assess the suitability of different probe signals. Possible methods include the quickest change-point detection \cite{Fanizza_2023}, or compound-hypothesis testing approaches \cite{Lami_2025} to relax the assumption of binary system states, thereby fixing the eavesdropper's impact on the system to a single parameter value, rather than a (more realistic) continuum. Finally, the single-active-segment model can be generalized to distributed sensing with many correlated fiber segments and to the investigation of engineered multimode structures that may provide additional metrological information.

\section*{Acknowledgements}
This work was financed by the Federal Ministry of Education and Research of Germany in the Q-STARS project, grant number 16KIS2604, as well as via grants 16KISQ093, 16KISQ039 and 16KISQ077. The generous support of the state of Bavaria via the 6GQT project is greatly appreciated. The research is part of the Munich Quantum Valley, which is supported by the Bavarian state government with funds from the Hightech Agenda Bayern Plus. Further support was received from the DFG via grant NO 1129/2-1.

\bibliographystyle{IEEEtran}
\bibliography{bib}

\clearpage
\onecolumn
\appendices

\section{Beam-splitter representation of distributed attenuation}
\label{sec:beam-splitter}

We consider the propagation of the frequency-resolved Stokes annihilation operator
\(a_s(z,\Omega)\), where \(z\) denotes the longitudinal propagation coordinate and
\(\Omega\) is the frequency offset from the Stokes carrier. The \(i\)th fiber
segment occupies the interval $z\in [z_i,z_i+\delta z],$ and is assumed to have a constant intensity attenuation coefficient
\(\alpha_{s,i}\) over this interval. Since \(\alpha_{s,i}\) is an intensity loss
coefficient, the corresponding deterministic field-amplitude damping rate is
\(\alpha_{s,i}/2\). The attenuation-only quantum Langevin equation for the Stokes field is
\begin{equation}
\frac{\partial a_s(z,\Omega)}{\partial z}
=
-\frac{\alpha_{s,i}}{2}a_s(z,\Omega)
+
\sqrt{\alpha_{s,i}}\,v(z,\Omega),
\label{eq:langevin_attenuation}
\end{equation}
where \(v(z,\Omega)\) is a continuum reservoir annihilation operator describing
vacuum fluctuations coupled into the optical mode by loss.  The reservoir field is assumed to be Markovian along the propagation coordinate and satisfies $[v(z,\Omega),v^\dagger(z',\Omega')]
=
\delta(z-z')\delta(\Omega-\Omega').$
The input Stokes field satisfies the canonical frequency-domain commutator $[a_s(z_i,\Omega),a_s^\dagger(z_i,\Omega')]
=
\delta(\Omega-\Omega').$
Furthermore, the input field and the reservoir field act on independent Hilbert
spaces and commute with each other.

For compactness, let us define $\alpha\equiv \alpha_{s,i}.$ Equation~\eqref{eq:langevin_attenuation} becomes
\begin{align}
\frac{\partial a_s(z,\Omega)}{\partial z}
+
\frac{\alpha}{2}a_s(z,\Omega)
&=
\sqrt{\alpha}\,v(z,\Omega),
 \nonumber \\
e^{\alpha z/2}
\frac{\partial a_s(z,\Omega)}{\partial z}
+
\frac{\alpha}{2}
e^{\alpha z/2}a_s(z,\Omega)
&=
\sqrt{\alpha}\,
e^{\alpha z/2}v(z,\Omega),
 \nonumber \\
\frac{\partial}{\partial z}
\left[
e^{\alpha z/2}a_s(z,\Omega)
\right]
&=
\sqrt{\alpha}\,
e^{\alpha z/2}v(z,\Omega).
\label{eq:langevin_total_derivative}
\end{align}
where we multiplied by \(e^{\alpha z/2}\) in the second step, and used the fact that the left-hand side is a total derivative in the third step.  We integrate Eq.~\eqref{eq:langevin_total_derivative} over the segment
\([z_i,z_i+\delta z]\):
\begin{align}
\int_{z_i}^{z_i+\delta z} dz\,
\frac{\partial}{\partial z}
\left[
e^{\alpha z/2}a_s(z,\Omega)
\right]
&=
\sqrt{\alpha}
\int_{z_i}^{z_i+\delta z} dz\,
e^{\alpha z/2}v(z,\Omega) \nonumber \\
e^{\alpha(z_i+\delta z)/2}
a_s(z_i+\delta z,\Omega)
-
e^{\alpha z_i/2}a_s(z_i,\Omega) &= \sqrt{\alpha}
\int_{z_i}^{z_i+\delta z} dz\,
e^{\alpha z/2}v(z,\Omega)
\label{eq:integrated_total_derivative}
\end{align}
where we used the fundamental theorem of calculus in the second step. Solving for \(a_s(z_i+\delta z,\Omega)\) gives the exact finite-segment input-output relation
\begin{equation}
a_s(z_i+\delta z,\Omega)
=
e^{-\alpha_{s,i}\delta z/2}a_s(z_i,\Omega)
+
\sqrt{\alpha_{s,i}}
\int_{z_i}^{z_i+\delta z} dz\,
e^{-\alpha_{s,i}(z_i+\delta z-z)/2}v(z,\Omega).
\label{eq:finite_segment_langevin_solution}
\end{equation}
The first term in Eq.~\eqref{eq:finite_segment_langevin_solution} is the
deterministically attenuated input field. The second term is the accumulated
reservoir contribution. A reservoir fluctuation injected at position \(z\) must
still propagate over the remaining distance \(z_i+\delta z-z\), and is therefore
attenuated by the amplitude factor
\(e^{-\alpha_{s,i}(z_i+\delta z-z)/2}\).

\subsection{Requirement of the Langevin reservoir term}
\label{subsec:deterministic_loss_failure}

The Langevin noise term $ v(z,\Omega)$ is required to preserve the commutation relation. If the reservoir term were omitted, the propagation equation would reduce to
\begin{equation}
\frac{\partial a_s(z,\Omega)}{\partial z}
=
-\frac{\alpha_{s,i}}{2}a_s(z,\Omega).
\label{eq:deterministic_attenuation_equation}
\end{equation}
Integrating Eq.~\eqref{eq:deterministic_attenuation_equation} from \(z_i\) to
\(z_i+\delta z\) gives $a_s(z_i+\delta z,\Omega)
=
e^{-\alpha_{s,i}\delta z/2}a_s(z_i,\Omega).$ However, the solution doesn't preserve the canonical bosonic algebra:

\begin{align}
[a_s(z_i+\delta z,\Omega),
a_s^\dagger(z_i+\delta z,\Omega')]
&=
\left[
e^{-\alpha_{s,i}\delta z/2}a_s(z_i,\Omega),
e^{-\alpha_{s,i}\delta z/2}a_s^\dagger(z_i,\Omega')
\right]
\nonumber\\
&=
e^{-\alpha_{s,i}\delta z}
[a_s(z_i,\Omega),a_s^\dagger(z_i,\Omega')]
\nonumber\\
&=
e^{-\alpha_{s,i}\delta z}\delta(\Omega-\Omega').
\label{eq:deterministic_commutator_loss}
\end{align}
Equation~\eqref{eq:deterministic_commutator_loss} is inconsistent with the
required output commutator $[a_s(z_i+\delta z,\Omega),
a_s^\dagger(z_i+\delta z,\Omega')]
=
\delta(\Omega-\Omega').$ To preserve the commutation relation, we thus need to add the 
Langevin operator \(v(z,\Omega)\)  as in Eq.~\eqref{eq:langevin_attenuation}.

Let us now verify explicitly that the solution
Eq.~\eqref{eq:finite_segment_langevin_solution} preserves the canonical
commutator. For notational clarity, set the left endpoint of the segment to
zero and write the right endpoint as \(z\). The solution then reads
\begin{equation}
a_s(z,\Omega)
=
e^{-\alpha z/2}a_s(0,\Omega)
+
\sqrt{\alpha}
\int_0^z dz_1\,
e^{-\alpha(z-z_1)/2}v(z_1,\Omega).
\label{eq:simplified_segment_solution}
\end{equation}
Define the deterministic and noise contributions by  $A(z,\Omega)$ and $B(z,\Omega)$ respectively, such that:
\begin{equation}
    a_s(z,\Omega)=A(z,\Omega)+B(z,\Omega), \quad a_s^\dagger(z,\Omega')
=
A^\dagger(z,\Omega')+B^\dagger(z,\Omega'),
\end{equation}

Using bilinearity of the commutator, we obtain
\begin{align}
[a_s(z,\Omega),a_s^\dagger(z,\Omega')]
&=
[A(z,\Omega)+B(z,\Omega),
A^\dagger(z,\Omega')+B^\dagger(z,\Omega')]
\nonumber\\
&=
[A(z,\Omega),A^\dagger(z,\Omega')]
+
[A(z,\Omega),B^\dagger(z,\Omega')]
\nonumber\\
&\quad+
[B(z,\Omega),A^\dagger(z,\Omega')]
+
[B(z,\Omega),B^\dagger(z,\Omega')].
\label{eq:commutator_decomposition}
\end{align}

The deterministic contribution is $[A(z,\Omega),A^\dagger(z,\Omega')] =e^{-\alpha z}\delta(\Omega-\Omega') $, while the cross terms vanish as input modes and the reservoir modes are independent. It remains to compute the reservoir contribution: 

\begin{align}
[B(z,\Omega),B^\dagger(z,\Omega')]
&=
\left[
\sqrt{\alpha}\int_0^z dz_1\,
e^{-\alpha(z-z_1)/2}v(z_1,\Omega),
\sqrt{\alpha}\int_0^z dz_2\,
e^{-\alpha(z-z_2)/2}v^\dagger(z_2,\Omega')
\right]
\nonumber\\
&=
\alpha
\int_0^z dz_1
\int_0^z dz_2\,
e^{-\alpha(z-z_1)/2}
e^{-\alpha(z-z_2)/2}
[v(z_1,\Omega),v^\dagger(z_2,\Omega')] \nonumber \\
&= \alpha
\int_0^z dz_1
\int_0^z dz_2\,
e^{-\alpha(z-z_1)/2}
e^{-\alpha(z-z_2)/2}
\delta(z_1-z_2)\delta(\Omega-\Omega')
\nonumber\\
&=
\alpha\delta(\Omega-\Omega')
\int_0^z dz_1
\int_0^z dz_2\,
e^{-\alpha(z-z_1)/2}
e^{-\alpha(z-z_2)/2}
\delta(z_1-z_2) \nonumber \\
&=
\alpha\delta(\Omega-\Omega')
\int_0^z dz_1\,
e^{-\alpha(z-z_1)/2}
e^{-\alpha(z-z_1)/2}
\nonumber\\
&=
\alpha\delta(\Omega-\Omega')
\int_0^z dz_1\,
e^{-\alpha(z-z_1)} \nonumber \\
&= \left(1-e^{-\alpha z}\right)\delta(\Omega-\Omega').
\end{align}
Substituting these commutation relations into
Eq.~\eqref{eq:commutator_decomposition}, we find
\begin{align}
[a_s(z,\Omega),a_s^\dagger(z,\Omega')]
&=
e^{-\alpha z}\delta(\Omega-\Omega')
+
\left(1-e^{-\alpha z}\right)\delta(\Omega-\Omega')
\nonumber\\
&=
\left[
e^{-\alpha z}+1-e^{-\alpha z}
\right]\delta(\Omega-\Omega')
\nonumber\\
&=
\delta(\Omega-\Omega').
\label{eq:commutator_preserved_simplified_segment}
\end{align}
Therefore, the Langevin noise term supplies exactly the commutator contribution
lost by deterministic attenuation.

\subsection{Effective Vacuum Mode and Beam-Splitter Form}
\label{subsec:effective_vacuum_mode}

We now show that the continuum-noise contribution in
Eq.~\eqref{eq:finite_segment_langevin_solution} can be written as a single
effective vacuum mode. Define the segment input and output operators $a_{s,\mathrm{in}}(\Omega)
:=
a_s(z_i,\Omega),$ and  $a_{s,\mathrm{out}}(\Omega)
:=
a_s(z_i+\delta z,\Omega)$ respectively. 
The segment intensity transmissivity is $\eta_i
:=
e^{-\alpha_{s,i}\delta z}$ , and 
Eq.~\eqref{eq:finite_segment_langevin_solution} becomes
\begin{equation}
a_{s,\mathrm{out}}(\Omega)
=
\sqrt{\eta_i}\,a_{s,\mathrm{in}}(\Omega)
+
\sqrt{\alpha_{s,i}}
\int_{z_i}^{z_i+\delta z} dz\,
e^{-\alpha_{s,i}(z_i+\delta z-z)/2}v(z,\Omega).
\label{eq:segment_solution_input_output}
\end{equation}
We now define the effective reservoir mode \(v_i(\Omega)\) by
\begin{equation}
v_i(\Omega)
:=
\frac{\sqrt{\alpha_{s,i}}}{\sqrt{1-\eta_i}}
\int_{z_i}^{z_i+\delta z} dz\,
e^{-\alpha_{s,i}(z_i+\delta z-z)/2}v(z,\Omega).
\label{eq:effective_vacuum_mode_definition}
\end{equation}
which gives the standard beam-splitter representation of a pure-loss Gaussian channel:
\begin{equation}
a_{s,\mathrm{out}}(\Omega)
=
\sqrt{\eta_i}\,a_{s,\mathrm{in}}(\Omega)
+
\sqrt{1-\eta_i}\,v_i(\Omega).
\label{eq:beam_splitter_form}
\end{equation}
The transformation in Eq.~\eqref{eq:beam_splitter_form} is therefore
not an additional approximation; it is the finite-segment representation of the
underlying Langevin attenuation equation.   In the language of the beam splitter, the optical input is transmitted with an amplitude coefficient
\(\sqrt{\eta_i}\), while the reservoir vacuum mode enters through the unused port
with amplitude coefficient \(\sqrt{1-\eta_i}\). We can also verify that $v_i(\Omega)$ is a a properly normalized bosonic mode: 
\begin{align}
[v_i(\Omega),v_i^\dagger(\Omega')]
&=
\frac{\alpha_{s,i}}{1-\eta_i}
\int_{z_i}^{z_i+\delta z}dz
\int_{z_i}^{z_i+\delta z}dz'\,
e^{-\alpha_{s,i}(z_i+\delta z-z)/2}\textit{}
e^{-\alpha_{s,i}(z_i+\delta z-z')/2}
[v(z,\Omega),v^\dagger(z',\Omega')] \nonumber \\
&= \frac{\alpha_{s,i}}{1-\eta_i}
\int_{z_i}^{z_i+\delta z}dz\,
e^{-\alpha_{s,i}(z_i+\delta z-z)}
\,\delta(\Omega-\Omega') \nonumber  \\
&=  \frac{\alpha_{s,i}}{1-\eta_i} \delta(\Omega-\Omega')  \int_{z_i}^{z_i+\delta z}dz\,e^{-\alpha_{s,i}(z_i+\delta z-z)} \nonumber  \\
&=  \frac{\alpha_{s,i}}{1-\eta_i} \delta(\Omega-\Omega')  \int_0^{\delta z}du\,e^{-\alpha_{s,i}u}  \nonumber  \\
&=  \frac{\alpha_{s,i}}{1-\eta_i} \frac{1-e^{-\alpha_{s,i}\delta z}}{\alpha_{s,i}} \delta(\Omega-\Omega')  \nonumber \\
&=  \frac{\alpha_{s,i}}{1-\eta_i} \frac{1-\eta_i}{\alpha_{s,i}}\delta(\Omega-\Omega')  \nonumber  \\
&= \delta(\Omega-\Omega') \nonumber 
\end{align}

The beam-splitter relation in Eq.~\eqref{eq:beam_splitter_form} also makes the
commutator preservation transparent: 
\begin{align}
[a_{s,\mathrm{out}}(\Omega),a_{s,\mathrm{out}}^\dagger(\Omega')]
&=
\eta_i\delta(\Omega-\Omega')
+
(1-\eta_i)\delta(\Omega-\Omega')
\nonumber\\
&=
\left[\eta_i+1-\eta_i\right]\delta(\Omega-\Omega')
\nonumber\\
&=
\delta(\Omega-\Omega').
\label{eq:beam_splitter_commutator_preserved}
\end{align}
Thus the beam-splitter form is exactly equivalent to the Langevin solution and
preserves the frequency-resolved bosonic algebra. 

For a sufficiently short segment, $\alpha_{s,i}\delta z \ll 1,$
 we get: 
\begin{equation}
\eta_i
=
e^{-\alpha_{s,i}\delta z}
\simeq
1-\alpha_{s,i}\delta z,
\qquad
1-\eta_i
\simeq
\alpha_{s,i}\delta z \qquad e^{-\alpha_{s,i}(z_i+\delta z-z)/2}
=
1+O(\delta z).
\end{equation}
 Consequently, $\sqrt{1-\eta_i}\,v_i(\Omega)
\simeq
\sqrt{\alpha_{s,i}\delta z}\,v_i(\Omega),$
which is the standard infinitesimal Langevin noise scaling. The $\sqrt{\delta z}$ dependence arises because the integrated
continuum noise has variance proportional to $\delta z$, as implied
by the continuum commutation relation.

\subsection{Continuous-Variable Representation of the Loss Channel}
\label{sec:cv_loss_channel_derivation}

 The corresponding input-output relation, together with its Hermitian adjoint, is
\begin{align}
a_{s,\mathrm{out}}(\Omega)
&=
\sqrt{\eta_i}\,a_{s,\mathrm{in}}(\Omega)
+
\sqrt{1-\eta_i}\,v_i(\Omega),  \qquad
a_{s,\mathrm{out}}^\dagger(\Omega')
=
\sqrt{\eta_i}\,a_{s,\mathrm{in}}^\dagger(\Omega')
+
\sqrt{1-\eta_i}\,v_i^\dagger(\Omega').
\label{eq:annihilation_loss_relation}
\end{align}
This does preserve the canonical commutation relation:

\begin{align}
\left[
a_{s,\mathrm{out}}(\Omega),
a_{s,\mathrm{out}}^\dagger(\Omega')
\right]
&=
\left[
\sqrt{\eta_i}\,a_{s,\mathrm{in}}(\Omega)
+
\sqrt{1-\eta_i}\,v_i(\Omega),
\sqrt{\eta_i}\,a_{s,\mathrm{in}}^\dagger(\Omega')
+
\sqrt{1-\eta_i}\,v_i^\dagger(\Omega')
\right]
\nonumber\\
&=
\eta_i
\left[
a_{s,\mathrm{in}}(\Omega),
a_{s,\mathrm{in}}^\dagger(\Omega')
\right]
+
\sqrt{\eta_i(1-\eta_i)}
\left[
a_{s,\mathrm{in}}(\Omega),
v_i^\dagger(\Omega')
\right]
\nonumber\\
&\quad
+
\sqrt{\eta_i(1-\eta_i)}
\left[
v_i(\Omega),
a_{s,\mathrm{in}}^\dagger(\Omega')
\right]
+
(1-\eta_i)
\left[
v_i(\Omega),
v_i^\dagger(\Omega')
\right] \nonumber \\
&=
\eta_i\,\delta(\Omega-\Omega')
+
\sqrt{\eta_i(1-\eta_i)}\,0
+
\sqrt{\eta_i(1-\eta_i)}\,0
+
(1-\eta_i)\,\delta(\Omega-\Omega')
\nonumber\\
&=
\left[\eta_i+(1-\eta_i)\right]\delta(\Omega-\Omega')
\nonumber\\
&=
\delta(\Omega-\Omega') \nonumber.
\end{align}

We now introduce the continuous-variable quadrature representation:
\begin{align}
q_{v_i}(\Omega)
&=
\frac{
v_i(\Omega)+v_i^\dagger(\Omega)
}{\sqrt{2}},
\qquad
p_{v_i}(\Omega)
=
\frac{
v_i(\Omega)-v_i^\dagger(\Omega)
}{i\sqrt{2}}, \nonumber \\
q_{s,\mathrm{in}}(\Omega)
&=
\frac{
a_{s,\mathrm{in}}(\Omega)+a_{s,\mathrm{in}}^\dagger(\Omega)
}{\sqrt{2}},
\qquad
p_{s,\mathrm{in}}(\Omega)
=
\frac{
a_{s,\mathrm{in}}(\Omega)-a_{s,\mathrm{in}}^\dagger(\Omega)
}{i\sqrt{2}}, \nonumber \\
q_{s,\mathrm{out}}(\Omega)
&=
\frac{
a_{s,\mathrm{out}}(\Omega)+a_{s,\mathrm{out}}^\dagger(\Omega)
}{\sqrt{2}},
\qquad
p_{s,\mathrm{out}}(\Omega)
=
\frac{
a_{s,\mathrm{out}}(\Omega)-a_{s,\mathrm{out}}^\dagger(\Omega)
}{i\sqrt{2}}.
\end{align}

Similarly, we now derive the transformation of the signal quadratures: 
\begin{align}
q_{s,\mathrm{out}}(\Omega)
&=
\frac{
a_{s,\mathrm{out}}(\Omega)+a_{s,\mathrm{out}}^\dagger(\Omega)
}{\sqrt{2}}
\nonumber\\
&=
\frac{
\sqrt{\eta_i}\,a_{s,\mathrm{in}}(\Omega)
+
\sqrt{1-\eta_i}\,v_i(\Omega)
+
\sqrt{\eta_i}\,a_{s,\mathrm{in}}^\dagger(\Omega)
+
\sqrt{1-\eta_i}\,v_i^\dagger(\Omega)
}{\sqrt{2}}
\nonumber\\
&=
\sqrt{\eta_i}
\frac{
a_{s,\mathrm{in}}(\Omega)+a_{s,\mathrm{in}}^\dagger(\Omega)
}{\sqrt{2}}
+
\sqrt{1-\eta_i}
\frac{
v_i(\Omega)+v_i^\dagger(\Omega)
}{\sqrt{2}}
\nonumber\\
&=
\sqrt{\eta_i}\,q_{s,\mathrm{in}}(\Omega)
+
\sqrt{1-\eta_i}\,q_{v_i}(\Omega).
\label{eq:q_output_transformation}
\end{align}
Similarly,
\begin{align}
p_{s,\mathrm{out}}(\Omega)
&=
\frac{
a_{s,\mathrm{out}}(\Omega)-a_{s,\mathrm{out}}^\dagger(\Omega)
}{i\sqrt{2}}
\nonumber\\
&=
\frac{
\sqrt{\eta_i}\,a_{s,\mathrm{in}}(\Omega)
+
\sqrt{1-\eta_i}\,v_i(\Omega)
-
\sqrt{\eta_i}\,a_{s,\mathrm{in}}^\dagger(\Omega)
-
\sqrt{1-\eta_i}\,v_i^\dagger(\Omega)
}{i\sqrt{2}}
\nonumber\\
&=
\sqrt{\eta_i}
\frac{
a_{s,\mathrm{in}}(\Omega)-a_{s,\mathrm{in}}^\dagger(\Omega)
}{i\sqrt{2}}
+
\sqrt{1-\eta_i}
\frac{
v_i(\Omega)-v_i^\dagger(\Omega)
}{i\sqrt{2}}
\nonumber\\
&=
\sqrt{\eta_i}\,p_{s,\mathrm{in}}(\Omega)
+
\sqrt{1-\eta_i}\,p_{v_i}(\Omega).
\label{eq:p_output_transformation}
\end{align}

Introducing the quadrature vectors
\begin{equation}
R_{s,\mathrm{in}}(\Omega)
:=
\begin{pmatrix}
q_{s,\mathrm{in}}(\Omega)\\
p_{s,\mathrm{in}}(\Omega)
\end{pmatrix},
\qquad
R_{s,\mathrm{out}}(\Omega)
:=
\begin{pmatrix}
q_{s,\mathrm{out}}(\Omega)\\
p_{s,\mathrm{out}}(\Omega)
\end{pmatrix},
\qquad
R_{v_i}(\Omega)
:=
\begin{pmatrix}
q_{v_i}(\Omega)\\
p_{v_i}(\Omega)
\end{pmatrix},
\label{eq:quadrature_vectors_defined} \nonumber
\end{equation}
Eqs.~\eqref{eq:q_output_transformation} and \eqref{eq:p_output_transformation} combine into
\begin{equation}
R_{s,\mathrm{out}}(\Omega)
=
\sqrt{\eta_i}\,R_{s,\mathrm{in}}(\Omega)
+
\sqrt{1-\eta_i}\,R_{v_i}(\Omega).
\label{eq:quadrature_vector_loss_relation}
\end{equation}

We next compute the transformation of first moments. Define
\begin{equation}
d_{\mathrm{in}}(\Omega)
:=
\left\langle R_{s,\mathrm{in}}(\Omega)\right\rangle,
\qquad
d_{\mathrm{out}}(\Omega)
:=
\left\langle R_{s,\mathrm{out}}(\Omega)\right\rangle,
\qquad
d_{v_i}(\Omega)
:=
\left\langle R_{v_i}(\Omega)\right\rangle.
\label{eq:displacement_vectors_defined}
\end{equation}
Taking the expectation value of Eq.~\eqref{eq:quadrature_vector_loss_relation}, one obtains
\begin{align}
d_{\mathrm{out}}(\Omega)
&=
\left\langle
R_{s,\mathrm{out}}(\Omega)
\right\rangle
\nonumber\\
&=
\left\langle
\sqrt{\eta_i}\,R_{s,\mathrm{in}}(\Omega)
+
\sqrt{1-\eta_i}\,R_{v_i}(\Omega)
\right\rangle
\nonumber\\
&=
\sqrt{\eta_i}\,
\left\langle R_{s,\mathrm{in}}(\Omega)\right\rangle
+
\sqrt{1-\eta_i}\,
\left\langle R_{v_i}(\Omega)\right\rangle
\nonumber\\
&=
\sqrt{\eta_i}\,d_{\mathrm{in}}(\Omega)
+
\sqrt{1-\eta_i}\,d_{v_i}(\Omega).
\label{eq:first_moment_general}
\end{align}
Since the ancillary mode is in the vacuum state,  $d_{v_i}(\Omega)=0,$ as 
\begin{equation}
\left\langle v_i(\Omega)\right\rangle
=
\left\langle v_i^\dagger(\Omega)\right\rangle
=
0,  \quad \left\langle q_{v_i}(\Omega)\right\rangle
=
\frac{
\left\langle v_i(\Omega)\right\rangle
+
\left\langle v_i^\dagger(\Omega)\right\rangle
}{\sqrt{2}}
=
0, \quad \left\langle p_{v_i}(\Omega)\right\rangle
=
\frac{
\left\langle v_i(\Omega)\right\rangle
-
\left\langle v_i^\dagger(\Omega)\right\rangle
}{i\sqrt{2}}
=
0. \nonumber
\end{equation}
Eq.~\eqref{eq:first_moment_general} reduces to
\begin{equation}
d_{\mathrm{out}}(\Omega)
=
\sqrt{\eta_i}\,d_{\mathrm{in}}(\Omega). 
\label{eq:first_moment_loss_channel}
\end{equation}

We now compute the covariance matrix. For a quadrature vector $R$, define
\begin{equation}
V_{\alpha\beta}
:=
\frac{1}{2}
\left\langle
\Delta R_\alpha \Delta R_\beta
+
\Delta R_\beta \Delta R_\alpha
\right\rangle,
\qquad
\Delta R_\alpha
:=
R_\alpha-\left\langle R_\alpha\right\rangle,
\label{eq:covariance_matrix_definition}
\end{equation}
where $\alpha,\beta\in\{q,p\}$. From Eq.~\eqref{eq:quadrature_vector_loss_relation},
\begin{align}
\Delta R_{s,\mathrm{out}}(\Omega)
&=
R_{s,\mathrm{out}}(\Omega)
-
\left\langle R_{s,\mathrm{out}}(\Omega)\right\rangle
\nonumber\\
&=
\sqrt{\eta_i}
\left[
R_{s,\mathrm{in}}(\Omega)
-
\left\langle R_{s,\mathrm{in}}(\Omega)\right\rangle
\right]
+
\sqrt{1-\eta_i}
\left[
R_{v_i}(\Omega)
-
\left\langle R_{v_i}(\Omega)\right\rangle
\right]
\nonumber\\
&=
\sqrt{\eta_i}\,\Delta R_{s,\mathrm{in}}(\Omega)
+
\sqrt{1-\eta_i}\,\Delta R_{v_i}(\Omega).
\label{eq:output_fluctuation_vector}
\end{align}
Therefore,
\begin{align}
\left[V_{\mathrm{out}}(\Omega)\right]_{\alpha\beta}
&=
\frac{1}{2}
\left\langle
\Delta R_{s,\mathrm{out},\alpha}(\Omega)
\Delta R_{s,\mathrm{out},\beta}(\Omega)
+
\Delta R_{s,\mathrm{out},\beta}(\Omega)
\Delta R_{s,\mathrm{out},\alpha}(\Omega)
\right\rangle
\nonumber\\
&=
\frac{1}{2}
\Big\langle
\left[
\sqrt{\eta_i}\,\Delta R_{s,\mathrm{in},\alpha}
+
\sqrt{1-\eta_i}\,\Delta R_{v_i,\alpha}
\right]
\left[
\sqrt{\eta_i}\,\Delta R_{s,\mathrm{in},\beta}
+
\sqrt{1-\eta_i}\,\Delta R_{v_i,\beta}
\right]
\nonumber\\
&\quad+
\left[
\sqrt{\eta_i}\,\Delta R_{s,\mathrm{in},\beta}
+
\sqrt{1-\eta_i}\,\Delta R_{v_i,\beta}
\right]
\left[
\sqrt{\eta_i}\,\Delta R_{s,\mathrm{in},\alpha}
+
\sqrt{1-\eta_i}\,\Delta R_{v_i,\alpha}
\right]
\Big\rangle
\nonumber\\
&=
\eta_i
\left[V_{\mathrm{in}}(\Omega)\right]_{\alpha\beta}
+
(1-\eta_i)
\left[V_{v_i}(\Omega)\right]_{\alpha\beta}
\nonumber\\
&\quad+
\frac{\sqrt{\eta_i(1-\eta_i)}}{2}
\left\langle
\Delta R_{s,\mathrm{in},\alpha}
\Delta R_{v_i,\beta}
+
\Delta R_{v_i,\beta}
\Delta R_{s,\mathrm{in},\alpha}
\right\rangle
\nonumber\\
&\quad+
\frac{\sqrt{\eta_i(1-\eta_i)}}{2}
\left\langle
\Delta R_{v_i,\alpha}
\Delta R_{s,\mathrm{in},\beta}
+
\Delta R_{s,\mathrm{in},\beta}
\Delta R_{v_i,\alpha}
\right\rangle.
\label{eq:output_covariance_expanded}
\end{align}
The signal and the ancillary vacuum are assumed to be initially uncorrelated, so the total input state factorizes as $\rho_{\mathrm{tot}}
=
\rho_{s,\mathrm{in}}\otimes |0\rangle_{v_i}\langle 0|.$
Therefore, $\left\langle
\Delta R_{s,\mathrm{in},\alpha}
\Delta R_{v_i,\beta}
\right\rangle
=
\left\langle
\Delta R_{s,\mathrm{in},\alpha}
\right\rangle
\left\langle
\Delta R_{v_i,\beta}
\right\rangle
=
0,$
and all cross-covariance terms in Eq.~\eqref{eq:output_covariance_expanded} vanish. Hence
\begin{equation}
V_{\mathrm{out}}(\Omega)
=
\eta_i V_{\mathrm{in}}(\Omega)
+
(1-\eta_i)V_{v_i}(\Omega).
\label{eq:covariance_general_loss}
\end{equation}

It remains to evaluate $V_{v_i}$. For a vacuum state satisfying $v_i|0\rangle_{v_i}=0$, the relevant correlations are $\langle v_i\rangle
=
\langle v_i^\dagger\rangle
=
0,
\langle v_i^\dagger v_i\rangle
=
0,
\langle v_i v_i^\dagger\rangle
=
1.$
The vacuum variance of $q_{v_i}$ is
\begin{align}
\left\langle q_{v_i}^2\right\rangle
=
\left\langle
\left(
\frac{v_i+v_i^\dagger}{\sqrt{2}}
\right)^2
\right\rangle
=
\frac{1}{2}
\left\langle
v_i^2
+
v_i v_i^\dagger
+
v_i^\dagger v_i
+
(v_i^\dagger)^2
\right\rangle
=
\frac{1}{2}
\left(
0+1+0+0
\right)
=
\frac{1}{2}.
\label{eq:vacuum_q_variance}
\end{align}
The vacuum variance of $p_{v_i}$ is
\begin{align}
\left\langle p_{v_i}^2\right\rangle
=
\left\langle
\left(
\frac{v_i-v_i^\dagger}{i\sqrt{2}}
\right)^2
\right\rangle
=
-\frac{1}{2}
\left\langle
v_i^2
-
v_i v_i^\dagger
-
v_i^\dagger v_i
+
(v_i^\dagger)^2
\right\rangle
=
-\frac{1}{2}
\left(
0-1-0+0
\right)
=
\frac{1}{2}.
\label{eq:vacuum_p_variance}
\end{align}
The symmetrized covariance between $q_{v_i}$ and $p_{v_i}$ is
\begin{align}
\frac{1}{2}
\left\langle
q_{v_i}p_{v_i}+p_{v_i}q_{v_i}
\right\rangle
=
\frac{1}{2}
\left\langle
\frac{v_i+v_i^\dagger}{\sqrt{2}}
\frac{v_i-v_i^\dagger}{i\sqrt{2}}
+
\frac{v_i-v_i^\dagger}{i\sqrt{2}}
\frac{v_i+v_i^\dagger}{\sqrt{2}}
\right\rangle
=
0.
\label{eq:vacuum_qp_correlation}
\end{align}
Thus
\begin{equation}
V_{v_i}
=
\begin{pmatrix}
\langle q_{v_i}^2\rangle &
\frac{1}{2}\langle q_{v_i}p_{v_i}+p_{v_i}q_{v_i}\rangle\\
\frac{1}{2}\langle p_{v_i}q_{v_i}+q_{v_i}p_{v_i}\rangle &
\langle p_{v_i}^2\rangle
\end{pmatrix}
=
\frac{1}{2}
\begin{pmatrix}
1&0\\
0&1
\end{pmatrix}
=
\frac{1}{2}I_2.
\label{eq:vacuum_covariance_matrix}
\end{equation}
Substituting Eq.~\eqref{eq:vacuum_covariance_matrix} into Eq.~\eqref{eq:covariance_general_loss}, we obtain
\begin{equation}
V_{\mathrm{out}}(\Omega)
=
\eta_i V_{\mathrm{in}}(\Omega)
+
\frac{1-\eta_i}{2}I_2.
\label{eq:loss_channel_covariance_final}
\end{equation}
Equations~\eqref{eq:first_moment_loss_channel} and \eqref{eq:loss_channel_covariance_final} give the continuous-variable Gaussian channel induced by the annihilation-operator transformation in Eq.~\eqref{eq:annihilation_loss_relation}.

In the standard Gaussian-channel notation,
\begin{equation}
d_{\mathrm{out}}
=
X_i d_{\mathrm{in}},
\qquad
V_{\mathrm{out}}
=
X_i V_{\mathrm{in}}X_i^{T}+Y_i,
\label{eq:general_gaussian_channel_form}
\end{equation}
the matrices are $X_i
=
\sqrt{\eta_i}\,I_2,
Y_i
=
\frac{1-\eta_i}{2}I_2.$
Indeed,
\begin{align}
X_iV_{\mathrm{in}}X_i^T+Y_i
&=
\left(\sqrt{\eta_i}I_2\right)
V_{\mathrm{in}}
\left(\sqrt{\eta_i}I_2\right)^T
+
\frac{1-\eta_i}{2}I_2=
\eta_i V_{\mathrm{in}}
+
\frac{1-\eta_i}{2}I_2,
\end{align}
which is identical to Eq.~\eqref{eq:loss_channel_covariance_final}.

\section{Derivation of the Output Displaced-Thermal State}
\label{app:displaced-thermal-output}

We will now derive the output
parameters used in Eqs. \ref{eq:clean_parameters_theta}. For a single bosonic mode with
annihilation operator $\hat a$, the displacement operator is
\begin{equation}
    D(\alpha)
    =
    \exp\!\left(\alpha \hat a^\dagger-\alpha^*\hat a\right).
\end{equation}
A coherent state is defined by $  |\alpha\rangle=D(\alpha)|0\rangle,$ with $ \hat a|\alpha\rangle=\alpha|\alpha\rangle$, and has a mean photon number
$\langle\hat a^\dagger\hat a\rangle=|\alpha|^2$. In terms of the
quadratures $\hat q=\frac{\hat a+\hat a^\dagger}{\sqrt{2}},$ and 
$\hat p=\frac{\hat a-\hat a^\dagger}{i\sqrt{2}}$, its displacement vector and covariance matrix are
\begin{equation}
    d_\alpha
    =
    \sqrt{2}
    \begin{pmatrix}
        \Re\alpha\\
        \Im\alpha
    \end{pmatrix},
    \qquad
    V_\alpha=\frac{1}{2}I_2.
\end{equation}

A thermal state with mean photon number $N$ is
\begin{equation}
    S_N(0)
    =
    \sum_{n=0}^{\infty}
    \frac{N^n}{(N+1)^{n+1}}
    |n\rangle\langle n|,
\end{equation}
and a displaced thermal state is $  S_N(\beta)
    =
    D(\beta)S_N(0)D^\dagger(\beta).$ Its first and second moments are
\begin{equation}
    \langle\hat a\rangle=\beta,
    \qquad
    d_\beta
    =
    \sqrt{2}
    \begin{pmatrix}
        \Re\beta\\
        \Im\beta
    \end{pmatrix},
    \qquad
    V=\left(N+\frac{1}{2}\right)I_2.
\end{equation}
In particular, $S_0(\beta)=|\beta\rangle\langle\beta|$.

We now consider a fiber composed of $L$ segments with uniform
intensity transmissivity $\eta$. The SBS interaction is assumed to be
active only in segment $i$. For every passive segment $j\neq i$, the
Gaussian-channel matrices are
\begin{equation}
    X_j=\sqrt{\eta}\,I_2,
    \qquad
    Y_j=\frac{1-\eta}{2}I_2.
\end{equation}
For the active segment, $  \mu_i(\Omega;\theta)
    =
    \sqrt{\eta}+\kappa_i(\Omega;\theta),$ 
and
\begin{equation}
    X_i
    =
    \begin{pmatrix}
        \Re\mu_i & -\Im\mu_i\\
        \Im\mu_i & \Re\mu_i
    \end{pmatrix}, \qquad Y_i
    =
    \left[
        \left(n_{\mathrm{th}}(\theta)+\frac{1}{2}\right)
        |\nu_i(\Omega;\theta)|^2
        +
        \frac{1-\eta}{2}
    \right]I_2.
\end{equation}

For the coherent input $|\sqrt{E}\rangle$, the amplitude immediately
before the active segment is
\begin{equation}
    \alpha_i^{\mathrm{in}}
    =
    \eta^{(i-1)/2}\sqrt{E}.
\end{equation}
The active segment multiplies this amplitude by $\mu_i$, while the
remaining $L-i$ passive segments contribute a factor
$\eta^{(L-i)/2}$. Hence the output displacement is
\begin{align}
    \beta_i(\Omega;\theta)
    &=
    \eta^{(L-i)/2}
    \mu_i(\Omega;\theta)
    \eta^{(i-1)/2}\sqrt{E}
    \nonumber\\
    &=
    \eta^{(L-1)/2}
    \left[
        \sqrt{\eta}+\kappa_i(\Omega;\theta)
    \right]\sqrt{E}.
    \label{eq:app-beta-i}
\end{align}

A coherent state remains coherent under pure vacuum loss, so the
covariance matrix entering the active segment is $I_2/2$. Immediately
after the active segment,
\begin{align}
    V_i^{+}
    &=
    X_i\frac{I_2}{2}X_i^{T}+Y_i
    \nonumber\\
    &=
    \left[
        \frac{|\mu_i|^2}{2}
        +
        \left(n_{\mathrm{th}}+\frac{1}{2}\right)|\nu_i|^2
        +
        \frac{1-\eta}{2}
    \right]I_2 \\
    &= \left(N_i^{\mathrm{loc}}+\frac{1}{2}\right)I_2
\end{align}
where, 
\begin{equation}
     N_i^{\mathrm{loc}}
    =
    \left(n_{\mathrm{th}}+\frac{1}{2}\right)|\nu_i|^2
    +
    \frac{|\mu_i|^2-\eta}{2}.
\end{equation}
After the remaining $L-i$ passive segments,
\begin{align}
    N_i(\Omega;\theta)
    &=
    \eta^{L-i}N_i^{\mathrm{loc}}
    \nonumber\\
    &=
    \eta^{L-i}
    \left[
        \left(n_{\mathrm{th}}(\theta)+\frac{1}{2}\right)
        |\nu_i(\Omega;\theta)|^2
        +
        \frac{
            |\sqrt{\eta}+\kappa_i(\Omega;\theta)|^2-\eta
        }{2}
    \right].
    \label{eq:app-N-i}
\end{align}

Consequently, the final Stokes state is $\rho_{\mathrm{out}}(\Omega;\theta)
    =
    S_{N_i(\Omega;\theta)}
    \!\left(\beta_i(\Omega;\theta)\right).$
As a consistency check, in the pure-loss limit
$\kappa_i=\nu_i=0$, one finds $N_i=0$ and
\begin{equation}
    \rho_{\mathrm{out}}
    =
    \left|
        \eta^{L/2}\sqrt{E}
    \right\rangle
    \left\langle
        \eta^{L/2}\sqrt{E}
    \right|,
\end{equation}
as expected for a coherent state transmitted through a pure-loss
channel.

\section{Quantum Fisher Information}
\label{sec:QFI}
\subsection{Symmetric Logarithmic Derivative (SLD)}

Let $\rho_\theta$ be a density operator depending smoothly on an unknown parameter $\theta$. The goal of quantum parameter estimation is to infer $\theta$ from measurements performed on $\rho_\theta$.

The symmetric logarithmic derivative (SLD), denoted by $L_\theta$, is defined as:
\begin{equation}
\label{eq:SLD}
  \frac{\partial \rho_\theta}{\partial \theta} =   \frac{1}{2}\{\rho_\theta,L_\theta\} = \frac{1}{2}\left(\rho_\theta L_\theta + L_\theta \rho_\theta\right)
\end{equation}
where $\{X,Y\}=XY+YX$ is the anticommutator. The SLD is the quantum analogue of the classical logarithmic derivative, or score function. To obtain an explicit expression for the SLD, we diagonalize the density matrix as $\rho_\theta
=
\sum_k p_k |k\rangle\langle k|$ with eigenvalues and eigenstate as $\rho_\theta |k\rangle = p_k |k\rangle$, and take the matrix elements of equation \ref{eq:SLD}
\begin{align}
\langle k|\partial_\theta \rho_\theta|l\rangle
&=
\frac{1}{2}
\langle k|\rho_\theta L_\theta + L_\theta \rho_\theta|l\rangle \nonumber \\
&= \frac{1}{2}
\langle k|p_k L_\theta + L_\theta p_l|l\rangle \nonumber \\
&= 
\langle k| \frac{1}{2}(p_k + p_l)L_\theta|l\rangle
\end{align}
Therefore, $(L_\theta)_{kl}
=
\frac{2(\partial_\theta \rho_\theta)_{kl}}{p_k+p_l}$ and $L_\theta$ is explicitly:
\begin{equation}
\label{eq:Llambda}
L_\theta
=
2\sum_{k,l:\,p_k+p_l>0}
\frac{\langle k|\partial_\theta \rho_\theta|l\rangle}{p_k+p_l}
|k\rangle\langle l|.    
\end{equation}

A particularly important case is unitary parameter encoding, $\rho_\theta
=
e^{-i\theta A}\rho\,e^{i\theta A},$
where $A=A^\dagger$ is the generator of the parameter shift. Differentiating gives $\partial_\theta \rho_\theta
=
-i[A,\rho_\theta],$ where $[X,Y]= XY-YX$ is the commutator. Now, the equation \ref{eq:SLD}, where we defined SLD, can be written as:
\begin{equation}
\partial_\theta \rho_\theta= -i[A,\rho_\theta] = \frac{1}{2}\{\rho_\theta,L_\theta\}
\end{equation}
The matrix elements becomes $   \langle k|\partial_\theta \rho_\theta|l\rangle
=
-i\langle k|A\rho_\theta-\rho_\theta A|l\rangle = i (p_k - p_l)  \langle k|A|l\rangle. $ Substituting into Eq. \ref{eq:Llambda} gives
\begin{equation}
L_\theta
=
2i
\sum_{k,l:\,p_k+p_l>0}
\frac{p_k-p_l}{p_k+p_l}
\langle k|A|l\rangle
|k\rangle\langle l|.    
\end{equation}

\subsection{Quantum Fisher Information}

In quantum parameter estimation, measurement outcomes are described by a POVM $\{\Pi_x\}$, and the probability distribution is given by the Born rule $p(x|\theta) = \mathrm{Tr}[\rho_\theta \Pi_x].$ The derivative of the probability is given by:
\begin{align}
\label{eq:partial_p}
\partial_\theta p(x|\theta)
&= \mathrm{Tr}[(\partial_\theta \rho_\theta)\Pi_x] = \frac{1}{2}\mathrm{Tr}[(\rho_\theta L_\theta + L_\theta \rho_\theta)\Pi_x] = \frac{1}{2}\left(
\mathrm{Tr}[\rho_\theta L_\theta \Pi_x]
+
\mathrm{Tr}[L_\theta \rho_\theta \Pi_x]
\right).
\end{align}

Using cyclicity of the trace, $\mathrm{Tr}[L_\theta \rho_\theta \Pi_x] = \mathrm{Tr}[\rho_\theta \Pi_x L_\theta],$ and Hermiticity of $\rho_\theta, \Pi_x, L_\theta$,
\begin{equation}
   \mathrm{Tr}[\rho_\theta L_\theta \Pi_x] =  \mathrm{Tr}([ \Pi_x L_\theta \rho_\theta]^\dagger) = (\mathrm{Tr}( \Pi_x L_\theta \rho_\theta])^* = (\mathrm{Tr}( \rho_\theta\Pi_x L_\theta ])^*. 
\end{equation}
Thus, equation \ref{eq:partial_p} can be expressed as:
\begin{align}
\partial_\theta p(x|\theta)
&=\frac{1}{2}\left(
\mathrm{Tr}[\rho_\theta L_\theta \Pi_x]
+
\mathrm{Tr}[L_\theta \rho_\theta \Pi_x]
\right) \nonumber \\
&=  \frac{1}{2} ((\mathrm{Tr}( \rho_\theta\Pi_x L_\theta ])^* +  \mathrm{Tr}[\rho_\theta \Pi_x L_\theta] ) \nonumber \\
&=  \mathrm{Re}\,\mathrm{Tr}[\rho_\theta \Pi_x L_\theta].
\end{align}
where we use $z+z^* = 2 \text{ Re} (z)$ for complex number $z$. Substituting $\partial_\theta p(x|\theta)$ into the  classical Fisher information $F(\theta)$

\begin{align}
   F(\theta) &= \int dx \frac{(\partial_\theta p(x|\theta))^2}{p(x|\theta)} \\
   &= \int dx \frac{
\left(\mathrm{Re}\,\mathrm{Tr}[\rho_\theta \Pi_x L_\theta]\right)^2
}{
\mathrm{Tr}[\rho_\theta \Pi_x]
} \\
&\le
\int dx \frac{
\left|\mathrm{Tr}[\rho_\theta \Pi_x L_\theta]\right|^2
}{
\mathrm{Tr}[\rho_\theta \Pi_x]
}
\end{align}
where we used $(\mathrm{Re}\,z)^2 \le |z|^2$ for the last inequality step. Let us denote $A_x = \frac{\sqrt{\rho_\theta}\sqrt{\Pi_x}}{\sqrt{\mathrm{Tr}[\rho_\theta \Pi_x]}}$ and $B_x = \sqrt{\Pi_x}L_\theta\sqrt{\rho_\theta}$, we have the identity
\begin{equation}
\frac{
|\mathrm{Tr}[\rho_\theta \Pi_x L_\theta]|^2
}{
\mathrm{Tr}[\rho_\theta \Pi_x]
}
=
|\mathrm{Tr}[A_x B_x]|^2.
\end{equation}
where we used $\mathrm{Tr}[\rho_\theta \Pi_x L_\theta]
=
\mathrm{Tr}[\sqrt{\rho_\theta}\sqrt{\Pi_x}\sqrt{\Pi_x}L_\theta\sqrt{\rho_\theta}]$. By the Cauchy--Schwarz inequality for operators,
\begin{equation}
|\mathrm{Tr}[A_x B_x]|^2 \le \mathrm{Tr}[A_x A_x^\dagger] \mathrm{Tr}[B_x^\dagger B_x]. 
\end{equation}
Since $\mathrm{Tr}[A_x A_x^\dagger]=
\frac{\mathrm{Tr}[\sqrt{\rho_\theta}\Pi_x\sqrt{\rho_\theta}]}
{\mathrm{Tr}[\rho_\theta \Pi_x]}
= 1$ and $\mathrm{Tr}[B_x^\dagger B_x]
=
\mathrm{Tr}[\Pi_x L_\theta \rho_\theta L_\theta]$, we have
\begin{equation}
\frac{
|\mathrm{Tr}[\rho_\theta \Pi_x L_\theta]|^2
}{
\mathrm{Tr}[\rho_\theta \Pi_x]
}
\le
\mathrm{Tr}[\Pi_x L_\theta \rho_\theta L_\theta].
\end{equation}
Substituting into the Fisher information:
\begin{align}
F(\theta)
&\le
\int dx \frac{
\left|\mathrm{Tr}[\rho_\theta \Pi_x L_\theta]\right|^2
}{
\mathrm{Tr}[\rho_\theta \Pi_x]}\\
&\le
\int dx\, \mathrm{Tr}[\Pi_x L_\theta \rho_\theta L_\theta] \\
&=
\mathrm{Tr}\left[
\left(\int dx\, \Pi_x\right) L_\theta \rho_\theta L_\theta
\right]\\
&= \mathrm{Tr}[L_\theta \rho_\theta L_\theta] \\
&=
\mathrm{Tr}[\rho_\theta L_\theta^2].
\end{align}
where we use completeness of the POVM, $\int dx\, \Pi_x = \mathbb{I}.$
This above chain of inequality prove that the Fisher information $F(\theta)$ of any quantum measurement is bounded by a quantity $F_Q(\theta) = \mathrm{Tr}[\rho_\theta L_\theta^2]$. We define the Quantum Fisher Information (QFI) as $F_Q(\theta)$:

\begin{equation}
F_Q(\theta) = \mathrm{Tr}[\rho_\theta L_\theta^2].
\end{equation}
Therefore, operationally QFI is the supremum of the classical Fisher information over all possible quantum measurements. That is,
\begin{equation}
F_Q(\theta) = \sup_{\{\Pi_x\}} F(\theta),
\end{equation}
where the supremum is taken over all POVMs.
In terms of eigenvalues and eigenstate of $\rho$, we have
\begin{align}
F_Q(\theta) &= \mathrm{Tr}[\rho_\theta L_\theta^2]\\
&= \sum_k p_k \langle k | L_\theta^2 | k \rangle\\
&= \sum_k \sum_l p_k\langle k | L_\theta | l \rangle
\langle l | L_\theta | k \rangle \\
&=  4 \sum_{k,l: p_k+p_l>0}
p_k \frac{(p_k - p_l)^2}{(p_k + p_l)^2}
\left| \langle k | A | l \rangle \right|^2
\end{align}
where we use $\langle k | A | l \rangle = \langle l | A | k \rangle^*$ and the explicit form of the symmetric logarithmic derivative (SLD) as:
\begin{equation}
L_\theta
=
2i \sum_{k,l: p_k+p_l>0}
\frac{p_k - p_l}{p_k + p_l}
\langle k | A | l \rangle
|k\rangle \langle l|, \qquad \langle k | L_\theta | l \rangle
=
2i \frac{p_k - p_l}{p_k + p_l}
\langle k | A | l \rangle
\end{equation}

To express $F_Q(\theta)$ in a symmetric form, we note that the summand is not invariant under exchange of indices $k \leftrightarrow l$ due to the factor $p_k$. Introducing
\begin{equation}
T_{kl}
=
p_k
\frac{(p_k - p_l)^2}{(p_k + p_l)^2}
\left| \langle k | A | l \rangle \right|^2,
\end{equation}
we observe that
\begin{equation}
T_{kl} + T_{lk}
=
\frac{(p_k - p_l)^2}{p_k + p_l}
\left| \langle k | A | l \rangle \right|^2.
\end{equation}
Using the identity $\sum_{k,l} T_{kl}
=
\frac{1}{2} \sum_{k,l} (T_{kl} + T_{lk}),$ we arrive at the symmetric expression
\begin{equation}
\label{eq:F_Q_explicit}
F_Q(\theta)
=
2 \sum_{k,l: p_k+p_l>0}
\frac{(p_k - p_l)^2}{p_k + p_l}
\left| \langle k | A | l \rangle \right|^2.
\end{equation}

This expression gives the Quantum Fisher Information for a quantum state undergoing unitary parameter encoding in terms of its eigenvalues and the generator $A$.

\subsubsection{Pure-State Limit of the Quantum Fisher Information}

Let us now derive the Quantum Fisher Information (QFI) for a pure state, $\rho = \ket{\psi}\bra{\psi}$,  undergoing unitary parameter encoding. Let $\{\ket{k}\}$ be an orthonormal eigenbasis of $\rho$, with
\begin{equation}
\ket{0} \equiv \ket{\psi}, \qquad p_0 = 1, \qquad p_l = 0 \;\; (l \neq 0).
\end{equation}

We now analyze the contributions to the double sum of Eq. \ref{eq:F_Q_explicit}. If both indices satisfy $k,l \neq 0$, then $p_k+p_l=0$ and such terms are excluded. If $k=l=0$, the numerator vanishes, $\frac{(p_0 - p_0)^2}{p_0 + p_0} = 0$, so this term does not contribute. Therefore, only terms with exactly one index equal to $0$ remain. For $k=0$, $l\neq 0$,
\begin{equation}
\frac{(p_0 - p_l)^2}{p_0 + p_l}
=
\frac{(1-0)^2}{1+0}
=
1,
\end{equation}
and similarly for $k\neq 0$, $l=0$. Hence,
\begin{equation}
F_Q
=
2 \sum_{l\neq 0}
\left| \braket{0 | A | l} \right|^2
+
2 \sum_{k\neq 0}
\left| \braket{k | A | 0} \right|^2 = 4 \sum_{l\neq 0}
\left| \braket{l | A | 0} \right|^2.
\end{equation}
where we used the Hermiticity of $A$ to get $\left| \braket{l | A | 0} \right|^2 = \left| \braket{0 | A | l} \right|^2$. Furthermore, invoking completeness $\sum_l \ket{l}\bra{l} = \mathbb{I}$, we obtain
\begin{align}
\sum_l
\left| \braket{l | A | 0} \right|^2
&=
\bra{0} A^2 \ket{0} = \sum_{l\neq 0}
\left| \braket{l | A | 0} \right|^2 + \left| \braket{0 | A | 0} \right|^2.
\end{align}
Identifying $\ket{0} = \ket{\psi}$, we arrive at
\begin{align}
F_Q
&=4 \sum_{l\neq 0}
\left| \braket{l | A | 0} \right|^2  \\
&=4 \sum_l
\left| \braket{l | A | 0} \right|^2 - \left| \braket{0 | A | 0} \right|^2 \\
&= 
4 \left(
\bra{\psi} A^2 \ket{\psi}
-
\braket{\psi | A | \psi}^2
\right) \\
&=  4 (\Delta A)^2
\end{align}
This shows that, for pure states, the Quantum Fisher Information is directly proportional to the variance of the generator of the parameter encoding.

\section{SLD Optimal Observable for the Displaced Thermal State}
\label{sec:SLD_optimal}

We now derive the symmetric logarithmic derivative (SLD) for the output state is $ \rho_{\rm out}(\theta)=S_{N(\theta)}(\beta(\theta))$
where \(\theta\) denotes the scalar parameter. The state $\rho_{\rm out}$ has following properties,
\begin{equation}
     \Sigma(\theta)
    =
    \left(N(\theta)+\frac12\right)I_2, \quad  \Delta X(\theta)
    =
    \sqrt{2}
    \begin{pmatrix}
        \operatorname{Re}\beta(\theta)\\
        \operatorname{Im}\beta(\theta)
    \end{pmatrix}, \quad \langle \hat{a}\rangle = \beta(\theta), \quad  \Delta \hat X(\theta)
    =
    \hat X-\Delta X(\theta),\quad  \hat X=
    \begin{pmatrix}
        \hat q\\
        \hat p
    \end{pmatrix}.
\label{eq:parameter1}
\end{equation}
where, $ \Sigma(\theta)$ is the covariance matrix which is isotropic, $\Delta X(\theta)$ is the displacement vector, and $ \Delta \hat X(\theta)$ is the centered quadrature operator. 

We node define operator, $\hat{b} = \hat a-\beta(\theta),$, to denote thermal state in the displaced frame such that $ \langle \hat{b} \rangle = 0$, and $\langle \hat{b}^\dagger \hat{b} \rangle = N. $ Also, since $\hat{a} = \frac{\hat{q} + i \hat{p}}{\sqrt{2}}$, and $\beta(\theta) = \operatorname{Re} \beta(\theta)+ \operatorname{Im} \beta(\theta)$, we get:
\begin{equation}
    \hat b
=
\frac{1}{\sqrt2}
\left[
\hat q-\sqrt2\,\mathrm{Re}\,\beta(\alpha)
+i\left(
\hat p-\sqrt2\,\mathrm{Im}\,\beta(\alpha)
\right)
\right], \qquad \hat q-\sqrt2\,\mathrm{Re}\,\beta
=
\frac{\hat b+\hat b^\dagger}{\sqrt2}, \qquad \hat p-\sqrt2\,\mathrm{Im}\,\beta
=
\frac{\hat b-\hat b^\dagger}{i\sqrt2}.
\label{eq:parameter2}
\end{equation}

The SLD $\mathcal L_\theta$ is defined by
\begin{equation}
    \partial_\theta \theta_{\rm out}(\theta)
    =
    \frac12
    \left[
    \theta_{\rm out}(\theta)\mathcal L_\theta
    +
    \mathcal L_\theta\theta_{\rm out}(\theta)
    \right].
\end{equation}
Since the parameter $\theta$ enters both the displacement $\beta(\theta)$
and the thermal occupation $N(\theta)$, the SLD separates into
\begin{equation}
    \mathcal L_\theta
    =
    \mathcal L_{\rm disp}
    +
    \mathcal L_{\rm noise}.
\end{equation}

Taking the expectation value of the square of the SLD gives the QFI, $    F_Q(\theta)
    =
    \operatorname{Tr}
    \left[
    \theta_{\rm out}(\theta)\mathcal L_\theta^2
    \right].$
Using the decomposition $  \mathcal L_\theta
    =
    \mathcal L_{\rm disp}
    +
    \mathcal L_{\rm noise},$
we obtain $    F_Q(\theta)
    =
    \langle \mathcal L_{\rm disp}^2\rangle
    +
    \langle \mathcal L_{\rm noise}^2\rangle
    +
    2\langle \mathcal L_{\rm disp}\mathcal L_{\rm noise}\rangle .$
For a displaced thermal state, the centered quadrature fluctuations are
Gaussian with vanishing odd moments, so the cross term vanishes $ \langle \mathcal L_{\rm disp}\mathcal L_{\rm noise}\rangle=0.$
Therefore,
\begin{equation}
    F_Q(\theta)
    =
    F_{\rm disp}(\theta)+F_{\rm noise}(\theta),
\end{equation}

\subsection{Displacement contribution}

The SLD operator $\mathcal L_{\rm disp}$ for the displacement contribution is
\begin{equation}
\label{eq:displaced}
    \mathcal L_{\rm disp}
    =
    \Delta \dot X^{T}\Sigma^{-1}\Delta \hat X
    =
    \frac{\Delta \dot X^{T}\Delta \hat X}{N+\frac12},
\end{equation}
where $\Delta \dot X=\partial_\theta \Delta X(\theta)$. Substituting the quadrature parameters from Eq. \ref{eq:parameter1} and \ref{eq:parameter2}, we get the expression for SLD $ \mathcal L_{\rm disp}$ for the displaced thermal states as:
\begin{align}
    \mathcal L_{\rm disp}
&=
\frac{\sqrt{2}}{N+\frac12}
\left[
\partial_\theta \mathrm{Re}\,\beta\,
\left(\hat q-\sqrt{2}\,\mathrm{Re}\,\beta\right)
+
\partial_\theta \mathrm{Im}\,\beta\,
\left(\hat p-\sqrt{2}\,\mathrm{Im}\,\beta\right)
\right] \\
&=\frac{1}{N+\frac12}
\left[
\partial_\theta \mathrm{Re}\,\beta\,
\left(\hat b+\hat b^\dagger\right)
-
i\,\partial_\theta \mathrm{Im}\,\beta\,
\left(\hat b-\hat b^\dagger\right)
\right] \notag \\
&= \frac{1}{N+\frac12}
\left[
\left(
\partial_\theta \mathrm{Re}\,\beta
-
i\,\partial_\theta \mathrm{Im}\,\beta
\right)\hat b
+
\left(
\partial_\theta \mathrm{Re}\,\beta
+
i\,\partial_\theta \mathrm{Im}\,\beta
\right)\hat b^\dagger
\right] \notag \\
&
=
\frac{
(\partial_\theta \beta^*)\hat b
+
(\partial_\theta \beta)\hat b^\dagger
}{
N(\theta)+\frac12
}.
\end{align}
where we used $\partial_\theta \beta
=
\partial_\theta \mathrm{Re}\,\beta
+
i\,\partial_\theta \mathrm{Im}\,\beta,$ and $\partial_\theta \beta^*
=
\partial_\theta \mathrm{Re}\,\beta
-
i\,\partial_\theta \mathrm{Im}\,\beta.$ As an example, consider the former toy case,  $\rho(\theta) = S_{a+b\theta}(c\sqrt{\theta})$ such that $N(\theta) = a  + b\theta$ and $\beta(\theta) = c \sqrt{\theta}$. Then $\partial_\theta \beta
=
\partial_\theta \beta^*
=
\frac{c}{2\sqrt{\theta}},$ and SLD becomes:
\begin{equation}
  \mathcal L_{\mathrm{disp}} =
\frac{
c(\hat b+\hat b^\dagger)
}{
2\sqrt{\theta}\left(N(\theta)+\frac12\right)
} =\frac{
c\hat a+ c\hat a^\dagger-2c^2\sqrt{\theta}
}{
\left(2N(\theta)+1\right)\sqrt{\theta}
}   
\end{equation}

To prove that Eq. \ref{eq:displaced} gives the correct QFI term, let us compute:
\begin{align}
    F_{\rm disp} &= \langle \mathcal L_\theta^2 \rangle \nonumber \\
    &=
    \left\langle
    \sum_{ij}
    (\Sigma^{-1}\Delta \dot X)_i
    (\Sigma^{-1}\Delta \dot X)_j
    \Delta \hat X_i \Delta \hat X_j
    \right\rangle \nonumber \\
    &=
    \sum_{ij}
    (\Sigma^{-1}\Delta \dot X)_i
    (\Sigma^{-1}\Delta \dot X)_j
    \langle \Delta \hat X_i \Delta \hat X_j \rangle  \nonumber  \\
    &=
    (\Sigma^{-1}\Delta \dot X)^T
    \Sigma
    (\Sigma^{-1}\Delta \dot X) \nonumber \\
    &=
    \Delta \dot X^T
    \Sigma^{-1}
    \Sigma
    \Sigma^{-1}
    \Delta \dot X \nonumber \\
    &=
    \Delta \dot X^T \Sigma^{-1} \Delta \dot X \nonumber \\
    &=  \frac{4|\partial_\theta \beta(\theta)|^2}{2N(\theta)+1}
\end{align}
where we use $    \langle \Delta \hat X_i \Delta \hat X_j \rangle
    =
    \Sigma_{ij}$ in the fourth step, and this QFI calculation matches the result of QFI from Eq. \ref{eq:third_term}.

\subsection{Noise contribution}
The noise contribution is quadratic in the centered mode occupation. In the displaced frame, let $ \hat b=\hat a-\beta(\theta)$ and $\hat n_b=\hat b^\dagger \hat b .$ Then
\begin{equation}
    \mathcal L_{\rm noise}
    = \frac{\partial_\theta N}{N(N+1)}
    \left(
    (\hat a-\beta)^\dagger(\hat a-\beta)-N
    \right) =
    \frac{\partial_\theta N}{N(N+1)}
    \left(
    \hat n_b-N
    \right)
\end{equation}
where \(N=\langle \hat n_b\rangle\). To prove that this term gives the
correct QFI contribution, we compute
\begin{align}
    F_{\rm noise}(\theta)
   & =
    \left\langle
    \mathcal L_{\rm noise}^2
    \right\rangle  \nonumber \\ 
    &= \left[
    \frac{\partial_\theta N}{N(N+1)}
    \right]^2
    \left\langle
    \left(
    \hat n_b-N
    \right)^2
    \right\rangle  \nonumber \\
    &= \left[
    \frac{\partial_\theta N}{N(N+1)}
    \right]^2\langle \hat n_b^2\rangle-\langle \hat n_b\rangle^2  \nonumber \\
    &= \left[
    \frac{\partial_\theta N}{N(N+1)}
    \right]^2 N(N+1)  \nonumber \\
    &= 
    \frac{[\partial_\theta N(\theta)]^2}
    {N(\theta)[N(\theta)+1]}
\end{align}
Thus, the quadratic number-like part of the SLD reproduces exactly the
thermal-noise contribution to the QFI.

Therefore, the SLD for the displaced thermal state is
\begin{equation}
    \mathcal L_\theta
     = \mathcal L_{\rm disp}
    +
    \mathcal L_{\rm noise} =
    \Delta \dot X^{T}\Sigma^{-1}\Delta \hat X
    +
    \frac{\partial_\theta N}{N(N+1)}
    \left[
    (\hat a-\beta)^\dagger(\hat a-\beta)-N
    \right].
\label{eq:sld_displaced_thermal}
\end{equation}
where the first term is linear in quadrature, while the second is quadratic in field operators, i.e., number-like. Upon taking the expectation value of its square, one obtains the QFI result from Eq. \ref{eq:F_coh}, thus showing its correctness.

\subsection{Optimal POVMs}

The quantum Cramér--Rao bound is saturated by measurements in the eigenbasis of the symmetric logarithmic derivative (SLD). Therefore, the
optimal measurement is given by the spectral decomposition of
\(\mathcal L_\theta\). From Eq.~\eqref{eq:sld_displaced_thermal}, the SLD for the displaced thermal state is
\begin{align}
    \mathcal L_\theta
   & =
    \Delta \dot X^{T}\Sigma^{-1}\Delta \hat X
    +
    \frac{\partial_\rho N}{N(N+1)}
    \left[
    (\hat a-\beta)^\dagger(\hat a-\beta)-N
    \right] \\
    &= 
\frac{
2(\partial_\theta \beta^*)\hat b
+
2(\partial_\theta \beta)\hat b^\dagger
}{
2N(\theta)+1
}+ \frac{\partial_\rho N}{N(N+1)}
    \left[
    (\hat a-\beta)^\dagger(\hat a-\beta)-N
    \right]
\end{align}

Defining the displaced mode $ \hat b=\hat a-\beta(\theta),   \hat n_b=\hat b^\dagger \hat b,$  the SLD can be written 
\begin{equation}
  \mathcal L_\theta
    =
    A^*\hat b + A\hat b^\dagger
    +
    \theta\left(\hat n_b-N\right), \qquad
    A
    =
    \frac{2\partial_\rho\beta(\rho)}{2N(\rho)+1},
    \qquad
    \theta
    =
    \frac{\partial_\rho N(\rho)}{N(\rho)[N(\rho)+1]}.
\end{equation}

To diagonalize \(\mathcal L_\rho\), we complete the square by introducing
a shifted annihilation operator $\hat c
    =
    \hat b + \frac{A}{\theta}
    =
    \hat a-\beta(\rho)+\frac{A}{\theta}.$ Up to an additive constant, the SLD becomes

\begin{align}
    \mathcal L_\theta
    &=
    A^*\hat b + A\hat b^\dagger
    +
    \theta\left(\hat n_b-N\right) \notag \\
    &=
    A^* \left( \hat{c} - \frac A\theta  \right)
    +
    A \left( \hat{c}^\dagger - \frac {A^*}\theta  \right)
    +
    \theta
    \left[
    \left(\hat c^\dagger-\frac{A^*}{\theta}\right)
    \left(\hat c-\frac{A}{\theta}\right)
    -N
    \right] \notag \\
    &=
    A^*\hat c
    -
    \frac{|A|^2}{\theta}
    +
    A\hat c^\dagger
    -
    \frac{|A|^2}{\theta}
    +
    \theta \hat c^\dagger \hat c
    -
    A\hat c^\dagger
    -
    A^*\hat c
    +
    \frac{|A|^2}{\theta}
    -
    \theta N \notag \\
    &=
    \theta \hat c^\dagger \hat c
    -
    \theta N
    -
    \frac{|A|^2}{\theta} \notag \\
    &= \theta\,\hat c^\dagger \hat c
    +
    \text{const}.
\end{align}

Thus, the eigenbasis of \(\mathcal L_\theta\) is the Fock basis of the
shifted mode \(\hat c\). The optimal measurement is therefore photon
counting after an appropriate displacement of the field. The optimal POVM is
\begin{equation}
    \Pi_n
    =
    D\!\left(\beta(\theta)-\frac{A}{\theta}\right)
    |n\rangle\langle n|
    D^\dagger\!\left(\beta(\theta)-\frac{A}{\theta}\right), \qquad   A
    =
    \frac{2\partial_\theta\beta(\theta)}{2N(\theta)+1},
    \qquad
    \theta
    =
    \frac{\partial_\theta N(\theta)}{N(\theta)[N(\theta)+1]}.
\end{equation}

If \(\partial_\theta N=0\), the SLD reduces to a linear quadrature operator,  and the optimal measurement is homodyne detection along the direction
\(\Delta \dot X\). Conversely, if \(\partial_\theta\beta=0\), the SLD is
purely number-like, and the optimal measurement is photon counting in the displaced frame.

As an example, consider a state $  S_{a+b\theta}\!\left(c\sqrt{\theta}\right)$, $\partial_\theta \beta(\theta)=
    \frac{c}{2\sqrt{\theta}},
    \partial_\theta N(\theta)=
    b$, then the effective shifted displacement from Eq. \ref{eq:optimal_POVMs} are
\begin{align}
    \beta(\theta)-\frac{A}{\theta}
    &=
    c\sqrt{\theta}
    -
    \frac{
    c(a+b\theta)(a+b\theta+1)
    }{
    b\sqrt{\theta}\,(2a+2b\theta+1)
    } .
\end{align}

\end{document}